\RequirePackage{fix-cm} 
\documentclass[smallextended]{svjour3} 
\smartqed 

\makeatletter
\let\cl@chapter\undefined
\makeatletter
\usepackage{etoolbox} \AtEndEnvironment{proof}{\phantom{}\qed}

\usepackage[utf8]{inputenc}
\usepackage[T1]{fontenc}
\usepackage{booktabs,array} \usepackage{bbm}
\usepackage{xifthen}
\usepackage{nicefrac}
\usepackage{grffile}
\usepackage{amsmath}
\usepackage{amsfonts}
\usepackage{amssymb}
\usepackage{amsmath}
\usepackage{xspace}
\usepackage[inline]{enumitem}
\usepackage{mathtools}
\usepackage[ruled,resetcount,noend,linesnumbered]{algorithm2e}
\usepackage{tikz}
\usepackage{enumitem}
\usepackage{rotating}
\usepackage{tablefootnote}
\usepackage[flushleft]{threeparttable} \usepackage{multirow}
\usepackage{multicol}
\usepackage{arydshln}
\usepackage{blindtext}
\usepackage{balance}
\usepackage{pifont}
\usepackage{subcaption}
\usepackage{tabularx}
\usepackage{cleveref}
\usepackage{natbib}
\usepackage{amssymb}
\usepackage{mathbbol}
\usepackage{adjustbox}
\usepackage{pgf}
\usepackage{lscape} 
\usepackage{xfp}

\newtheorem{assumption}{Assumption}

\makeatletter 
\g@addto@macro{\@algocf@init}{\SetKwInOut{Parameter}{Parameters}} 
\makeatother
 
\DeclareMathOperator*{\argmin}{argmin}
\newcommand{\Var}{\operatorname{Var}}
\DeclareMathOperator*{\Vol}{Vol}
\DeclareMathOperator*{\supp}{supp}
\newcommand{\EE}{\mathbb{E}}
\DeclareMathOperator*{\Order}{O}

\DeclareMathOperator*{\tand}{and}
\DeclareMathOperator*{\mean}{mean}
\DeclareMathOperator*{\twhen}{when}

\newcommand{\prob}{P}
\newcommand{\unif}{\normalfont\textsc{unif}}
\DeclareMathOperator*{\increment}{\mathrel{+}=}
\DeclareMathOperator*{\prodincr}{\mathrel{*}=}
\DeclareMathOperator*{\append}{\mathrel{\cup}=}
\newcommand{\asc}{\overset{a.s}{=}}
\newcommand{\bdot}{\left(\cdot\right)}
\newcommand{\iid}{i.i.d.\ \xspace}

\newcommand{\ind}[1]{{\bf 1}{\left\{#1\right\}}}
\newcommand{\ig}[1]{\ind{#1 \sim H}}

\newcommand{\EEX}[1]{\EE\left[#1\right]}

\newcommand{\rvrSize}{{\normalfont\textsc{m}}\xspace}
\newcommand{\hrvrSize}[2]{\hat{\rvrSize}_{#1,#2}}
\newcommand{\RR}{\mathbb{R}}
\newcommand{\name}{Ripple\xspace}
\newcommand{\CC}{{\em CC}\xspace}

\newcommand{\MCMC}{{\em MCMC}\xspace}
\newcommand{\Motivo}{{Motivo}\xspace}
\newcommand{\RGPM}{{RGPM}\xspace}
\newcommand{\IMPRG}{{IMPRG}\xspace}

\newcommand{\RSS}{{RSS}\xspace}
\newcommand{\GUISE}{{GUISE}\xspace}

\newcommand{\PSRW}{{PSRW}\xspace}

\newcommand{\RWT}{{\em RWT}\xspace}

\newcommand{\RWTE}{\RWT Estimate\xspace}

\newcommand{\hmu}{\hat{\mu}}
\newcommand{\hmus}{\hmu_{*}}
\newcommand{\hsigma}{\hat{\sigma}}

\newcommand{\optionalprefix}[1]{\ifthenelse{\isempty{#1}}{}{$#1$-}}
\newcommand{\optionalsuperscript}[1]{\ifthenelse{\isempty{#1}}{}{^{(#1)}}}

\newcommand{\HON}[1][]{\optionalprefix{#1}{\em HON}\xspace}
\newcommand{\CIS}[1][]{\optionalprefix{#1}{\em CIS}\xspace}

\newcommand{\cG}[1][]{\mathcal{G}\optionalsuperscript{#1}}

\newcommand{\cS}[1][]{\mathcal{S}\optionalsuperscript{#1}}
\newcommand{\cI}[1][]{\mathcal{I}\optionalsuperscript{#1}}
\newcommand{\cJ}[1][]{\mathcal{J}\optionalsuperscript{#1}}
\newcommand{\cV}[1][]{\mathcal{V}\optionalsuperscript{#1}}

\newcommand{\cC}[1][]{\mathcal{C}\optionalsuperscript{#1}}
\newcommand{\cE}[1][]{\mathcal{E}\optionalsuperscript{#1}}
\newcommand{\bPhi}[1][]{{\boldsymbol \Phi}\optionalsuperscript{#1}}

\newcommand{\hPhi}[1][]{\widehat{\boldsymbol \Phi}\optionalsuperscript{#1}}

\newcommand{\dg}[1][]{{\bf d}\optionalsuperscript{#1}}
\newcommand{\bdg}[1][]{\bar{\bf d}\optionalsuperscript{#1}}
\newcommand{\tdg}[1][]{\widetilde{\bf d}\optionalsuperscript{#1}}
\newcommand{\N}[1][]{{\bf N}\optionalsuperscript{#1}}

\newcommand{\vo}{\zeta}

\newcommand{\cH}{\mathcal{H}}

\newcommand{\compoperator}{\circ}
\newcommand{\comp}[2]{#1\compoperator #2}
\newcommand{\xo}{x_{0}}
\newcommand{\epart}{Ergodicity-Preserving Stratification\xspace}
\newcommand{\EPART}{{\em EPS}\xspace}

\newcommand{\sg}[2]{#1\left(#2\right)} \newcommand{\cA}{\mathcal{A}}
\newcommand{\cB}[2]{\mathcal{B}_{#1,#2}}

\newcommand{\hbU}[2]{\widehat{{\bf U}}_{#1,#2}}
\newcommand{\hdeg}[1]{\widehat{\deg}_{#1}}
\newcommand{\hdg}{\widehat{\dg}}
\newcommand{\dgvo}[1]{\dg(\vo_{#1})}
\newcommand{\hdgvo}[1]{\hdg(\vo_{#1})}
\newcommand{\hpvo}[1]{\widehat{p}_{\bPhi_{#1}}(\vo_{#1},\cdot)}
\newcommand{\pvo}[1]{p_{\bPhi_{#1}}(\vo_{#1},\cdot)}

\newcommand{\hbeta}[2]{\widehat{\beta}_{#1,#2}}

\newcommand{\kdesc}{k-1}

\newcommand{\cT}{\mathcal{T}}
\newcommand{\dT}{\mathcal{T}^{\dagger}}

\newcommand{\dist}{{\normalfont\textsc{dist}}}
\newcommand{\rejec}{C_{\normalfont\textsc{rej}}}
\newcommand{\diam}{D}

\newcommand{\dds}{\ddot{s}}

\newcommand{\Vstar}{V^{*}}

\newcommand{\bX}{{\bf X}}

\newcommand{\bias}{{\normalfont\textsc{bias}}}

\newcommand{\cX}{\mathcal{X}}

\newcommand{\bY}{{\bf Y}}

\setlist[itemize]{leftmargin=*}

\usepackage[misc,geometry]{ifsym} 
 \newcommand*{\affmark}[1][*]{\textsuperscript{#1}}

\begin{document}

\titlerunning{Sequential Stratified Regenerations}        \title{Sequential Stratified Regeneration: \MCMC for Large State Spaces with an Application to Subgraph Count Estimation}

\author{Carlos H. C. Teixeira\protect\affmark[1]\affmark[*]   \and
        Mayank  Kakodkar\protect\affmark[2]\affmark[*] \and 
        Vin\'{i}cius Dias\protect\affmark[1,3] \and
        Wagner Meira Jr.\protect\affmark[1] \and 
        Bruno Ribeiro\protect\affmark[2]
}

\authorrunning{Teixeira et al.} 

\institute{
            \Letter \xspace Carlos H. C. Teixeira \\
            \email{carlos@dcc.ufmg.br}\\
             \protect\affmark[1] Universidade Federal de Minas Gerais, Belo Horizonte, Brazil\\
              \protect\affmark[2]  Purdue University, West Lafayette, USA \\
              \protect\affmark[3]  Universidade Federal de Ouro Preto, Ouro Preto, Brazil\\
              \protect\affmark[*]  Equal contribution.
    }

\maketitle
\begin{abstract}
This work considers the general task of estimating the sum of a bounded function over the edges of a graph, given neighborhood query access and where access to the entire network is prohibitively expensive.
To estimate this sum, prior work proposes Markov chain Monte Carlo (\MCMC) methods that use random walks started at some seed vertex and whose equilibrium distribution is the uniform distribution over all edges, eliminating the need to iterate over all edges. 
Unfortunately, these existing estimators are not scalable to massive real-world graphs.
In this paper, we introduce \name, an \MCMC-based estimator that achieves unprecedented scalability by stratifying the Markov chain state space into ordered strata with a new technique that we denote {\em sequential stratified regenerations}.
We show that the \name estimator is consistent, highly parallelizable, and scales well.

We empirically evaluate our method by applying \name to the task of estimating connected, induced subgraph counts given some input graph.
Therein, we demonstrate that \name is accurate and can estimate counts of up to $12$-node subgraphs, which is a task at a scale that has been considered unreachable, not only by prior \MCMC-based methods but also by other sampling approaches.
For instance, in this target application, we present results in which the Markov chain state space is as large as $10^{43}$, for which \name computes estimates in less than $4$ hours, on average.

\end{abstract} 
\keywords{
	Markov Chain Monte Carlo,
	Random Walk,
	Regenerative Sampling,
	Motif Analysis,
	Subgraph Counting,
	Graph Mining
	}

\section{Introduction}
This work considers the following general task:
Let $\cG = (\cV,\cE)$ be a simple graph, where $\cV$ is the set of vertices, $\cE$ is the set of edges, and $\cE$ contains at most a single edge between any pair of vertices and no self-loops.
Our goal is to efficiently estimate the sum of a bounded function over all the edges of $\cG$,
\begin{equation}
	\label{eq.edge.sum.task}
	\mu(\cE) = \sum_{(u,v)\in\cE} f(u,v)\,
\end{equation}
where $f\colon \cE \to \RR$, $f(\cdot)<B$ is a bounded function for some constant $B \in \RR$ under the following query model from \citet{Avrachenkov2016}.

\begin{assumption}[Query Model]
\label{a.query.model}
Assume we are given arbitrary seed vertices and can query the neighborhood $\N(u) \triangleq \{v \in \cV : (u,v) \in \cE\}$ for any vertex $u \in \cV$ such that accessing the entire graph $\cG$ is prohibitively expensive.
\end{assumption}
This setting arises naturally in the subgraph counting problem, which we study in \Cref{sec.subgraph.counting}.
Simple Monte Carlo procedures are not useful because random vertex and edge queries are not directly available, and reservoir sampling would require iteration over all edges. 
Standard Markov chain Monte Carlo (\MCMC) methods cannot estimate the quantity in \Cref{eq.edge.sum.task} and are limited to estimate $\nicefrac{\mu(\cE)}{|\cE|}$, because $|\cE|$ in our task is unknown~\citep{ribeiro2012estimation}.
Generally, under \Cref{a.query.model}, \Cref{eq.edge.sum.task} is estimated using specialized \MCMC estimators that use a random-walk-like Markov chain that has a uniform distribution over the edges $\cE$ as its equilibrium distribution. 
However, these estimators~\citep{Avrachenkov2016} are impractical in large graphs because their running time is $O(|\cE|)$.

Traditional \MCMC methods are limited by their reliance on the Markov chain on $\cG$ reaching equilibrium or {\em burning in}.
Because the rate of convergence to equilibrium depends on the spectral gap~\citep{aldous-fill-2014}, a significant number of Markov chain steps is needed to {\em burn in} in order to produce accurate estimates of \Cref{eq.edge.sum.task}, particularly in large graphs.
Parallel approaches that divide the state space into disjoint ``chunks'', which are to be processed in parallel~\citep{wilkinson2006parallel,neiswanger2013asymptotically}, offer no respite because we cannot access the entire graph. In fact, $\cG$ may not even have disconnected components (i.e., disjoint chunks) that can be parallelized.
Therefore, traditional \MCMC on $\cG$ offers no meaningful parallelization opportunities and running times may be arbitrarily long.

{\em\bf Contributions.} This work introduces {\em sequential stratified regeneration} (\name), a novel parallel \MCMC technique that expands the application frontier of \MCMC to large state-space graphs $\cG$.
\name stratifies the underlying Markov chain state space into ordered strata that need not be disjoint chunks, rather, they need to be connected.
Markov chain regeneration~\citep{nummelin1978splitting} is then used to compute estimates in each stratum sequentially, using a recursive method, which improves regeneration frequencies and reduces variance. 
\name offers an unprecedented level of efficiency and parallelism for \MCMC sampling on large state-space graphs while retaining the benefits of \MCMC-based algorithms, such as low memory demand (polynomial w.r.t.\ output).

Surprisingly, the parallelism of \name comes from the regeneration rather than the stratification: the strata's job is to keep regeneration times short. 
We demonstrate that the estimates obtained by \name are consistent, among other theoretical guarantees. 
In addition, we empirically show the power of \name in a real-world application by specializing \Cref{eq.edge.sum.task} to subgraph counting in multi-million-node attributed graphs–to the best of our knowledge, a task at a scale that has been thought unreachable by any other \MCMC method. 
Our specific contributions to the subgraph counting problem include 
streaming-based optimizations coupled with a parallel reservoir sampling algorithm, 
novel efficiency improvements to the random walk on the \HON~\citep{Wang:2014} and
a theoretical analysis of scalability in terms of running time and memory w.r.t. the subgraph size, verified empirically on large datasets.

 \section{Background and Prior Work}
\label{sec.preliminaries}
The \MCMC random-walk-like Markov chain over the graph $\cG$ is defined as:
\begin{definition}[Random Walk on $\cG$] 
	\label{def.simplerw}	
	Given a simple graph $\cG = (\cV,\cE)$, a simple random walk is a time-homogenous Markov chain $\bPhi$ with state space $\cV$ and transition probability
	$p_{\bPhi}(u,v) = \nicefrac{1}{\dg(u)}$, when $(u,v) \in \cE$ and $0$ otherwise,
	where $\dg(u) = |\N(u)|$ is the degree of $u$ in $\cG$ and $\N(u) = \{v \colon (u,v) \in \cE\}$ is the neighborhood of $u$. 
\end{definition}
It is easy to check that the above random walk can be sampled under \Cref{a.query.model} and that on a connected graph, this walk samples edges uniformly at random in a steady state (check \Cref{sec.ht.mcmc} for details). 
Our notation is summarized in \Cref{sec.notation}.

\subsection{Regenerations in Discrete Markov Chains}
\label{sec.regen.mc}
The rate of convergence to stationarity of the random walk $\bPhi$ from \Cref{def.simplerw} depends on the spectral gap\footnote{The spectral gap is defined as $\delta = 1-\max\{|\lambda_2|, |\lambda_{|\cV|}|\}$, where $\lambda_i$ denotes the $i$-th eigenvalue of the transition probability matrix of $\bPhi$.}~\citep{aldous-fill-2014}.
As such, practitioners are encouraged to run a single, long sample path, which prevents them from splitting the task among multiple cores.
Usually, because the spectral gap is unknown or loosely bounded, practitioners use various diagnostics to \emph{eyeball} if the chain has mixed~\citep{rosenthal1995minorization}. 
 The variance of an estimate computed from a stationary chain~\citep{ribeiro2012estimation} also depends on the spectral gap.

A solution to the above problems is to {\em split}~\citep{nummelin1978splitting} the Markov chain using regenerations. 
Discrete Markov chains regenerate every time they enter a fixed state, which is referred to as a regeneration point.
This naturally yields the definition of a {\em random walk tour} (\RWT).

\begin{definition}[\RWT over $\bPhi$]
	\label{def:rwt}
	Given a time-homogenous Markov chain $\bPhi$ over finite state space $\cV$ and a fixed point $\xo\in \cV$, an \RWT $\bX =(X_{i})_{i=1}^{\xi}$ is a sequence of states visited by $\bPhi$ between two consecutive visits to $\xo$, that is, $X_{1} = \xo$ and $\xi = \min\{i>1 \colon X_{i+1} = \xo \}$ is the first return time to $\xo$.
\end{definition}
Because of the strong Markov property \citep[Chap-2,Thm-7.1]{bremaud2001markov}, {\RWT}s started at $\xo$ are \iid and can be used to estimate $\mu(\cE)$ from \Cref{eq.edge.sum.task} when $|\cE|$ is unknown~\citep{Avrachenkov2016,avrachenkov2018revisiting,teixeira2018graph,savarese2018monte,cooper2016fast,massoulie2006peer}. 
\begin{lemma}[{\RWTE}]
	\label{lem.rwte}
	Given the graph $\cG = (\cV,\cE)$ and the random walk $\bPhi$ from \Cref{def.simplerw},
	consider $f\colon \cE\to\RR$ bounded by $B$, and $\cT$, a set of $m$ {\RWT}s  started at $\xo \in \cV$ (\Cref{def:rwt}) sampled in a parallel $z$ core environment assuming each core samples an equal number of tours.
	Then,
	\begin{equation}
		\label{eq.rwte}
		\hmus(\cT; f, \cG) 
		= \frac{\dg(\xo)}{2m} 
		\sum_{\bX \in \cT} \sum_{j=1}^{|\bX|} 
		f(X_j,X_{j+1})\,,
	\end{equation}
	is an unbiased and consistent estimator of $\mu(\cE) = \sum_{(u,v)\in\cE} f(u,v)$ if $\cG$ is connected, where each $X_j$ refers to the $j$th state in the \RWT $\bX\in \cT$.

	The expected running time for sampling $m$ tours is $\Order\left(m/z \frac{2|\cE|}{\dg(\xo)}\right)$, and when $\cG$ is non-bipartite, the variance of the estimate is bounded as 
	\begin{equation}
	\label{eq.vb.rwt}
	\Var\left(\hmus(\cT)\right) 
	\leq
	\frac{3 B^2}{m} \frac{{|\cE|}^{2}}{\delta(\bPhi)}
	\,,
	\end{equation}
	where $\delta(\bPhi)$ is the spectral gap as defined in the beginning of this section.
\end{lemma}
The \RWTE can be considered a Las Vegas transformation of \MCMC, which takes random time but yields unbiased estimates of objectives, such as \Cref{eq.edge.sum.task}. The parallelism in the expected running time in \Cref{lem.rwte} is directly due to the independence of {\RWT}s.
Moreover, confidence intervals for the \RWTE can be computed, because $\sqrt{m}\frac{\hmus(\cT)-\mu(\cE)}{\hsigma(\cT)}$ approaches the standard normal distribution for sufficiently large $m$, where $\hsigma(\cT)^{2}$ is the empirical variance of $\hmus(\bX)$, the \RWTE computed using an individual tour $\bX \in \cT$. 

\subsection{Improving the Regeneration Frequency}
From \Cref{lem.rwte}, it is clear that increasing the degree of the regeneration point $\dg(\xo)$ and spectral gap $\delta(\bPhi)$ and decreasing $|\cE|$ reduces the variance as well as the running time of the \RWTE.
\citet{Avrachenkov2016} showed that using the {\em supernode} in a contracted graph as a regeneration point achieves the above reductions.
\begin{definition}[Contracted Graph] \citep{Avrachenkov2016}
    \label{def.contracted}
	Given a graph $\cG = (\cV,\cE)$ from \Cref{def.simplerw} and a set of vertices $I\subset\cV$, a contracted graph is a multigraph $\cG_I$ formed by collapsing $I$ into a single node $\vo_{I}$.
	The vertex set of $\cG_I$ is then given by $\cV\backslash I \cup \{\vo_{I}\}$, and its edge multiset is obtained by conditionally replacing each endpoint of each edge with $\vo_{I}$ if it is a member of $I$ and removing self-loops on $\vo_{I}$. 
	We refer to the set $I$ and the vertex $\vo$ as the supernode.
\end{definition}
Contractions benefit {\RWT}s because the supernode degree $\dg_{\cG_I}(\vo_{I})$ in $\cG_I$ and the spectral gap $\delta(\bPhi_{I})$ of the random walk on the contracted graph increase monotonically with $|I|$~\citep{Avrachenkov2016}. Moreover, {\RWT}s can be sampled on $\cG_{I}$ without explicit construction, as we see next.
\begin{remark}
	\label{rk.contracted}
	Let the multi-set $\N(\vo_{I}) \triangleq \uplus_{u \in I} \N_{\cG}(u) \backslash I$ be the neighborhood of the supernode in $\cG_I$ from \Cref{def.contracted}.
	Let $\bPhi_{I}$ be the simple random walk on $\cG_{I}$.
	An \RWT $(X_{i})_{i=1}^{\xi}$ on $\bPhi_{I}$ from $\vo_{I}$ is sampled by setting $X_{1} = \vo_{I}$, sampling $X_2$ u.a.r.\ from $\N(\vo_{I})$ and subsequently sampling transitions from $\bPhi$ until the chain enters $I$, i.e.,\ $\xi = \min\{i>1 \colon X_{\xi+1} \in I\}$.
\end{remark}

This construction naturally stratifies $\cE$ and decomposes $\mu(\cE)$ as $\mu(\cE_{*})+\mu(\cE\backslash\cE_{*})$, where we can exactly compute the $\mu(\cE_{*})$ and compute an \RWTE of $\mu(\cE\backslash\cE_{*})$ on the contracted graph. 
However, to compute the supernode degree, $\dg_{\cG_I}(\vo_{I})$; furthermore, to sample from $\N(\vo_{I})$, we need to enumerate the set of the edges incident on $I$ in $\cG$ given by $\cE_{*} \subset \cE$.
As such, a massive supernode $I$ (which is crucial when $|\cE|$ is large) makes enumerating $\cE_{*}$ prohibitively expensive. 
We overcome these issues and gain additional control over regenerations by further stratifying $|\cE|$.

\section{Sequential Stratified Regenerations}
\label{sec.estimator}

\name controls regeneration times through a {\em sequential stratification} of the vertices and edges of $\cG$ into ordered strata as illustrated in \Cref{fig.strata}, which allows us to control the regeneration frequency and the \RWTE variance.
For each stratum, we then construct a graph in which the supernode is created by collapsing all prior strata, from which {\RWT}s can be sampled.
We use the {\RWT}s from the previous strata to estimate the degree of and sample transitions from the supernode.
The core idea is described in two steps: \Cref{ripplestep.strat} details the stratification and conditions that it needs to satisfy and \Cref{sec.recursion} describes the recursion.
Finally, we show that the estimator bias converges to zero asymptotically in the number of tours.
Particularly for subgraph counting, we show that {\name}'s time complexity is independent of the (higher-order) graph size ($|\cE|$) and only depends polynomially on the diameter and maximum degree of the {\em input} graph and the subgraph size (\Cref{sec.subgraph.counting}).

\begin{figure}\centering
    \begin{minipage}[t]{0.35\linewidth}
        \centering
        \includegraphics[width=\linewidth]{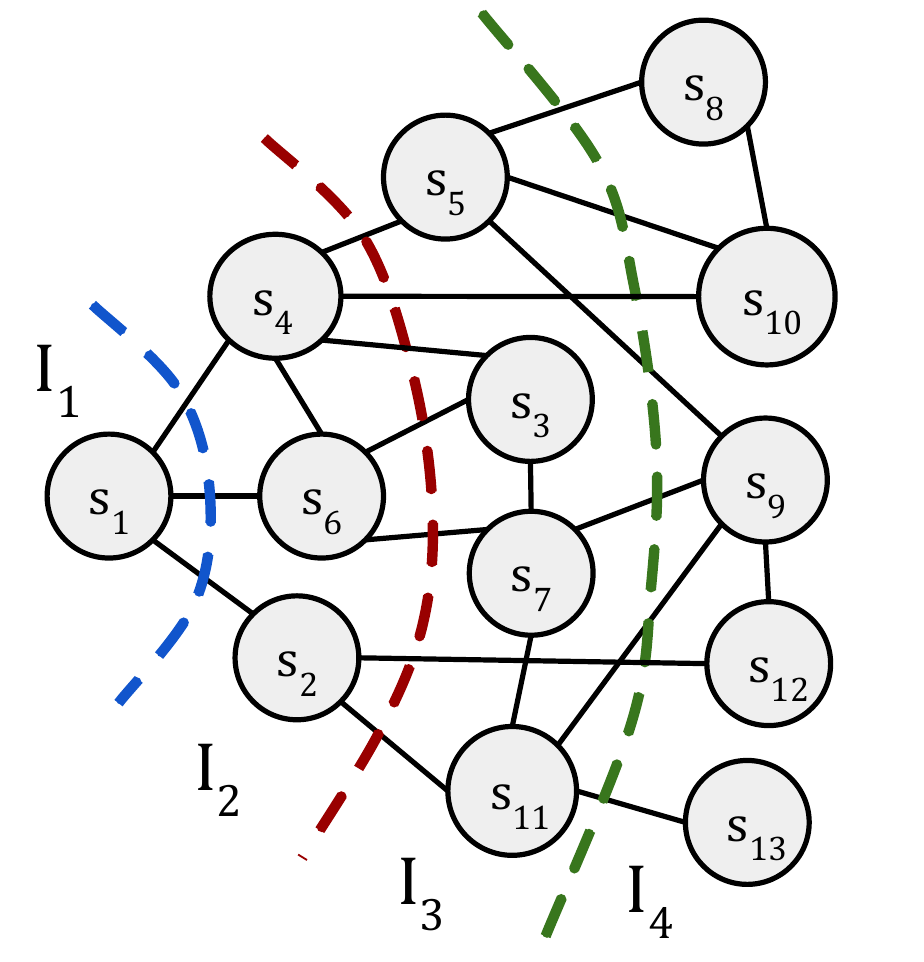}	
\subcaption{Stratification $\cV \equiv \cI_{1:4}$}
        \label{fig.strata.full}
    \end{minipage}
    \hfill
    \begin{minipage}[t]{0.25\linewidth}
        \centering
        \includegraphics[width=\linewidth]{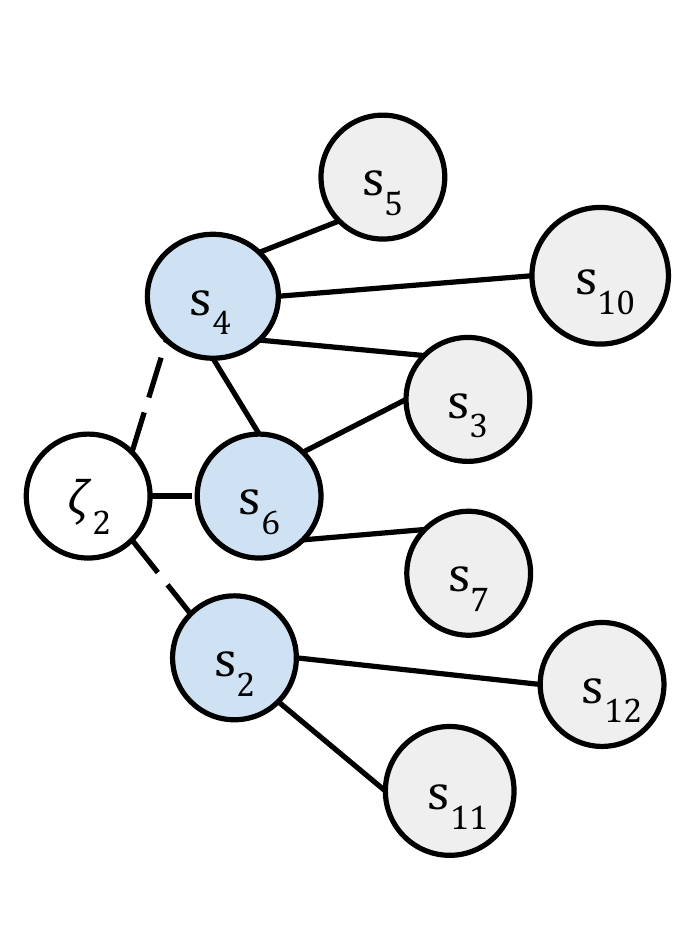}	
        \subcaption{$\cG_{2}$}
        \label{fig.strata.g2}
    \end{minipage}
    \hfill
    \begin{minipage}[t]{0.21\linewidth}
        \centering
        \includegraphics[width=\linewidth]{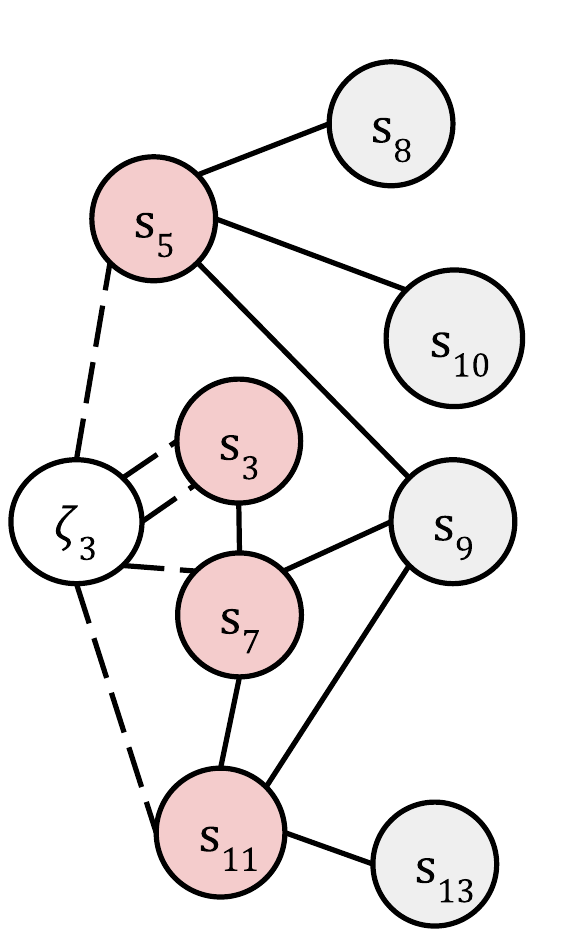}	
        \subcaption{$\cG_{3}$}
        \label{fig.strata.g3}
    \end{minipage}
    \hfill
    \begin{minipage}[t]{0.14\linewidth}
        \centering
        \includegraphics[width=\linewidth]{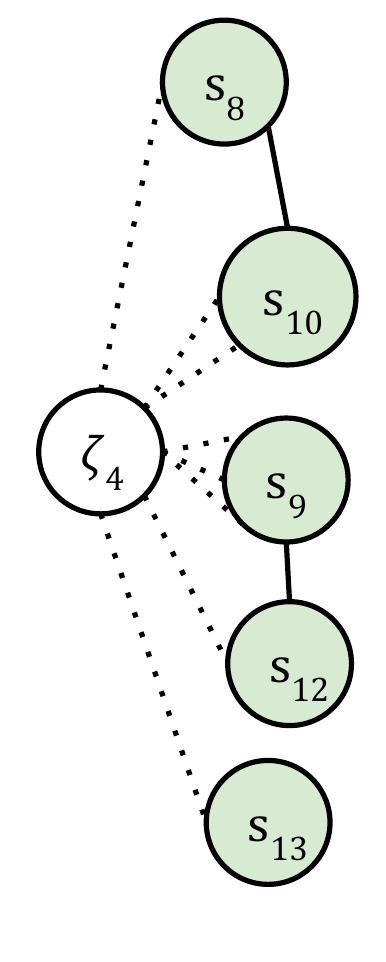}	
        \subcaption{$\cG_{4}$}
        \label{fig.strata.g4}
    \end{minipage}
    \caption{
        \Cref{fig.strata.full} shows a simple graph $\cG$ that is stratified into four strata $\{\cI_1, \cI_2, \cI_3, \cI_4\}$.  
        \Cref{fig.strata.g2,fig.strata.g3,fig.strata.g4} show the second, third and fourth graph strata constructed by \Cref{def.gr}.
        In the multi-graph $\cG_{2}$ (\Cref{fig.strata.g2}), vertices in $\cI_{1}$ are collapsed into $\vo_{2}$ and only edges incident on $\cI_{2}$ are preserved.
        The edge set therefore contains $\cJ_{2}$ and the edges between $\vo_{2}$ and $\cI_{2}$.
        Consequently, self-loops on $\vo_{2}$ and edges between $\cI_{3:4}$ are absent.
        \Cref{fig.strata.g3,fig.strata.g4} follow suit. 
        In each stratum $\cG_{r}$, {\RWT}s from $\vo_{r}$ are started by sampling u.a.r.\ from the dotted edges and estimates are computed over the solid edges.
    }
    \label{fig.strata}
\end{figure}

\subsection{Sequential Stratification}
\label{ripplestep.strat}
Consider the following vertex and edge stratification procedure. 
\begin{definition}[Sequential Stratification]
    \label{def.dependent.strat}
    Given $\cG = (\cV,\cE)$ from \Cref{def.simplerw}, a function $\rho\colon \cV \to \{1, \ldots, R\}$ induces the stratification $(\cI_{r},\cJ_{r})_{r=1}^{R}$ if $s \in \cI_{\rho(s)}$, for each $s\in \cV$, and $(u,v) \in \cJ_{\min\left( \rho(u),\rho(v) \right)}$, for each $(u,v) \in \cE$.
\end{definition}
Note that these strata are pairwise disjoint and their union is the set of vertices and edges of the graph.
Next, we describe the contracted graph over which {\RWT}s are to be sampled in each stratum.
\begin{definition}[$r$-th Graph Stratum]
    \label{def.gr}
    Let $\cA_{i:j} \triangleq \cup_{x=i}^{j} \cA_{x}$ be defined for any ordered tuple of sets.
    Let $(\cI_{r},\cJ_{r})_{r=1}^{R}$ be the stratification induced by $\rho$ from \Cref{def.dependent.strat} on $\cG = (\cV,\cE)$.
    The $r$-th graph stratum $\cG_{r} = (\cV_{r}, \cE_{r})$, $r>1$, is obtained by removing all edges not incident on $\cI_{r}$ and vertices that do not neighbor vertices in $\cI_{r}$ and subsequently contracting $\cI_{1:r-1}$ into $\vo_{r}$ according to \Cref{def.contracted}.
    Further, let $\bPhi_{r}$ denote the simple random walk on $\cG_{r}$.
\end{definition} 
It can be shown that the vertex set $\cV_{r}$ contains the $r$-th stratum $\cI_{r}$, the $r$-th supernode $\vo_{r}$, obtained by collapsing $\cI_{1:r-1}$, and vertices from subsequent strata neighboring $\cI_{r}$, $\cup_{u\in\cI_{r}}\N(u)\cap \cI_{r+1:R}$.
The edge multiset $\cE_{r}$ is the union of $\cJ_{r}$ and edges that connect $\vo_{r}$ to vertices in $\cI_{r}$ resulting from the graph contraction.
A detailed example is shown in \Cref{fig.strata}.  
Note that when $R=2$, \name reduces to the estimator from \citet{Avrachenkov2016}.

\paragraph{\epart.}
Because the \RWTE is consistent only if the underlying graph is connected according to \Cref{lem.rwte}, we have the following definition:
\begin{definition}[\epart (\EPART)]
    \label{def:epart}
    The stratification due to $\rho$ from \Cref{def.dependent.strat} is an \epart if each graph stratum from \Cref{def.gr} is connected, i.e.,\ $\bPhi_{r}$, $r>1$, is irreducible.
\end{definition}
We propose necessary and sufficient conditions on $\rho$ that yield an \EPART.
\begin{proposition}
    \label{prop.epart.suff}
    $\rho$ yields an \EPART if the following three conditions are satisfied:
    \begin{enumerate}[label=(\alph*), ref=\Cref{prop.epart.suff}~(\alph*)]
    \item \label{component.cond} for at least one vertex in each connected component of $\cG$, $\rho$ evaluates to $1$;
    \item \label{interconnect.cond} for each $u \colon \rho(u) = r$, there exists $v \in \N(u)$ such that $\rho(v) \leq r$; and
    \item \label{preconnect.cond} there exists $(u_0,v_0) \in \cE$ such that $\rho(u_0) = r$ and $\rho(v_0) < r$.
    \end{enumerate}
\end{proposition}
Although the optimal stratification would depend on $\cG$ and the quantity being estimated, an ideal stratification would yield graph strata wherein the supernode degree and connectivity are maximized (\Cref{lem.rwte}) while minimizing the number of strata (because of the bias propagation described in \Cref{thm.bias.propagation}).
$\rho$ needs to be efficient as well because we will see that it is evaluated at each step of the random walk and the \name estimators from \Cref{def.recursive.step,def.estimator} heavily depend on it.
In \Cref{prop.ripple.complexity} we show that return times to the supernode are inversely proportional to the fraction of vertices in $\cI_{r}$ connected to $\cI_{1:r-1}$.

\subsection{Recursive Regenerations}
\label{sec.recursion}
Assume for the moment that in each stratum, $r = 2, \ldots, R$, we know the degree of the supernode $\dg(\vo_{r})$ and can sample directly from $p_{\bPhi_{r}}(\vo_{r}, \cdot)$, which is the transition probability out of $\vo_{r}$ in the graph stratum $\cG_{r}$.
We could then sample {\RWT}s $\cT_{r}$ and compute stratumwise {\RWTE}s, which when combined as $\mu(\cJ_{1}) + \sum_{r=2}^{R} \hmu_r(\cT_{r})$ provide an unbiased estimate of $\mu(\cE)$ as a direct consequence of \Cref{lem.rwte} and the linearity of expectations.
Unfortunately, the impracticality of this assumption, especially under \Cref{a.query.model} (when $R>2$), necessitates the following relaxation.

\begin{definition}[Supernode Estimates, $\hdgvo{r}$ and $\hpvo{r}$]
   \label{def.sn.estimates}
   Given an \EPART of $\cG$ (\Cref{def:epart}), the supernode estimates in the $r$-th graph stratum $\cG_{r}$ consist of the estimate of the degree $\hdgvo{r}$ and a sample from some approximate transition probability out of the supernode $\hpvo{r}$.
   Let $\hPhi_{r}$ be the random walk on $\cG_{r}$, where transitions are sampled according to $\bPhi_{r}$ everywhere except $\vo_{r}$, where they are sampled from $\hpvo{r}$.
\end{definition}
Although $\hPhi_{r}$ may not be reversible, {\RWT}s on $\hPhi_{r}$ retain pairwise independence and the benefits stated after \Cref{lem.rwte}. 
We leverage this fact in the following recursive solution that 
computes {\em supernode estimates} in the current stratum using supernode estimates and tours sampled in the previous strata.

\begin{definition}[\name's Recurrence Relation]
    \label{def.recursive.step}
    Given a graph $\cG$ stratified according to $\rho$ (\Cref{def:epart}) and some stratum $r$, $1 < r \leq R$, assume access to the result of previous recursive steps, i.e.,\ the set of $m_{q}$ {\RWT}s ($\dT_{q}$), supernode degree estimates $\hdgvo{q}$ and estimated transition probabilities out of the supernode $\hpvo{q}$ (\Cref{def.sn.estimates}) for all $2\leq q<r$. The estimate of the number of edges between $\cI_{q}$ and $\cI_{r}$ is given by
    \begin{equation}\label{eq:prwt.deg}
        \hbeta{q}{r} = 
		\frac{\hdgvo{q}}{|\dT_{q}|} \sum_{\bX \in \dT_{q}}\sum_{j=2}^{|\bX|} \ind{\rho(X_{j}) = r} 
		\,,
    \end{equation}
    where $X_{j}$ is the $j$-th state visited in tour $\bX$, and 
    by convention, $\hbeta{1}{r} = |\cE\cap \cI_{1} \times \cI_{r}|$ is exactly computed.
    The $r$-th supernode degree is then estimated as 
    \begin{equation}
        \label{eq:deg.full}
        \hdgvo{r} = \sum_{q=1}^{r-1}\hbeta{q}{r} \,.
    \end{equation}
    Transitions from $\hpvo{r}$ are sampled by sampling $q \in \{1, \ldots, r-1\}$ with probability $\hbeta{q}{r}$ and then sampling u.a.r.\ from $\hbU{q}{r}$, which is defined as
    \begin{equation}
        \label{eq.prwt.stitch.tour}
		\hbU{q}{r} = 
		\uplus_{\bX \in \dT_{q}}\uplus_{j=2}^{|\bX|}
		\left\{X_{j}  \colon \rho(X_{j}) = r \right\}
		\,,\, \twhen q>1\,,
	\end{equation}
	and as $\uplus_{u \in \cI_{1}} \N(u) \cap \cI_{r}$ by convention when $q=1$, where $\uplus$ is the multi-set union.
	$\hbU{q}{r}$, $q > 1$, is thus the multi-set of all states in $\cI_{r}$ visited by {\RWT}s on $\hPhi_{q}$.
	An \RWT so started stops when it reaches some state $X'$, where $\rho(X')=r$.
\end{definition}
\Cref{prop.rwt.phi.r} (\Cref{proofs.sec.estimator}) contains additional details for sampling {\RWT}s on $\bPhi_{r}$.
The above recursion therefore allows us to estimate supernode degrees and sample {\RWT}s to compute an estimate of $\mu(\cE)$ from \Cref{eq.edge.sum.task} as follows:
\begin{definition}[\name's $\mu$ Estimator]
    \label{def.estimator}
    Given the supernode degree estimates $\hdgvo{r}$ and {\RWT}s $\dT_{r}$ sampled in each graph stratum from \Cref{def.recursive.step} and the edge strata $\cJ_{r}$, $2\leq r \leq R$ based on an \EPART of $\cG$ from \Cref{def:epart}, the \name estimate is defined as
    \begin{align}\label{eq:prwt.full}
        \hmu_{\name}
        =&
        \mu(\cJ_{1}) +
        \sum_{r=2}^{R}\hmu\left(\dT_{2:r}; f \right) \,,\\
        \label{eq:prwt.partial}
        \text{where},\,
        \hmu\left(\dT_{2:r}; f \right)
        =&
        \frac{\hdgvo{r}}{2 |\dT_{r}|}\sum_{\bX \in \dT_{r}}\sum_{j=2}^{|\bX|-1} f(X_{j},X_{j+1})  \,,
    \end{align}	
    and $X_{j}$ is the $j$th state visited by the \RWT $\bX \in \dT_{r}$. 
    The dependence of $\dT_{r}$ and $\hdgvo{r}$ on $\dT_{2:r-1}$ is suppressed for brevity.
\end{definition}
This estimate of $\mu(\cE)$ is unbiased when the number of tours is infinite.
\begin{theorem}
    \label{thm.estimator}
    The \name estimate from \Cref{def.estimator} is a consistent estimator of $\mu(\cE)$ (asymptotically unbiased in the number of tours), that is,
    \begin{equation*}
        \lim_{|\dT_{2}| \to\infty} \ldots \lim_{|\dT_{R}| \to\infty} 
        \mu(\cJ_{1}) +
        \sum_{r=2}^{R}\hmu\left(\dT_{2:r}; f \right)
        \asc
        \mu(\cE)\,.
    \end{equation*}
\end{theorem}
In the finite regime, however, there exists a bias in each stratum that depends on the estimation bias in the previous strata, which we quantify as follows:
\begin{theorem}
    \label{thm.bias.propagation}
    Given the random walk $\bPhi_{r}$ on the \EPART-stratum $\cG_{r}$ from \Cref{def.gr,def:epart}, the estimates of the degree and transition probability at the supernode $\hdgvo{r}$ and $\hpvo{r}$ from \Cref{def.sn.estimates}, and assuming aperiodic $\bPhi_{r}$, the bias of the \name estimate in the $r$th stratum  from \Cref{eq:prwt.partial} is given by 
\begin{equation*}
        \left| \EEX{\hmu\left(\dT_{2:r}; f \right) \big| \dT_{2:r-1} } - \mu(\cJ_{r})\right|
    \leq 
    \left( \lambda_{r}\nu_{r} + |1-\lambda_{r}| \right)
    \frac{\sqrt{3} B|\cE_{r}|}{\sqrt{\delta_{r}}}\,,
    \end{equation*}
    where $\delta_{r}$ is the spectral gap of $\bPhi_{r}$, $B$ is the upper bound of $f$, $\nu_{r} = \|\hpvo{r} - \pvo{r}\|_2$ is the $L^{2}$ distance between transition probabilities out of $\vo_{r}$ \citep{aldous-fill-2014}(\Cref{def.l2.dist}) and $\lambda_{r} = \nicefrac{\hdgvo{r}}{\dgvo{r}}$.
\end{theorem}
Therefore, the bias in each stratum affects the bias in subsequent strata.
Consequently, we control the empirical variance in each stratum by increasing the number of tours sampled (we detail this for subgraph counting in \Cref{sec.subgraph.counting}).

\section{Applying \name to Count Subgraphs}
\label{sec.subgraph.counting}
We now focus on a concrete implementation of \name to count subgraphs on a given simple input graph $G=(V,E,L)$ with vertices $V$, edges $E$, and attribute function $L$, which is assumed to be finite and undirected.
In general, a subgraph induced by any $V'\subset V$ on $G$ is given by $\sg{G}{V'} = (V', E\cap (V'\times V'), L)$.
However, in this work, we are interested in subgraphs $\sg{G}{V'}$ that are connected and where $|V'| = k$, referred to as a connected, induced subgraph (\CIS) of size $k$ or \CIS[k].
As such, the task is defined as
\begin{definition}[Subgraph Count]
	\label{def.problem}
	Let $\cV[k]$ be the set of all {\CIS[k]}s of graph $G$, let $\sim$ denote the graph isomorphism equivalence relation (or any equivalence relation), and let $\cH$ be an arbitrary set of pairwise nonequivalent {\CIS[k]}s.
	The subgraph count is defined as the $|\cH|$-dimensional vector $\cC[k] = (\cC[k]_H)_{H \in \cH}$, where $\cC[k]_H = \sum_{s \in \cV[k]} \ig{s}$, and $\ind{\cdot}$ is the indicator function.
\end{definition}
Therefore, $\cC[k]$ contains the count of subgraphs in $\cV[k]$ equivalent to each subgraph in $\cH$. 
We suppress the dependence of $\cC[k]$ on $\cH$ for simplicity.

Subgraph counting is challenging when $k>3$ in real-world input graphs because $\cV[k]$ is not tractably enumerable and naively sampling $k$ vertices to obtain {\CIS}s is challenging because $\nicefrac{|\cV[k]|}{|V|^{k}} \to 0$ (as evidenced by \Cref{tab:datasets}).
Next, we address this issue by reducing the subgraph counting problem to an edge sum (\Cref{eq.edge.sum.task}) over a higher-order graph that only provides neighborhood query access for large-real-world input graphs.
We also propose a stratification strategy compatible with the access model and introduce novel solutions to improve speed and memory requirements.
We defer the straightforward aspects to \Cref{sec.name.summary.alg}, wherein we summarize the entire algorithm (\Cref{alg:summary}).

\subsection{\MCMC on the Subgraph Space}\label{sec.mcmc.subgraph}
\citet{Wang:2014} proposed a network over subgraphs called the \HON, which exposes neighborhood query access from \Cref{a.query.model} and is therefore amenable to \MCMC solutions (which we optimize in \Cref{alg:r-sampling}).
\begin{definition}[Higher-Order Network ({\HON[k]}) \citep{Wang:2014}] \label{def:hon}	
The higher-order network or \HON $\cG[k] = (\cV[k], \cE[k])$ is a graph whose vertices are the set of all \CIS[k] contained in the input graph $G$, and $(u,v)$ form an edge in $\cE[k]$ if they share all but $k-1$ vertices, that is, $|V(u) \cap V(v)| = k-1$.
\end{definition}
In the \HON[k-1], the subgraph induced by an edge $(u,v) \in \cE[k-1]$, i.e., $\sg{G}{V(u) \cup V(v)}$, is a \CIS[k].
Thus, the subgraph counts from \Cref{def.problem} can then be expressed as an edge sum over $\cG \equiv \cG[k-1]$ as
\begin{equation}
	\label{eq.mu.edge.sum}
	\cC[k] = \mu(\cE[k-1]) =  \sum_{(u,v) \in \cE[k-1]} \left(\frac{\ig{\sg{G}{V(u) \cup V(v)}}}{\gamma(u,v)}\right)_{H \in \cH} \,,
\end{equation}
where $\gamma(u,v) = |\{(\ddot{u},\ddot{v}) \in \cE \colon 
V(u) \cup V(v) \equiv V(\ddot{u}) \cup V(\ddot{v})\}|$ is the number of edges that represent the same subgraph as $(u,v)$.
The set of edges sampled by a random walk on $\cG[k-1]$ is called the {\em pairwise subgraph random walk} (\PSRW). 
Having reduced the subgraph counting task to \Cref{eq.edge.sum.task}, we proceed with implementing \name.

\subsection{\epart for Subgraph Counting}
\label{eps.subgraph}
Toward using \name, we propose an \epart of $\cG$ via the stratification function $\rho$.
\begin{proposition}[\EPART for subgraphs]
	\label{prop.subgraph.epart}
	Consider the set of $n_1$ {\em seed} subgraphs $\cI_{1}$ whose vertex sets in $G$ are pairwise non-intersecting. Let $V(\cI_{1}) \triangleq \cup_{\dds\in\cI_{1}} V(\dds)$ be the set of all vertices in $G$ forming subgraphs in $\cI_{1}$.
	Let $\dist(u)$ be the shortest path distance from $u \in V$ to any vertex in $V(\cI_{1})$.
Define $\rho$ as 
	\begin{equation*}
		\rho(s) =  1  + \sum_{u\in V(s)} \left( \dist(u) + \ind{u \in V(\cI_{1}) \backslash \Vstar)}\right)\,,
	\end{equation*}
	where $\Vstar$ is the largest connected subset of $V(s)$ such that $\Vstar \subseteq V(\dds)$ for some seed vertex $\dds \in \cI_{1}$ with ties broken arbitrarily.
	If $\cI_{1}$ contains a subgraph from each connected component of $G$, the stratification from \Cref{def.dependent.strat} generated using $\rho$ is an \epart (\Cref{def:epart}).
\end{proposition}
$\dist$ can be precomputed for all $u\in V$ using a single BFS in $\Order(|V| + |E|)$, and $\rho$ can be computed in $\Order(k)$.
Although $R$ is unknown a priori, it is upper bounded as $(k-1) \cdot \diam_{G}$, where $\diam_{G}$ is the diameter of $G$ and the \name estimator simply ignores {\em empty} strata, i.e., strata in which the estimated degree of the supernode $\dgvo{r} = 0$.
To control bias, we aim to reduce $\max_{u\in V} \dist(u)$ by recruiting seed subgraphs in $\cI_{1}$, which are far apart in $G$.

\subsection{Miscellaneous Optimizations}
\label{sec.implementation}
\paragraph{Controlling Memory through Streaming.}
In each pair of strata $r<t$, \Cref{def.recursive.step} uses tours $\dT_{r}$ to compute $\hmu(\dT_{2:r};f)$, $\hbeta{r}{t}$ and $\hbU{r}{t}$, which are, respectively, the estimates of $\mu(\cJ_{r})$ and the size of and sample from the set of vertices in $\cI_{t}$ connected to $\cI_{r}$.
Although $\hmu(\dT_{2:r};f)$ and $\hbeta{r}{t}$ can be computed as running sums, storing $\hbU{r}{t}$ requires memory on the order of the sum of all tour lengths, which is random. 
Our solution is to use \emph{Algorithm R}~\citep{vitter1985random}, to sample a fixed-size ($\rvrSize$) sample without replacement from all the tours in $\dT_{r}$ (See \Cref{sec.res.sampl}).
We note that although the hyperparameter $\rvrSize$ controls memory, it may introduce bias when the number of tours $|\dT_{r}| > \rvrSize$ due to (possible) oversampling, which we observe in \Cref{fig:acc-rvr-5} (\Cref{app:additional-results}).

\paragraph{Speeding up Subgraph Random Walks.}
To sample a random walk in the \HON, naively sampling u.a.r.\ from the neighborhood of a \CIS[k-1] requires $\Order(k^{4} \Delta_{G})$ operations, where $\Delta_{G}$ is the maximum degree in the input graph (see \Cref{sec.nei.alt.def}).
In \Cref{alg:r-sampling}, we propose a rejection sampling algorithm that does so efficiently using {\em articulation points}~\citep{hopcroft1973algorithm}. 

\begin{algorithm}\caption{Efficient Neighborhood Sampling in $\cG[\kdesc]$} \label{alg:r-sampling}
	\KwIn{\CIS[\kdesc] $s$, Graph $G$}
	\KwOut{$x\sim\unif(\N_{\cG[\kdesc]}(s))$}
	\SetKw{KwAnd}{and}
	\SetKw{KwOr}{or}
	\SetKw{KwContinue}{continue}
	\SetKw{KwInl}{Input:}
	\SetKw{KwOutl}{Output:}
	
	Let $\deg_s = \sum_{u \in V(s)} \dg(u)$ and $\cA_{s}$ be the \emph{articulating points} of $s$\label{alg:line.ap}

	\While {True}{
		Sample $u$ from  $V(s)$ w.p.\ $\propto \deg_s - \dg(u)$
		\tcp*{$u$ is the vertex to remove} \label{alg:line:remove}
		
		Sample $a$ from $V(s) \backslash \{u\}$ w.p.\ $\propto \dg(a)$ \label{alg:line:anchor}
		
		Sample $v \sim \unif(\N(a))$
		\tcp*{$v$ is the vertex to add} \label{alg:line:add}
		
$\bias = |N(v) \cap V(s) \backslash \{u\}|$\tcp*{$v$'s sampling bias} \label{alg:line:bias2}
		\If{ $\unif(0,1) \leq \nicefrac{1}{\bias}$ } { \label{alg:line:rej1}
			$x = \sg{G}{V(s)\cup \{v\} \backslash \{u\}}$ \label{alg:line:x}

			\If { $u \neq v$ \KwAnd ($u \notin \cA_{s}$ \KwOr $x$ is connected)} { \label{alg:line:rej2}
				\Return{$x$}
				\tcp*{Connectivity Check}
			}
}
}
\end{algorithm}

\begin{proposition}
	\label{prop:r-sampling}
	Given a subgraph $s \in \cV[\kdesc]$, \Cref{alg:r-sampling} samples u.a.r.\ from $\N_{\cG[\kdesc]}(s)$ in $\Order(k^2 \frac{\Delta_s + k |\cA_{s}|}{k-|\cA_{s}|} )$ expected time, where $\Delta_s \triangleq \max_{u \in V(s)} \dg_{G}(u)$ is the maximum degree of vertices in $s$, and $\cA_{s}$ contains articulation points of $s$.
\end{proposition}
Therefore, the running time of \Cref{alg:r-sampling} is $\in \Order(k\Delta_s + k^2)$ when $s$ is dense ($|\cA_{s}| \approx 0$) and increases to $\Order(k^{2}\Delta_s + k^4)$ for sparse subgraphs, which is faster than the naive algorithm.

\paragraph{From Error Bounds to Tour Counts.}
\name auto-decides the number of {\RWT}s required in each stratum based on an approximate error bound $\epsilon$ provided as input such that the number of tours $\to\infty$ as $\epsilon \to 0$, and the \name estimate converges to the ground truth (\Cref{thm.estimator}).
Specifically, {\RWT}s are sampled until we satisfy 
\begin{equation}
	\label{eq.auto.tours}
	\nicefrac{\hsigma(\dT_{r}; f_{1})}{\sqrt{|\dT_{r}|}} \leq \epsilon \, \hmu(\dT_{r}; f_{1}) \,,
\end{equation}
where $\hmu(\dT_{r}; f_{1})$ is the \name estimate from \Cref{eq:prwt.partial} of the number of edges in the $r$-th graph stratum $\cG_{r}$ (i.e., $f_{1}\bdot = 1$), and $\hsigma^{2}(\dT_{r}; f_{1}) = \widehat{\Var}_{\bX\sim\dT_{r}}(\hmu(\bX; f_{1}))$ is the former's empirical variance over tours.

\paragraph{Performance Guarantees.}
Ignoring the complexity of loading the input graph into memory, we show that for subgraph counting, the memory and time requirements of \name are a polynomial in $k$.
In \Cref{sec.name.summary.alg}, we state and prove a detailed version in which the complexity also depends polynomially on the diameter and maximum degree of $G$ and is invariant to $|V|$ and $|E|$.
\begin{proposition}
	\label{prop.ripple.complexity.short} 
	Assuming a constant $m$ {\RWT}s sampled per stratum and ignoring graph loading,
	the \name estimator for \CIS[k] counts detailed in \Cref{sec.name.summary.alg}-\Cref{alg:summary} has total memory and time complexity in $\widehat{\Order}(k^3 + |\cH|)$ and $\widehat{\Order}(k^7 + |\cH|)$, respectively, when 
	all factors other than $k$ and $|\cH|$ are ignored.
\end{proposition} 
More details for subgraph counting with \name are provided in \Cref{sec.name.summary.alg}.

\section{Experiments and Results}
\label{sec:results}
We now evaluate the \name estimator for $k$-node subgraph (\CIS[k]) counts on large-real-world networks. 
We show that \name outperforms the state-of-the-art method in terms of time and space and that \name converges to the ground truth for various pattern sizes as hyperparameters are varied. 
Additional experiments that evaluate the parallelism, etc.,\ are deferred to \Cref{app:additional-results}.
Our code is available at \texttt{https://github.com/dccspeed/ripple}.

\begin{itemize}
    
\item {\it Execution environment.}
Our experiments were performed on a dual Intel Xeon Gold 6254 CPU with 72 virtual cores (total) at 3.10 GHz and 392 GB of RAM. In addition, this machine is equipped with a fast SSD NVMe PCIex4 with 800 GB of free space available.

\item {\it Baselines.}
We use Motivo~\citep{Bressan:2019}, a fast and parallel C++ system for subgraph counting, as the baseline because it is the only method capable of counting large patterns (k>6), to the best of our knowledge.
Additionally, notice that existing \MCMC methods for subgraph counting, such as \IMPRG~\citep{chen2018mining} and \RGPM~\citep{teixeira2018graph}, cannot count beyond $k=5$ in practice.

\item {\it Datasets.} 
We use large networks from SNAP~\citep{snapnets}, representing diverse domains, which have been used to evaluate many subgraph counting algorithms~\citep{Bressan:2018, Bressan:2019}. \Cref{tab:datasets} presents the basic features of these datasets, including the order of magnitude of the \name estimates of the subgraph counts $|\cV[k]|$, $k =6, 8, 10, 12$. 

\item {\it Hyper-parameters $\cI_1$, $\rvrSize$ and $\epsilon$.} 
Finally, we evaluate the trade-off between accuracy and resource consumption by varying the aforementioned hyperparameters, detailed in \Cref{eps.subgraph,sec.implementation}. ($\rvrSize$ is evaluated in \Cref{app:additional-results}.)
\end{itemize}

\newcommand{\ra}[1]{\renewcommand{\arraystretch}{#1}}
\begin{table*}
\centering
\scalebox{0.9}{
 \begin{tabular}{lrrrr|cccc}
\bottomrule
 \multirow{2}{*}{\textbf{Graph}}  & \multirow{2}{*}{$\mathbf{|V|}$} &
 \multirow{2}{*}{$\mathbf{|E|}$} & \multirow{2}{*}{$\mathbf{\diam_{G}}$} & \multirow{2}{*}{$\mathbf{\Delta_{G}}$} &
 \multicolumn{4}{c}{\textbf{Magnitude of Est. \# of {\CIS}s}}\\
    &  &  & & & $|\cV[6]|$ & $|\cV[8]|$ & $|\cV[10]|$ & $|\cV[12]|$\\
\bottomrule 
Amazon   & $334,863$ & $925,872$ & 
44 & 549 & 
$10^{11}$ & $10^{15}$ & $10^{19}$ & $10^{22}$\\
DBLP   & $317,080$ &  $1,049,866$ & 
21 & 343 & 
$10^{12}$ &  $10^{16}$ & $10^{19}$ & $10^{23}$ \\
Cit-Pat.   & $3,774,768$ &  $16,518,948$  &
22 & 793 &
$10^{14}$ & $10^{18}$ & $10^{22}$ & $10^{26}$ \\
Pokec  &  $1,632,803$ & $30,622,564$ &  
11 & 14,854 & 
$10^{18}$ & $10^{25}$  & $10^{32}$ & $10^{38}$  \\
LiveJ.  & $3,997,962$  & $34,681,189$	& 
17 & 14,815 & 
$10^{19}$ &  $10^{25}$ & $10^{32}$ & $10^{38}$\\
Orkut  & $3,072,441$ & $117,185,083$ & 
9 & 33,313 & 
$10^{21}$ & $10^{28}$ & $10^{35}$ & $10^{43}$ \\
\bottomrule
 \end{tabular}
}

\caption{The graphs that we used along with their diameter $\diam_{G}$, maximum degree $\Delta_G$ and the estimated orders of magnitude of {\CIS[k]} counts,
$|\cV[k]|$.
}
 \label{tab:datasets}
\end{table*}

\subsection{Scalability Assessment}
\label{sec.scalability}
We start by assessing the scalability of the methods when estimating 
 \CIS[k] counts for $k\geq6$. 
To the best of our knowledge, Motivo is the only existing method capable of estimating these patterns. Motivo has two phases: a build-up phase, which constructs an index table in the disk, and a sampling phase that queries this table. 
We only measure the time taken by the build-up phase and the out-of-core (disk) usage because
this is a bottleneck for Motivo.
As such, we report the best-case scenario for Motivo, and the reported values are lower bounds for the actual time and space requirement.
For \name, we report the total time and the RAM usage as the space cost because 
our method works purely in memory.
Both methods were executed using all threads available.

In \Cref{tab:sys.comp.time,tab:sys.comp.mem}, we compare the running time and space usage of \name and Motivo.
We also report their rate of increase in terms of the subgraph size $k$ in columns $\nicefrac{\text{Time}^{(k)}}{\text{Time}^{(k-2)}}$ and $\nicefrac{\text{Space}^{(k)}}{\text{Space}^{(k-2)}}$.
We fix $\epsilon=0.003$, $|\cI_1|=10^{4}$ and $\rvrSize=10^7$ based on the analysis in \Cref{sec:conv-analysis} and \Cref{app:additional-results}.
For Motivo, we follow the authors' suggestions. 
In \Cref{app:additional-results}-\Cref{tab:disp-ripple}, we report the dispersion $\frac{\max-\min}{\mean}$ of the \name estimates generated in the measured runs to ensure that the results are not arbitrary.

\begin{table}\centering
	\scalebox{0.9}{
\begin{tabular}{llrrrrr}
\toprule

           &   &      \multicolumn{2}{c}{\textbf{\Motivo} Build-up only} & \multicolumn{2}{c}{\textbf{\name} ($\epsilon=0.003$)} & 
           \multirow{2}{*}{\bf \shortstack[c]{\\ \name \\gain (hrs)}}
           \\ 
           \cmidrule(lr){3-4}  
           \cmidrule(lr){5-6}
\textbf{Dataset} & \textbf{k} & \multicolumn{1}{c}{\bf Time (hrs)} &
\multicolumn{1}{c}{$\frac{\text{Time}^{(k)}}{\text{Time}^{(k-2)}}$} &  \multicolumn{1}{c}{\bf Time (hrs)}  & 
\multicolumn{1}{c}{$\frac{\text{Time}^{(k)}}{\text{Time}^{(k-2)}}$} &
\\
\toprule
\multirow{4}{*}{Amazon} 
          & 6  
          & $\mathbf{0.002\pm0.000}$& $-$ 
          & $0.020\pm0.000$ & $-$  
& \fpeval{0.002-0.020}
          \\
          & 8 
          & $\mathbf{0.006\pm0.000}$ & \fpeval{round(0.006/0.002,1)}$\times$
          & $0.029\pm0.000$ &  \fpeval{round(0.029/0.020,1)}$\times$
& \fpeval{0.006-0.029}
          \\
          & 10
          & $0.082\pm0.000$ & \fpeval{round(0.082/0.006,1)}$\times$
          & $\mathbf{0.056\pm 0.000}$ & \fpeval{round(0.056/0.029,1)}$\times$ 
& $+$\fpeval{0.082-0.056}
          \\
          & 12 
          & $3.630\pm 0.002$ &  \fpeval{round(3.630/0.082,1)}$\times$
          & $\mathbf{0.095\pm0.002}$ & \fpeval{round(0.095/0.056,1)}$\times$
& $+$\fpeval{3.630-0.095}
          \\
\cline{1-7}
\hline\multirow{4}{*}{DBLP} 
          & 6  
          & $\mathbf{0.002\pm0.000}$& $-$ 
          & $0.013\pm0.000$ & $-$ 
& \fpeval{0.002-0.013}
          \\
          & 8 
          & $\mathbf{0.007\pm0.000}$& \fpeval{round(0.007/0.002,1)}$\times$
          & $0.030\pm0.000$ & \fpeval{round(0.030/0.013,1)}$\times$
& \fpeval{0.007-0.030}
          \\
          & 10
           & $0.156\pm0.000$& \fpeval{round(0.156/0.007,1)}$\times$ 
           & $\mathbf{0.082\pm0.000}$ &  \fpeval{round(0.082/0.030,1)}$\times$
& $+$\fpeval{0.156-0.082}
          \\
          & 12 
          & $9.099\pm0.002$ & \fpeval{round(9.099/0.156,1)}$\times$  
          & $\mathbf{0.105\pm0.002}$ & \fpeval{round(0.105/0.082,1)}$\times$
& $+$\fpeval{9.099-0.105}
          \\
\cline{1-7}
\hline\multirow{4}{*}{Patents}
        & $6$  
        & $\mathbf{0.022\pm0.000}$ & $-$ 
        & $0.033\pm 0.000$ & $-$ 
& \fpeval{0.022-0.033}
        \\
        & 8  
        & $0.098\pm0.000$ & \fpeval{round(0.098/0.022,1)}$\times$
        & $\mathbf{0.051\pm 0.000}$  & \fpeval{round(0.051/0.033,1)}$\times$
& $+$\fpeval{0.098-0.051}
        \\
        & 10 
        &   {\em $>1.1$ hrs, crashed} & $-$ &  $\mathbf{0.090\pm0.001}$ & \fpeval{round(0.09/0.051,1)}$\times$
& $-$ \\
        & 12 
        &  {\em $>0.5$ hrs, crashed} & $-$
        & $\mathbf{0.117\pm0.003}$ & \fpeval{round(0.117/0.09,1)}$\times$
& $-$ \\
\cline{1-7}
\hline\multirow{4}{*}{Pokec}
        & $6$  
        & $\mathbf{0.012\pm0.000}$ & $-$ 
        & $0.459\pm 0.142$ & $-$ 
& \fpeval{0.012-0.459}
        \\
        & 8  
        & $\mathbf{0.128\pm0.000}$ & \fpeval{round(0.128/0.012,1)}$\times$
        & $0.759\pm 0.282$  & \fpeval{round(0.759/0.459,1)}$\times$
& \fpeval{0.128-0.759}
        \\
        & 10 
        &  $5.965\pm0.000$ & \fpeval{round(5.965/0.128,1)}$\times$
        &  $\mathbf{1.400\pm0.592}$ & \fpeval{round(1.400/0.759,1)}$\times$
& $+$\fpeval{5.965-1.4}
        \\
        & 12 
        &  {\em $>1.5$ hrs, crashed} & $-$ & $\mathbf{1.469\pm0.334}$ & \fpeval{round(1.469/1.400,1)}$\times$
& $-$ \\
\cline{1-7}

\hline\multirow{4}{*}{LiveJ.} 
        & 6  
        & $\mathbf{0.024\pm0.000}$ & $-$ 
        & $0.351\pm0.009$ & $-$ 
& \fpeval{0.0240-0.351}
        \\
        & 8 
        & $\mathbf{0.205\pm0.000}$  & \fpeval{round(0.205/0.024,1)}$\times$ 
        & $0.642\pm0.074$& \fpeval{round(0.642/0.351,1)}$\times$  
& \fpeval{0.205-0.642}
        \\
        & 10  
        & {\em $>2.3$ hrs, crashed} & $-$ & $\mathbf{1.76\pm 1.550}$ & \fpeval{round(1.76/0.642,1)}$\times$ 
& $-$ \\
        & 12 
        & {\em $>0.7$ hrs, crashed} &  $-$ 
        & $\mathbf{2.189\pm1.350}$ & \fpeval{round(2.189/1.76,1)}$\times$ 
& $-$ \\
         
\cline{1-7}
\hline\multirow{4}{*}{Orkut} 
        & 6  
        &  $\mathbf{0.032\pm0.000}$ & $-$ 
        & $0.669\pm0.026$ & $-$ 
& \fpeval{0.032-0.669}
        \\
        &  8  
        & $\mathbf{0.585\pm0.006}$  & \fpeval{round(0.585/0.032,1)}$\times$ 
        & $1.744\pm0.983$ & \fpeval{round(1.744/0.669,1)}$\times$ 
& \fpeval{0.585-1.744}
        \\
        & 10  
        &  {\em $>8.9$ hrs, crashed} & $-$ &  $\mathbf{2.633\pm1.065}$ & \fpeval{round(2.633/1.744,1)}$\times$ 
& $-$ \\
        & 12 
        & {\em $>1.8$ hrs, crashed} & $-$ 
        &  $\mathbf{3.967\pm3.162}$ & \fpeval{round(3.967/2.633,1)}$\times$ 
& $-$ \\
\bottomrule
\end{tabular}
}   \caption{
Running time comparison between \name and \Motivo. The last column shows that for large $k$, \name provides gains of up to 9 hours when Motivo can run to completion.  
\Motivo crashes for large $k$ on large graphs.
  }
	\label{tab:sys.comp.time}
\end{table}

\paragraph{Running time Scalability (\Cref{tab:sys.comp.time}).}
Although \Motivo outperforms \name for $k = 6,8$, it does not scale well for $k = 10,12$, where the execution terminates because of insufficient storage space.
Particularly, for DBLP, Motivo required approximately 10 minutes to process \CIS[10] but almost 9 hours for \CIS[12], a growth rate of $58\times$.
On the other hand, \name not only succeeded in {\bf all} configurations in less than 4 hours on average but also exhibited a smoother growth in running time, with the largest increase ratio being $2.7\times$, observed for DBLP and LiveJournal when $k$ went from $8$ to $10$. 
Furthermore, $\nicefrac{\text{Time}^{(k)}}{\text{Time}^{(k-2)}} < (\nicefrac{\text{k}}{\text{(k-2)}})^7$ in all cases according to \Cref{prop.ripple.complexity}.

\begin{table}\centering
	\scalebox{0.9}{
\begin{tabular}{llrrrrr}
\toprule
           &   &      \multicolumn{2}{c}{\textbf{\Motivo} Build-up only} & \multicolumn{2}{c}{\textbf{\name} ($\epsilon=0.003$)} &
           \multirow{2}{*}{\bf \shortstack[c]{\\ \name \\gain (GB)}}
           \\ 
           \cmidrule(lr){3-4}  
           \cmidrule(lr){5-6}
\textbf{Dataset} & \textbf{k} & 
\multicolumn{1}{c}{\bf Space (GB)} &
\multicolumn{1}{c}{$\frac{\text{Space}^{(k)}}{\text{Space}^{(k-2)}}$} &  \multicolumn{1}{c}{\bf Space (GB)} &
\multicolumn{1}{c}{$\frac{\text{Space}^{(k)}}{\text{Space}^{(k-2)}}$} &
\\
\toprule
\multirow{4}{*}{Amazon} 
          & 6  
          & $\mathbf{0.53\pm0.00}$& $-$ 
          & $4.69\pm0.06$ & $-$
& \fpeval{0.53-4.69}
          \\
          & 8 
          & $\mathbf{4.00\pm0.00}$ & \fpeval{round(4/0.53,1)}$\times$
          &  $5.73\pm0.12$ &   \fpeval{round(5.73/4.69,1)}$\times$
& \fpeval{4-5.73}
          \\
          & 10
          & $48.00\pm0.00$ & \fpeval{round(48/4,1)}$\times$
          &  $\mathbf{7.38\pm0.36}$  & \fpeval{round(7.38/5.73,1)}$\times$
&$+$\fpeval{48-7.38}
          \\
          & 12 
          & $559\pm 0.00$& \fpeval{round(559/48,1)}$\times$
          & $\mathbf{9.09\pm1.02}$ & \fpeval{round(9.09/7.38,1)}$\times$
&$+$\fpeval{559-9.09}
          \\
\cline{1-7}

\hline\multirow{4}{*}{DBLP} 
          & 6  
          & $\mathbf{0.50\pm0.00}$& $-$ 
          & $4.58\pm 0.02$ & $-$ 
& \fpeval{0.50-4.58}
          \\
          & 8 
          & $\mathbf{4.00\pm0.00}$& \fpeval{round(4/0.5,1)}$\times$ 
          & $6.31\pm0.00$ & \fpeval{round(6.31/4.58,1)}$\times$
& \fpeval{4.00-6.31}
          \\
          & 10
          & $50.00\pm0.00$& \fpeval{round(50/4,1)}$\times$
          & $\mathbf{7.99\pm 0.01}$ & \fpeval{round(7.99/6.31,1)}$\times$
& $+$\fpeval{50.00-7.99}
          \\
          & 12 
          & $611.00\pm0.00$& \fpeval{round(611/50,1)}$\times$
          & $\mathbf{10.45\pm0.02}$ & \fpeval{round(10.45/7.99,1)}$\times$
& $+$\fpeval{611.00-10.45}
          \\
\cline{1-7}

\hline\multirow{4}{*}{Patents}
        & 6  
        & $\mathbf{7.00\pm0.00}$ & $-$ 
        & $11.50\pm0.05$  & $-$ 
& \fpeval{7.00-11.50}
        \\
        & 8  
        & $66.00\pm0.00$ & \fpeval{round(66/7,1 )}$\times$ 
        & $\mathbf{13.80\pm0.03}$ & \fpeval{round(13.80/11.5,1)}$\times$
& $+$\fpeval{66.00-13.80}
        \\
        & 10 
        & {\em $>800$, crashed} &  $-$ & $\mathbf{15.85\pm0.08}$ & \fpeval{round(15.85/13.8,1)}$\times$
& $>800$ \\
        & 12 
        &  {\em $>800$, crashed} & $-$
        & $\mathbf{18.12\pm 0.10}$ & \fpeval{round(18.12/15.85,1)}$\times$
& $>800$ \\
\cline{1-7}
        
\hline\multirow{4}{*}{Pokec}
        & $6$  
        & $\mathbf{3.7\pm0.00}$ & $-$ 
        & $13.69\pm 0.06$ & $-$ 
& \fpeval{3.7-13.69}
        \\
        & 8  
        & ${36.00\pm0.00}$ & \fpeval{round(36.00/3.7,1)}$\times$
        & $\mathbf{17.17\pm 0.03}$  & \fpeval{round(17.17/13.69,1)}$\times$
& \fpeval{36-17.17}
        \\
        & 10 
        &  $407.00\pm0.00$ & \fpeval{round(407.00/36.00,1)}$\times$
        &  $\mathbf{20.31\pm0.01}$ & \fpeval{round(20.31/17.17,1)}$\times$
& $+$\fpeval{407-20.31}
        \\
        & 12 
        &  {\em $>800$, crashed} & $-$ & $\mathbf{22.82\pm0.03}$ & \fpeval{round(22.82/20.31,1)}$\times$
& $>800$ \\
\cline{1-7}

\hline\multirow{4}{*}{LiveJ.} 
        & 6
        & $\mathbf{7.70\pm0.00}$ & $-$ 
        & $18.26\pm 0.02$  & $-$ 
& \fpeval{7.7-18.26}
        \\
        & 8 
        & $73.00\pm0.00$ & \fpeval{round(73/7.7,1)}$\times$ 
        & $\mathbf{21.26\pm0.00}$ & \fpeval{round(21.26/18.26,1)}$\times$ 
& $+$\fpeval{73-21.26}
        \\
        & 10  
        &  {\em $>800$, crashed} & $-$ & $\mathbf{24.43\pm0.72}$ & \fpeval{round(24.43/21.26,1)}$\times$ 
& $>800$ \\
        & 12 
        & {\em $>800$, crashed} & $-$
        & $\mathbf{27.75\pm0.00}$ & \fpeval{round(27.75/24.43,1)}$\times$ 
& $>800$ \\
         
\cline{1-7}
\hline\multirow{4}{*}{Orkut} 
        & 6
        & $\mathbf{7.90\pm0.00}$ & $-$ 
        & $40.38\pm 0.00$ & $-$ 
& \fpeval{7.9-40.38}
        \\
        &  8  
        & $78.00\pm0.000$ & \fpeval{round(78/7.9,1)}$\times$ 
        & $\mathbf{43.49\pm 0.00}$ &  \fpeval{round(43.49/40.38,1)}$\times$ 
& $+$\fpeval{78-43.49}
        \\
        & 10  
        & {\em $>800$, crashed} & $-$ & $\mathbf{46.63\pm 0.00}$ & \fpeval{round(46.63/43.49,1)}$\times$ 
& $>800$ \\
        & 12 
        & {\em $>800$, crashed} & $-$ & 
        $\mathbf{49.73\pm 0.00}$ & \fpeval{round(49.73/46.63,1)}$\times$ 
& $>800$ \\
\bottomrule
\end{tabular}
}   \caption{Space usage comparison between \name and \Motivo.\Motivo runs out of {\em disk} space for larger datasets in which $k\geq10$, while \name scales almost linearly. 
  \name saves up to 600 GB of space when \Motivo can run.
  }
	\label{tab:sys.comp.mem}
\end{table}
  
\paragraph{Space Scalability (\Cref{tab:sys.comp.mem}).}
The trends in space usage mirror those of the running time, where we see an almost exponential increase w.r.t.\ $k$ for \Motivo compared to a near constant increase for \name despite its polynomial complexity (\Cref{prop.ripple.complexity}).
For example, in Amazon, {\Motivo}'s space demand increases by $7.5\times$ when $k$ goes from $6$ to $8$ and increases to $12\times$ from $10$ to $12$.
{\name}'s largest rate of increase is $1.4\times$ when $k$ goes from $6$ to $8$ for DBLP, and it saves up to 600 GB of space when Motivo does not crash.

\subsection{Accuracy and Convergence Assessment}
\label{sec:conv-analysis}
Next, we evaluate the accuracy and convergence of \name on small and large subgraph patterns, where the former refers to subgraph sizes in which the number of isomorphic subgraphs can be exactly computed using {\em ESCAPE}~\citep{pinar2017escape}, i.e.,\ $k\leq 5$. 

\begin{figure}\centering
\begin{minipage}[t]{0.29\linewidth}
	\centering
	\includegraphics[width=\linewidth]{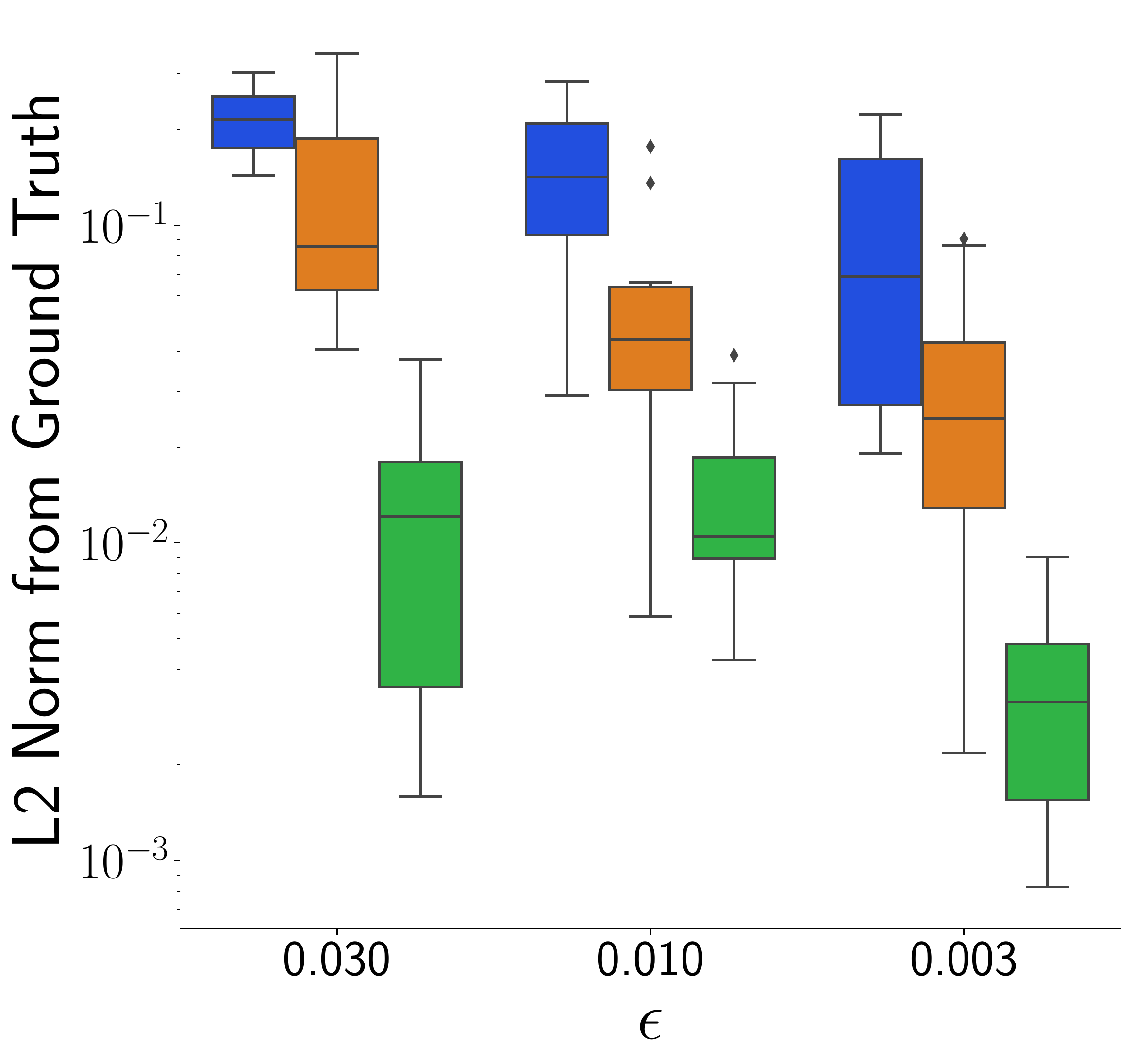}	
	\subcaption{Amazon}
\end{minipage}
\qquad
\begin{minipage}[t]{0.29\linewidth}
	\centering
	\includegraphics[width=\linewidth]{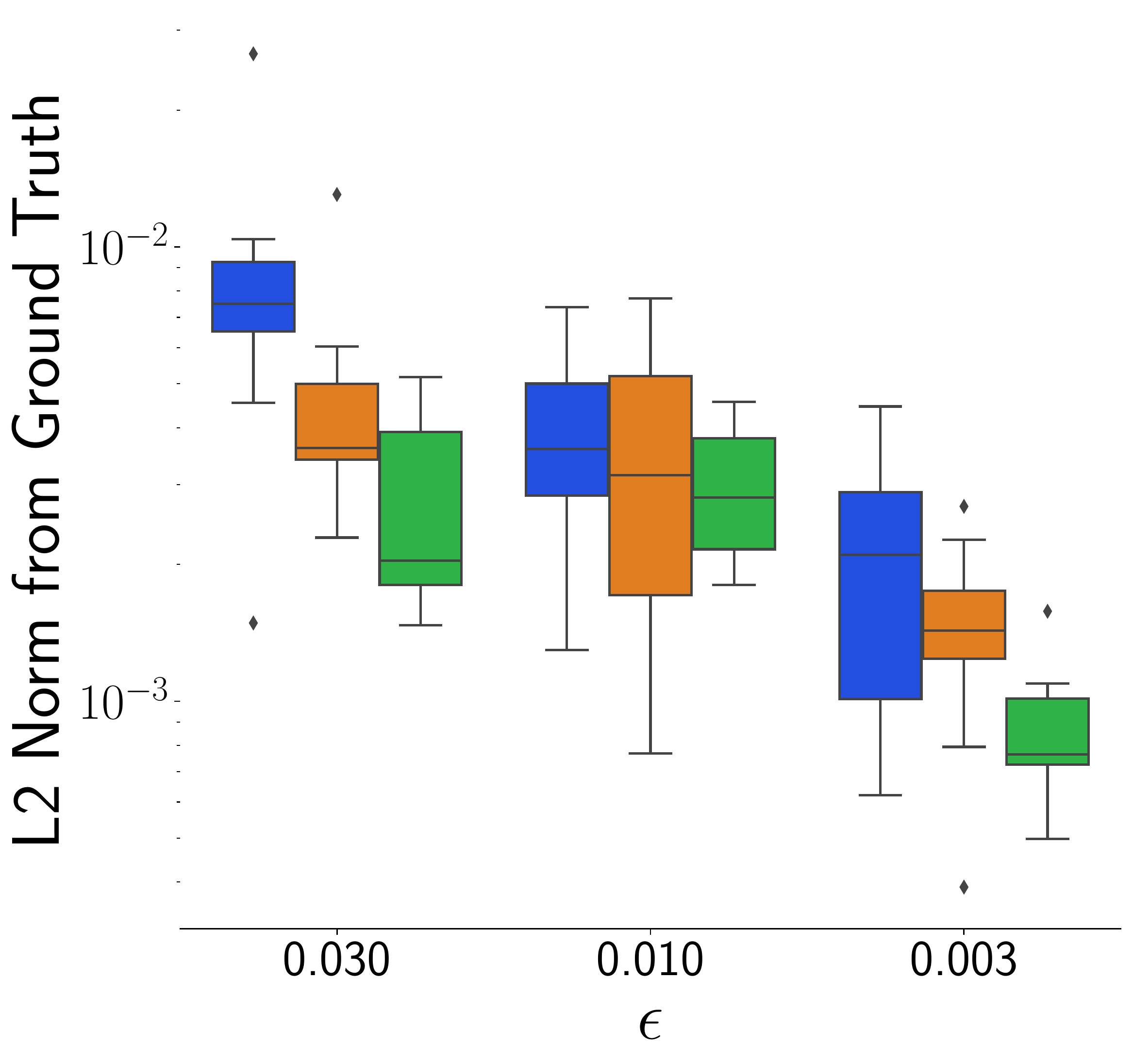}
	\subcaption{DBLP}
\end{minipage}
\qquad
\begin{minipage}[t]{0.29\linewidth}
	\centering
	\includegraphics[width=\linewidth]{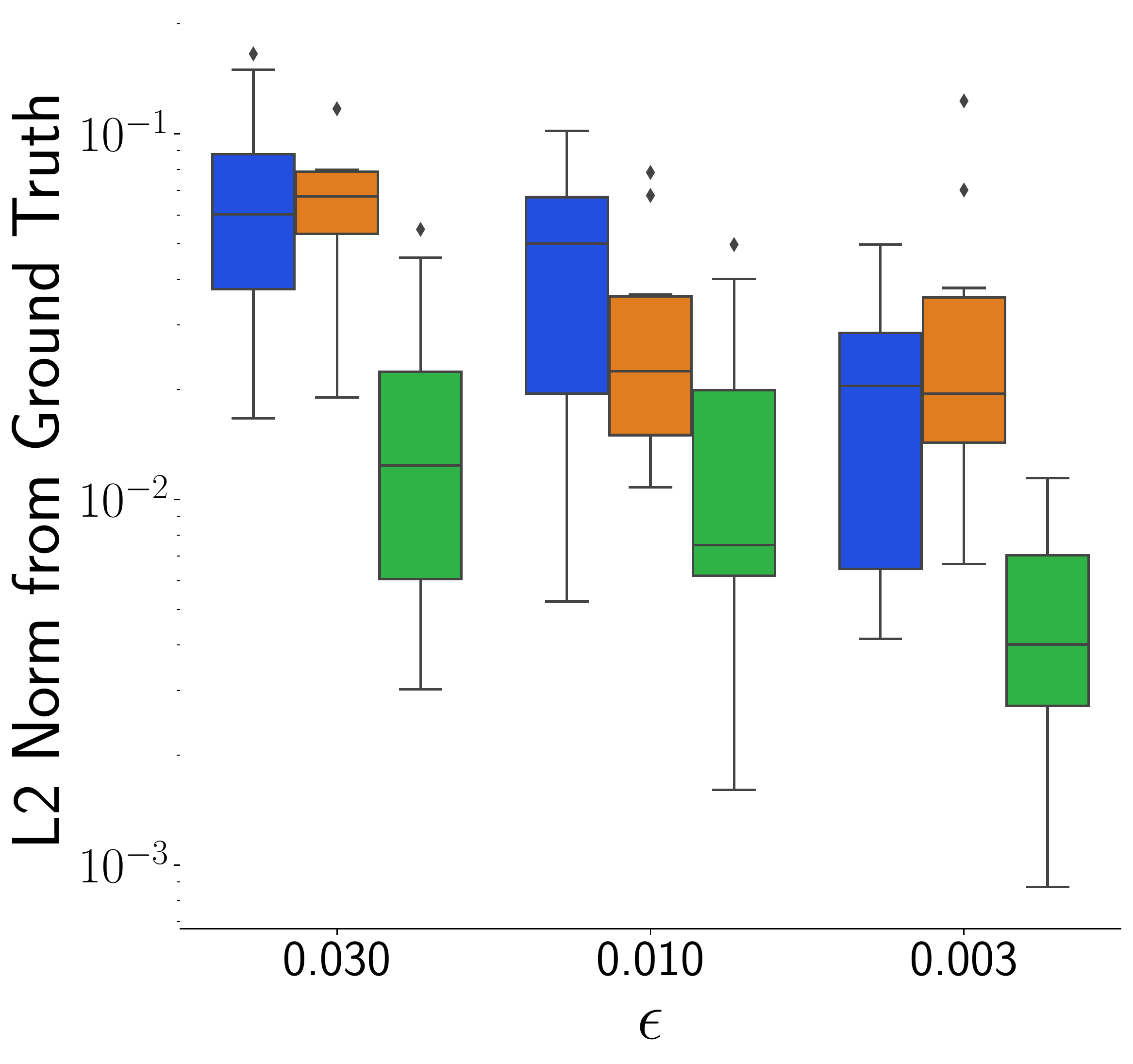}	
	\subcaption{Patents}
\end{minipage}
\qquad
\begin{minipage}[t]{0.29\linewidth}
	\centering
	\includegraphics[width=\linewidth]{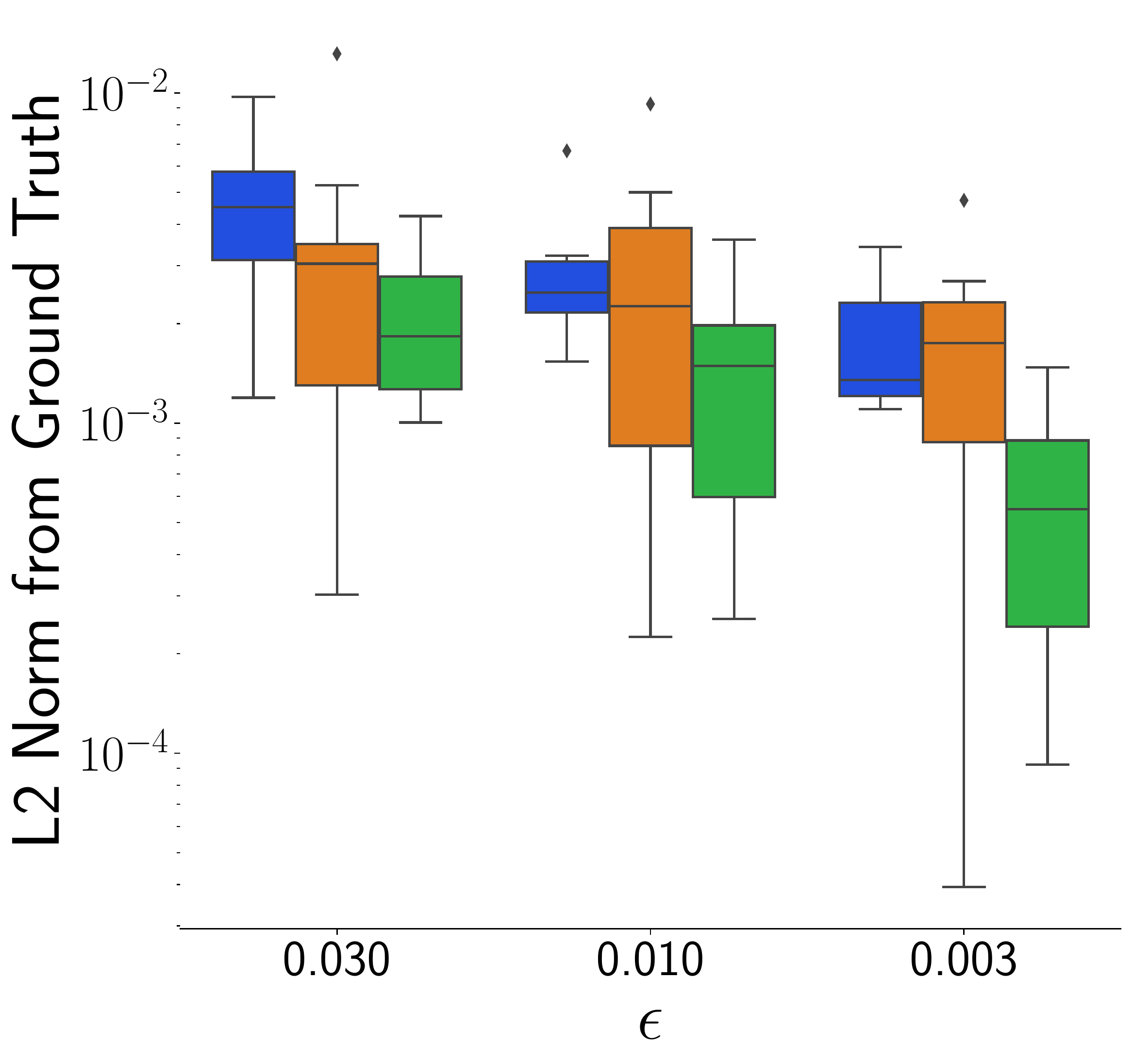}	
    \subcaption{Pokec}
\end{minipage}
\qquad
\begin{minipage}[t]{0.29\linewidth}
	\centering
	\includegraphics[width=\linewidth]{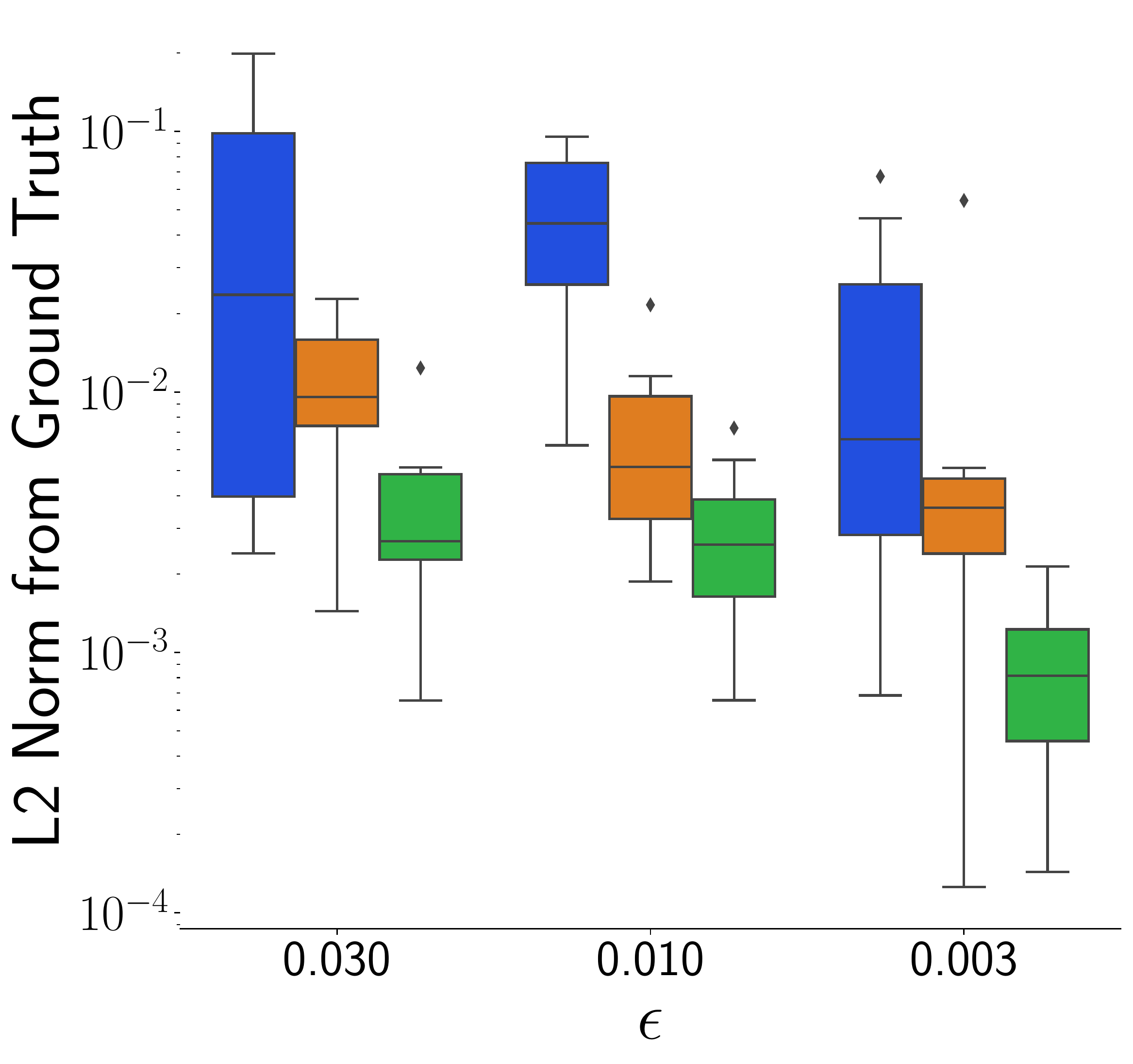}
	\subcaption{Live Journal}
\end{minipage}
\qquad
\begin{minipage}[t]{0.29\linewidth}
	\centering
	\includegraphics[width=\linewidth]{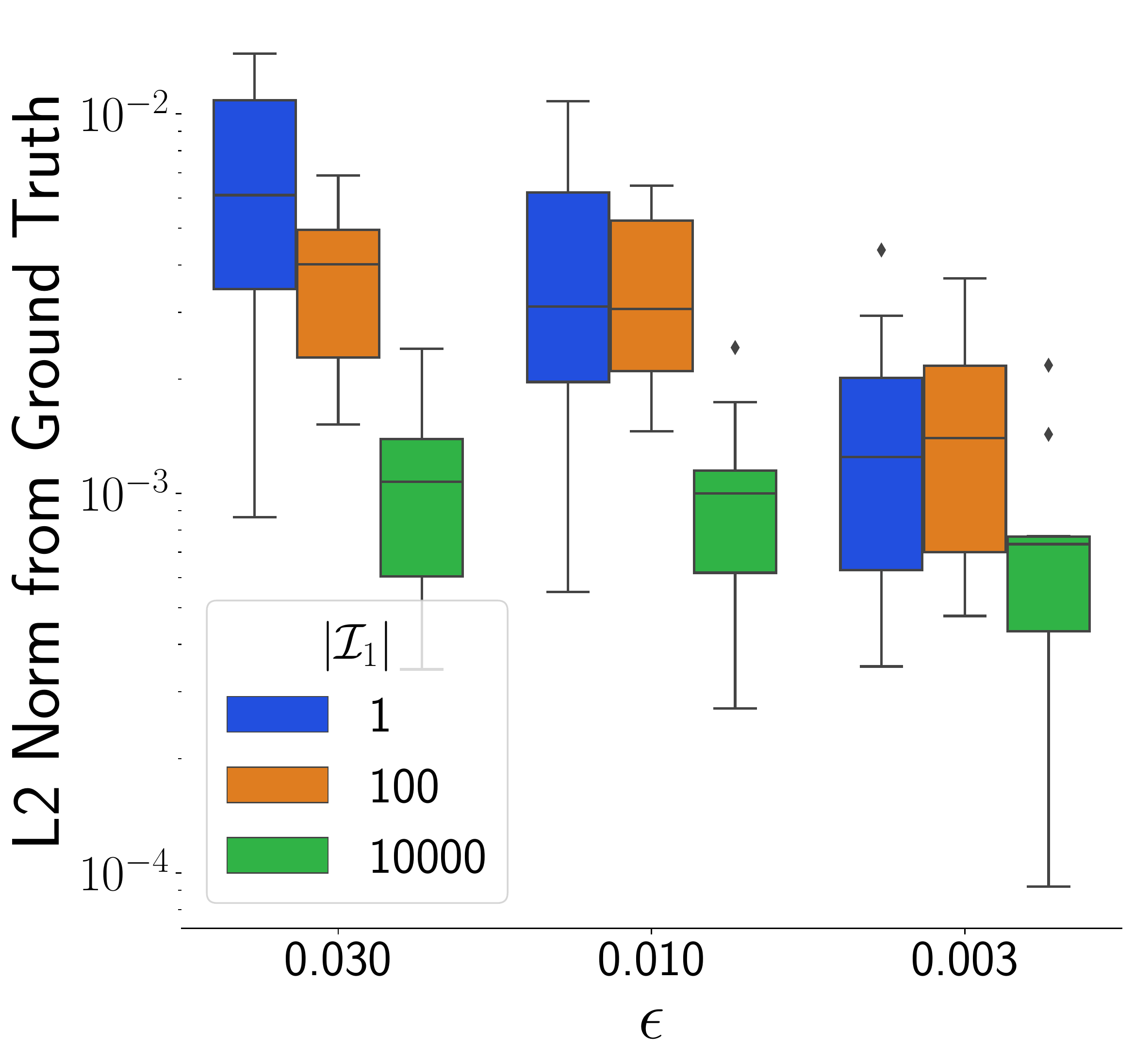}
	\subcaption{Orkut}
\end{minipage}
\caption{Accuracy and convergence analysis for \CIS[5]s. 
We plot the L2-norm between the \name estimate and the exact value of the count vector $\cC[5]$ (\Cref{eq.mu.edge.sum}) of all non-isomorphic subgraph patterns against various configurations of the parameters $\epsilon$ and $|\cI_1|$. 
As expected, the accuracy improves as the error bound $\epsilon$ decreases and the number of seed subgraphs $|\cI_1|$ increases.
Each box and whisker represents $10$ runs.
}
\label{fig:acc-l2-5}
\end{figure} 

\paragraph{Accuracy on Small $k$.}
For $k \in \{3, 5\}$, we evaluate the L2-norm between the \name estimate and the exact value of the count vector $\cC[k]$ (\Cref{eq.mu.edge.sum}) of all non-isomorphic subgraph patterns.
\Cref{fig:acc-l2-5} shows results for $k = 5$ (where the number of patterns of interest $|\cH| = 21$) for different settings of the parameters $\epsilon$ and $|\cI_1|$.
In all datasets, we note that the L2-norm decreases as $\epsilon$ decreases from $0.3$ to $0.003$ and as $|\cI_1|$ increases from $100$ to $10^4$.
Between the worst setting, $(\epsilon,|\cI_1|)=(0.3,100)$, and the best $(\epsilon,|\cI_1|)=(0.003,10^4)$, we see an error reduction close to an order of magnitude.
This is due to \Cref{thm.bias.propagation} and \Cref{lem.rwte} because reducing
$\epsilon$ increases the number of tours, lowers the error and therefore leads to reduced error propagation. 
Increasing $\cI_1$ also reduces the number of strata and therefore error propagation.
Results for $k=5$ using the L-$\infty$ norm are deferred to \Cref{app:additional-results}-\Cref{fig:acc-linf-5}.

\begin{figure}\centering
\begin{minipage}[t]{0.29\linewidth}
	\centering
	\includegraphics[width=\linewidth]{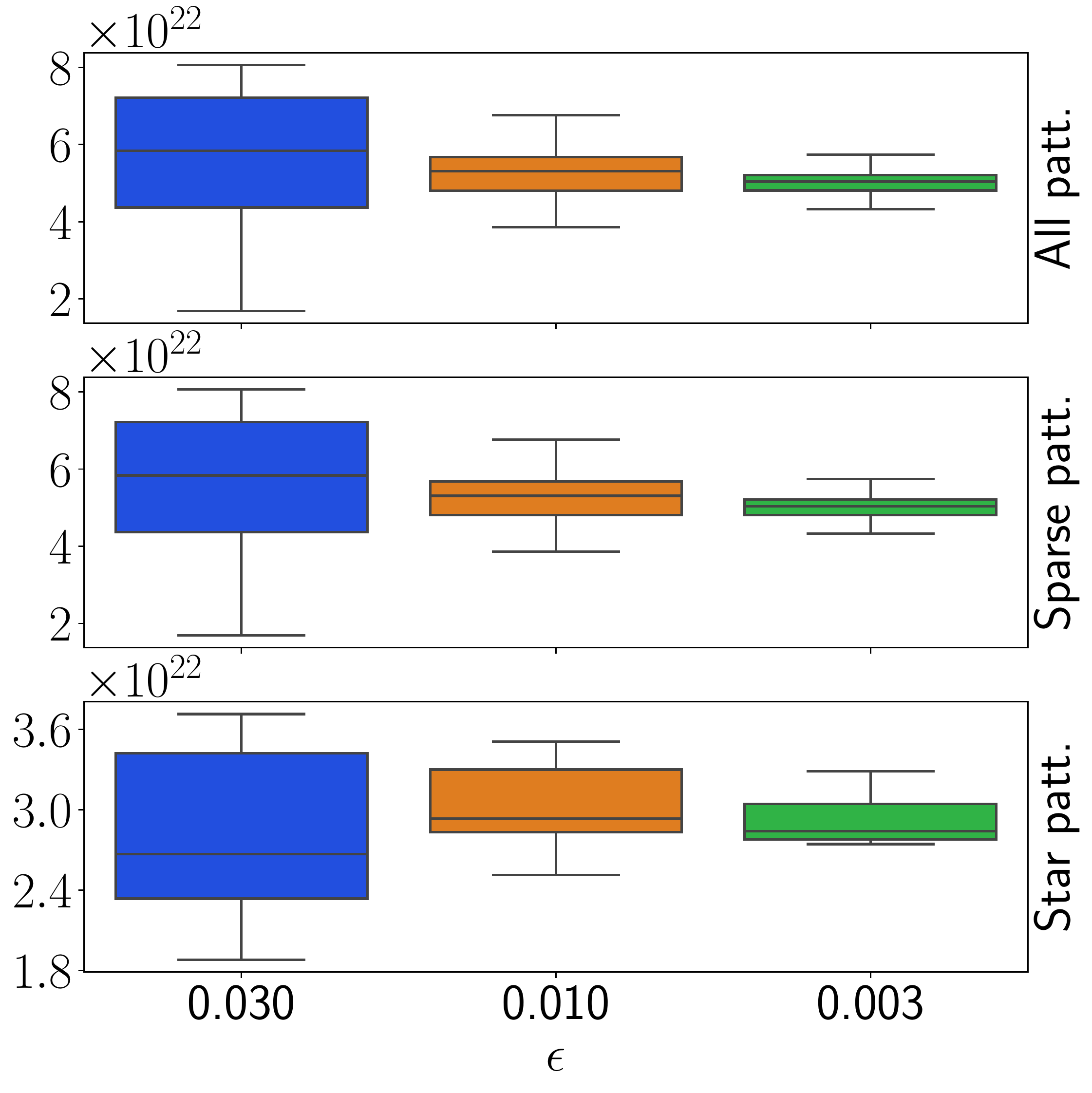}	\subcaption{Amazon}
\end{minipage}
\qquad
\begin{minipage}[t]{0.29\linewidth}
	\centering
	\includegraphics[width=\linewidth]{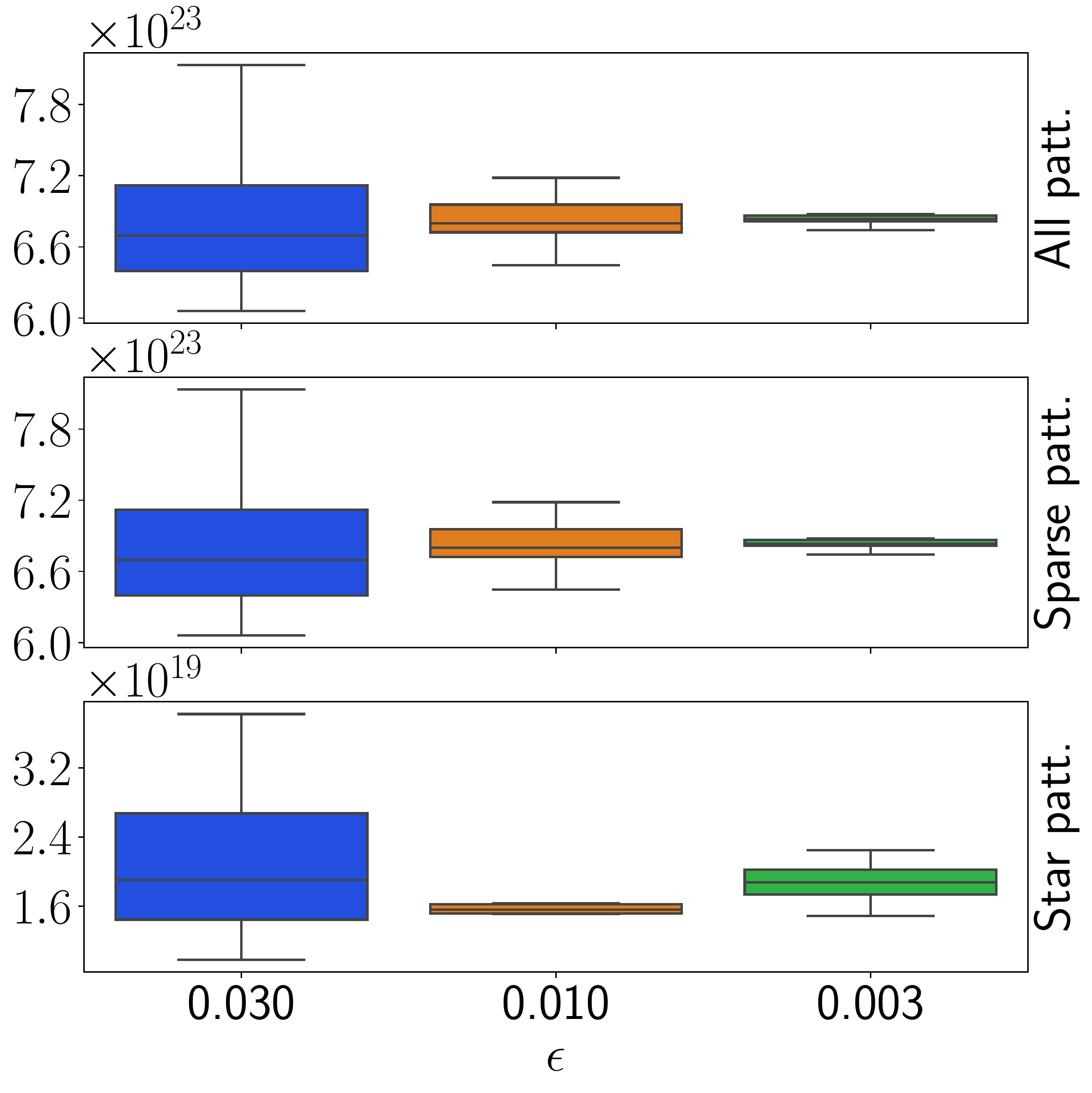}
	\subcaption{DBLP}
\end{minipage}
\qquad
\begin{minipage}[t]{0.29\linewidth}
	\centering
	\includegraphics[width=\linewidth]{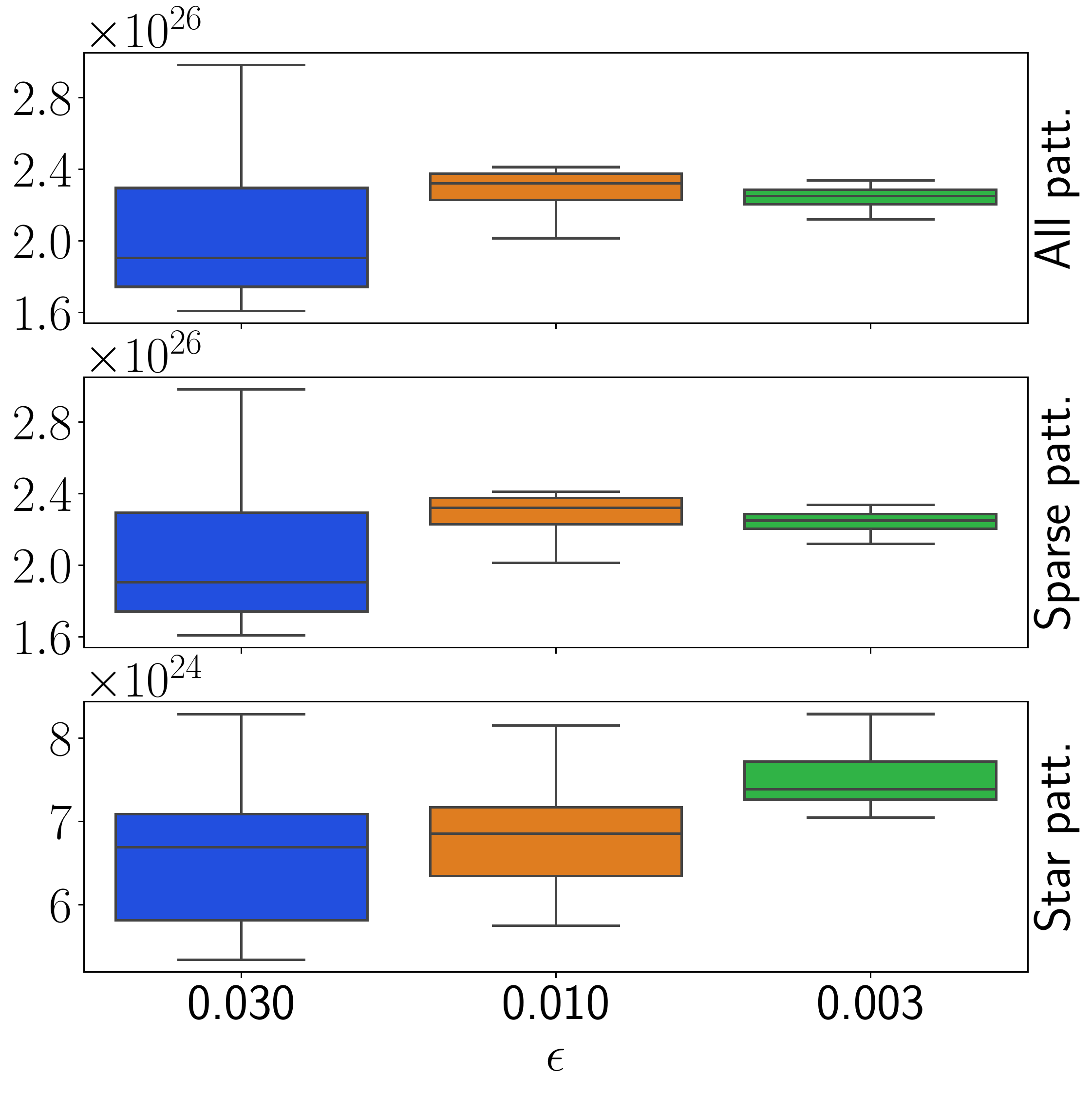}	
	\subcaption{Patents}
\end{minipage}
\qquad
\begin{minipage}[t]{0.29\linewidth}
	\centering
	\includegraphics[width=\linewidth]{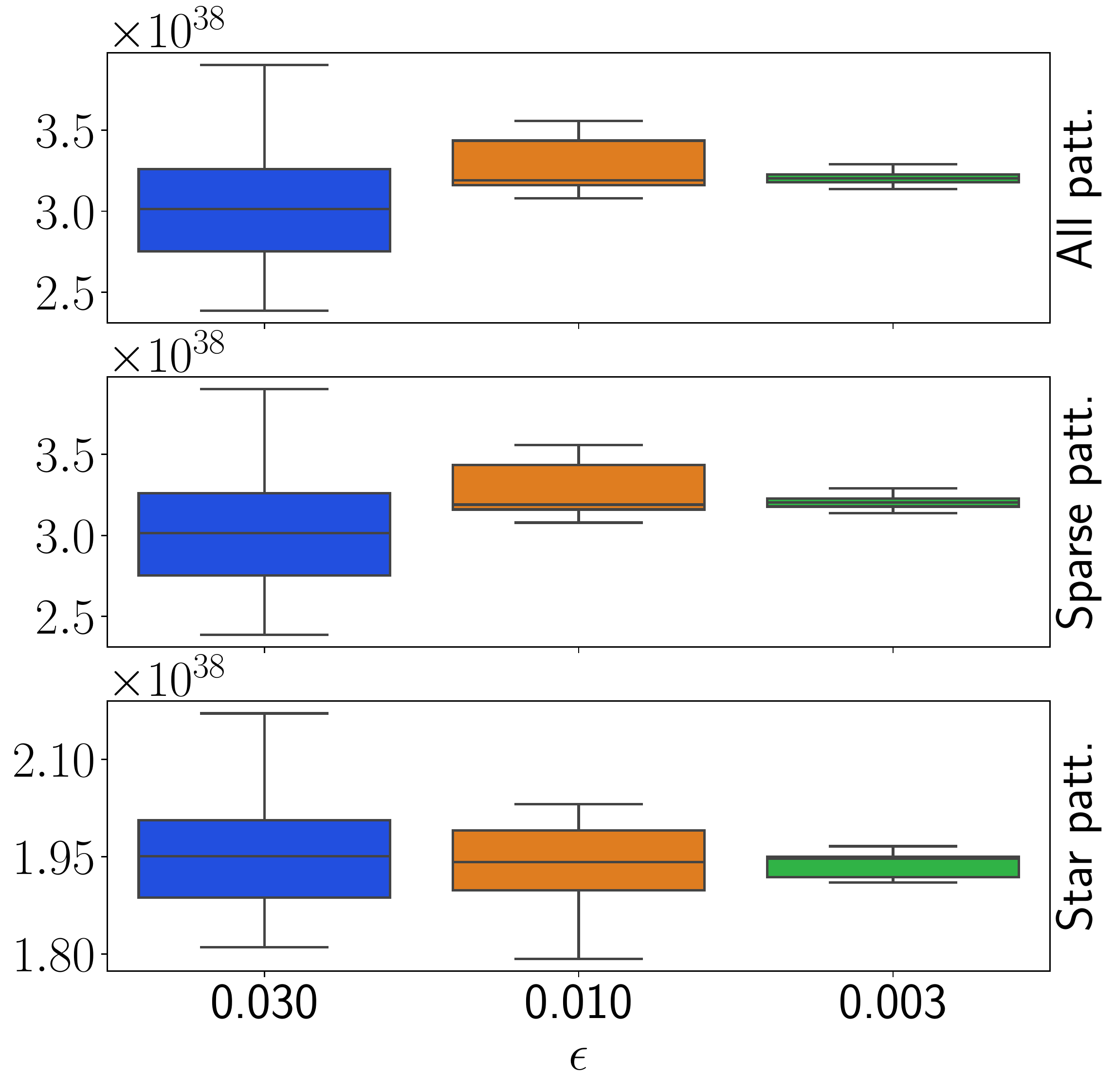}	
    \subcaption{Pokec}
\end{minipage}
\qquad
\begin{minipage}[t]{0.29\linewidth}
	\centering
	\includegraphics[width=\linewidth]{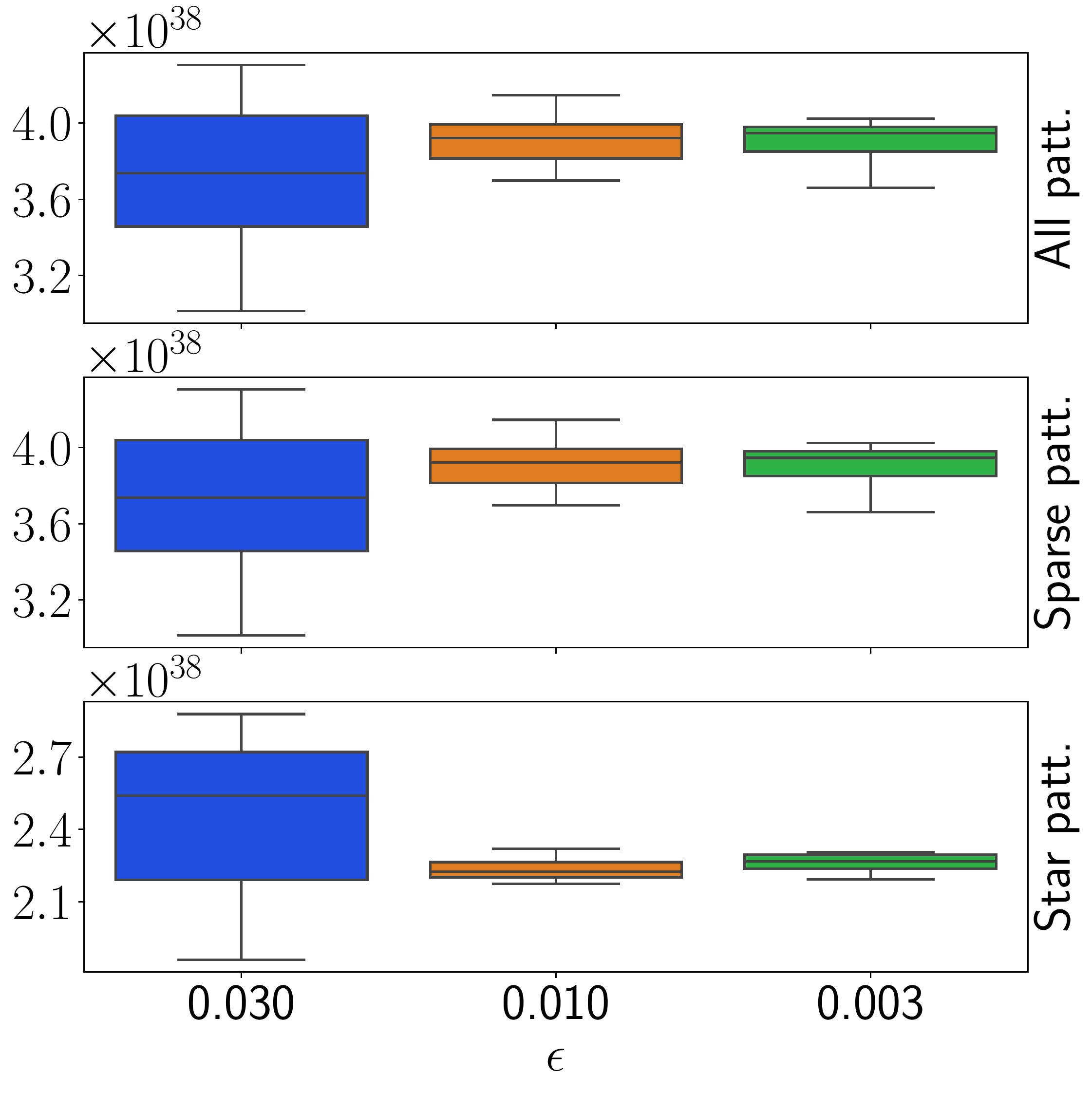}
	\subcaption{Live Journal}
\end{minipage}
\qquad
\begin{minipage}[t]{0.29\linewidth}
	\centering
	\includegraphics[width=\linewidth]{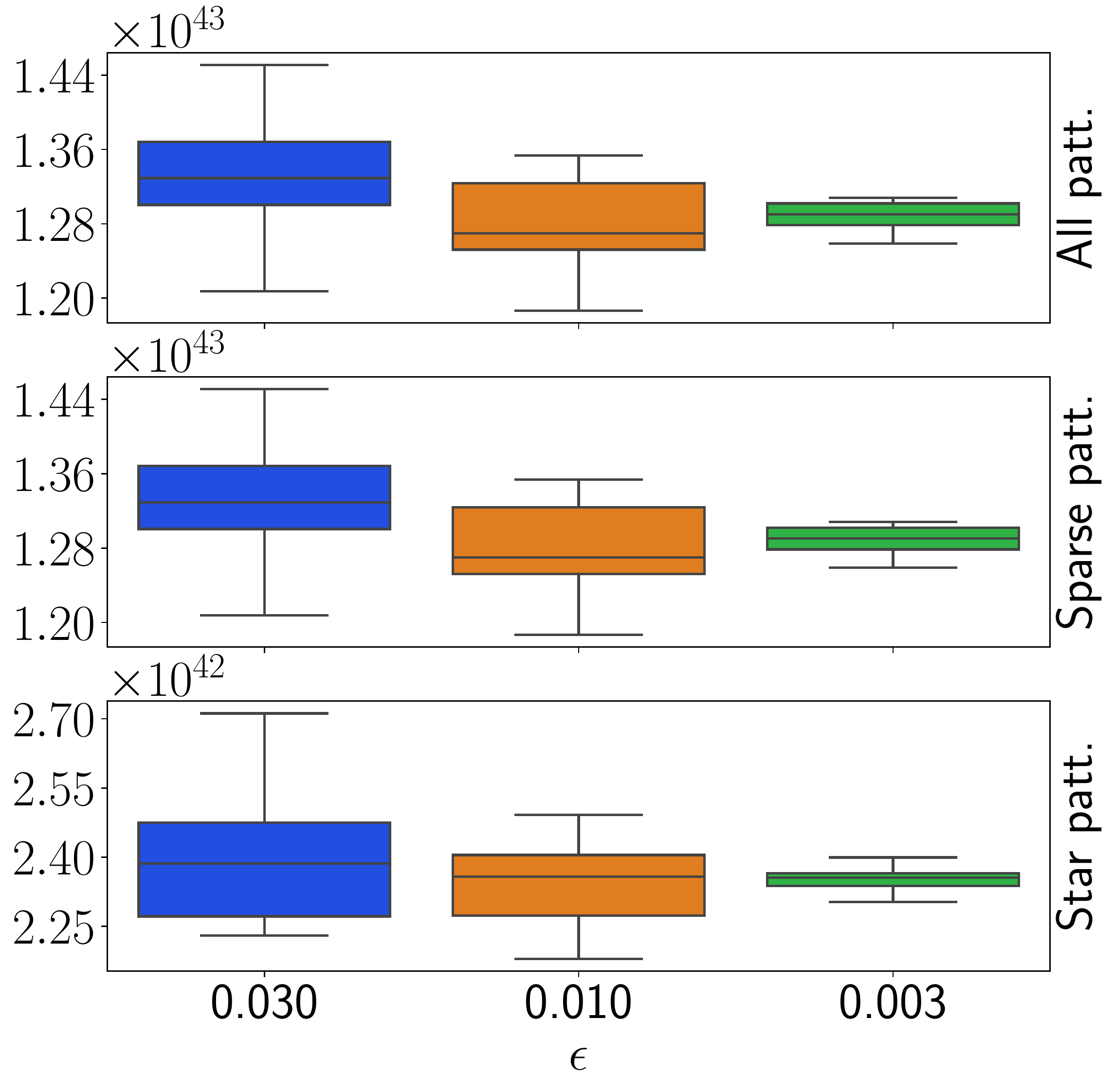}
	\subcaption{Orkut}
\end{minipage}
\caption{Convergence of \name estimates of \CIS[12] pattern counts.
We estimate the total number of subgraphs $|\cV[12]|$ and the number of sparse patterns and stars. 
Estimates over $10$ runs are presented as box and whiskers plots, which exhibit a reduction in variance as $\epsilon$ increases. 
Indeed, almost all patterns are sparse, and the most frequent substructure is a star.
}
\label{fig:conv-count-12}
\end{figure} 

\paragraph{Convergence for Large $k$.} When $k>5$, subgraph counts for real-world graphs are computationally intractable. Therefore, we show that \name converges in these cases as we increase the computing effort. 
Consider the hypothesis that sparse patterns are frequent in power-law networks as $k$ increases. 
To glean empirical evidence for this, we choose an appropriate pattern set $\cH$ and equivalence relationship in \Cref{def.problem}, and we use \name to compute the total number of \CIS[k]s and the number of sparse subgraphs and stars.
A subgraph is defined as sparse if its density lies between $0$ and $0.25$, according to \citet{liu2008effective}. 
In \Cref{fig:conv-count-12}, we show that \name converges for all datasets, and as expected, most patterns are sparse, with close to half of the patterns in many of the studied networks being stars. 
This proportion is attenuated in DBLP and Patents, where dense substructures naturally emerge from collaboration/citation among the authors that these graphs represent.

 \section{Related Work}
\label{sec:rel}
For better presentation, we split this section into two parts: (1) parallel \MCMC techniques and (2) methods for subgraph counting.

\paragraph{Parallel \MCMC through Splitting.} 
Since \citet{nummelin1978splitting,athreya1978new}, multiple techniques have been proposed to circumvent the burn-in period by splitting the chain into i.i.d.\ sample paths.
This approach allows practitioners to compute unbiased estimates in parallel and determine confidence intervals.
Perfect sampling methods based on coupling~\citep{propp1996exact} require the transitions to be monotonic w.r.t.\ some ordering over the state space, and annealing/tempering ~\citep{neal2001annealed} methods require some notion of temperature, which are absent in general graph random walks.
Methods such as \citep{mykland1995regeneration,jacob2020unbiased,glynn2014exact} require a minorization condition to hold, albeit implicitly.

Regeneration point-based methods on finite state chains \citep{cooper2016fast,massoulie2006peer,Avrachenkov2016,avrachenkov2018revisiting,savarese2018monte,teixeira2018graph} are more general because they only rely on standard ergodicity conditions.
Although \citet{cooper2016fast,massoulie2006peer} used tours to estimate graph properties, \citet{Avrachenkov2016,avrachenkov2018revisiting} proposed supernodes to reduce running times.
The studies in \citet{savarese2018monte,teixeira2018graph} further used supernode-based tours to estimate gradients in \emph{RBM}s and to count subgraphs.
To the best of our knowledge, no existing regeneration point method controls running times through stratification.

\paragraph{Subgraph Counting through Sampling.}

Many random walk algorithms have been proposed to sample subgraphs, with some methods only capable of estimating subgraph pattern {\em distributions}, which is much easier than estimating {\em counts}. 
The studies of \GUISE~\citep{bhuiyan2012guise} and \RSS~\citep{matsuno2020improved} use a Metropolis-Hastings~\citep{hastings1970monte} walk, and the latter improves the mixing time of the underlying Markov chain using canonical paths~\citep{sinclair1992improved}.
Waddling~\citep{han2016waddling} and \IMPRG~\citet{chen2018mining} perform a simple random walk over the input graph and use specialized estimators to sample $5$-node patterns. 
Although \PSRW~\citep{Wang:2014} first proposed the \HON-based random walk and \RGPM~\citep{teixeira2018graph} used tours on it to estimate subgraph counts, both are limited to $k \leq 5$ due to the size of the \HON.

Multiple attempts to Monte Carlo sample subgraphs have been proposed whose scaling is limited because of the complexity of computing either the importance weights, rejection rate or variance \citep{kashtan2004efficient,Wernicke:2006,iyer2018asap,yang2018ssrw,Wang:2018}.
Efficient methods that sample dense regions/subgraphs are unfortunately not extensible to sparse patterns \citet{Jain:2017,jain2020provably}.
Motivo~\citep{Bressan:2018,Bressan:2019} is an example of color-coding methods in which an index table is built using a deterministic dynamic programming algorithm, which is then exploited to sample subgraphs uniformly and independently.
 However, \CC methods suffer from the exponential time and space complexities associated with building and accessing the index table. 
 \Motivo proposed succinct index tables and efficient out-of-core I/O mechanisms to ameliorate this issue and extended the applicability of \CC methods to larger subgraphs.
Please, check \citet{ribeiro2019survey} for an extensive survey on subgraph counting methods.

\section{Conclusions}
In this paper, we propose the \name estimator that uses {\em sequentially stratified regenerations} to control the running time of a random walk tour-based \MCMC.
We prove that the estimator is consistent (w.r.t.\ the number of random walk tours) and that the time and memory complexity of our implementation for the subgraph counting problem is linear in the number of patterns of interest and polynomial in the subgraph size.
We empirically verify our claims on multiple graph datasets and show that \name can accurately estimate subgraph counts with a smaller memory footprint compared to that of the state-of-the-art Motivo~\citep{Bressan:2019}.
\name is currently the only subgraph pattern count estimator that can estimate $k=10,12$ node patterns in million-node graphs.
Beyond our specific application, \name provides a promising way to expand the sphere of influence of regenerative simulation in discrete reversible \MCMC. \bibliographystyle{spbasic}
\bibliography{references}

\appendix
\captionsetup[figure]{font=footnotesize,labelfont=footnotesize}
\captionsetup[table]{font=footnotesize,labelfont=footnotesize}

\section{Notation}
\label{sec.notation}
The most important notations from the paper are summarized in \Cref{tab.notations}.
\newcolumntype{L}{>{\centering\arraybackslash}m{3cm}}
\begin{table*}[h!]
\centering
\scalebox{0.9}{
 \begin{tabular}{p{0.25\linewidth} | p{0.75\linewidth}}
 \bottomrule
\textbf{Symbol} & \textbf{Explanation}\\
\bottomrule
$\cG = (\cV,\cE)$ & The graph where we have neighbor query access and whose edge sum is being computed. \\
$N(u)$, $\dg(u)$ & Neighborhood and degree of a vertex in $\cG$ if no subscript is specified. \\
$\mu(\cE)$ & The sum over edges in $\cE$ (or some subset) of some function $f$.\\
$\bPhi$, $p_{\bPhi}(u,v)$, $\pi_{\bPhi}(u)$ & The random walker on $\cG$, its transition probability and stationary distribution.\\
$\bX$, $\xi$, $\cT$& an \RWT (tour), its length and a set of $m$ {\RWT}s.\\
$\hmus(\cT; f, \cG)$ & an \RWTE of $\mu(\cE)$.\\
$\delta$ & The spectral gap of the transition probability matrix of a chain.\\
$\hsigma(\cdot)^{2}$& Empirical variance of an \RWTE.\\
$\vo_I$, $\cG_I$& Collapsed state and graph obtained by collapsing $I\subset\cV$.\\
$q<r<t$& Strata ids always used in the same order $1\leq q<r<t\leq R$ .\\
$\rho\colon \cV \to \{1, \ldots, R\}$& Stratification function.\\
$\cI_{r},\cJ_{r}$& $r$-th vertex and edge stratum.\\
$\cG_{r} = (\cV_{r}, \cE_{r})$& $r$-th graph stratum.\\
$\vo_r$, $\cT_{r}$& Supernode in each stratum and a set of $m_r$ perfectly sampled tours from $\vo_r$.\\
$\dg(\vo_{r})$, $p_{\bPhi_{r}}(\vo_{r}, \cdot)$&Degree and transition probability out of the supernode.\\
$\hdgvo{r}$, $\hpvo{r}$& Estimated degree and transition probability out of the supernode.\\
$\dT_{q}$&$m_q$ {\RWT}s samples using supernode estimates.\\
$\hbeta{q}{r}$& The estimate of the number of edges between $\cI_{q}$ and $\cI_{r}$.\\
$\hbU{q}{r}$& Multiset of states visited by $\dT_{q}$ that lie in $\cI_{r}$.\\
$\hmu_{\name}$, $\hmu\left(\dT_{2:r}; f \right)$& Overall and per-stratum \name estimate.\\
$\delta_{r}$, $\nu_{r}$, $\lambda_{r}$ & Spectral gap and the errors in the supernode estimates in the $r$-th stratum.\\
$G=(V,E,L)$ & The labelled input graph in which we want to count subgraphs.\\
$\sg{G}{V'}$&Subgraph induced by $V'$ in $G$.\\
$\cH$& Nonequivalent (non-isomorphic) patterns of interest.\\
$\cC[k] = (\cC[k]_H)_{H \in \cH}$ & The subgraph pattern count vector.\\
$\cG[k] = (\cV[k], \cE[k])$& The \HON[k] that provides neighborhood query access and is used to count subgraphs.\\
$\gamma(u,v)$& Number of edges in $\cE[k-1]$ that represent the same subgraph as $(u,v)$.\\
$\dist(u)$ & Shortest path distance from $u \in V$ to any seed vertex in $V(\cI_{1})$.\\
$\Vstar$ & The largest connected subset of $V(s)$ that constitutes an intersection between $s$ and $\dds \in \cI_{1}$.\\
$\rvrSize$& Reservoir size.\\
$\Delta_{G}$, $D_{G}$ & Maximum degree in and diameter of $G$.\\
$\cA_{s}$& Articulation points in $s$.\\
$\epsilon$& Per-stratum error bound used to control tour count.\\
\bottomrule
\end{tabular}
}

\caption{Table of Notations}
\label{tab.notations}
\end{table*}

\section{Proofs for \Cref{sec.preliminaries}}
\subsection{\MCMC Estimates}
\label{sec.ht.mcmc}
Given a graph $\cG$, when the $|\cE|$ is unknown, the \MCMC estimate of $\nicefrac{\mu(\cE)}{|\cE|}$ is given by:
\begin{proposition}[\MCMC Estimate~\citep{geyer1992practical,geman1984stochastic,hastings1970monte}]
	\label{def.ht.mcmc}
	When $\cG$ from \Cref{def.simplerw} is connected, the random walk $\bPhi$ is reversible and positive recurrent with stationary distribution $\pi_{\bPhi}(u) = \nicefrac{\dg(u)}{2|\cE|}$. 
	Then, the \MCMC estimate
	\begin{equation*}
		\hmu_{0}\left( (X_i)_{i=1}^{t}\right) =
		 \frac{1}{t-1}\sum_{i=1}^{t-1} f(X_i,X_{i+1})\,,
	\end{equation*}
	computed using an arbitrarily started sample path $(X_i)_{i=1}^{t}$ from $\bPhi$ is an asymptotically unbiased estimate of $\nicefrac{\mu(\cE)}{|\cE|}$.
	When $\cG$ is non-bipartite, i.e., $\bPhi$ is aperiodic, and $t$ is large, $\hmu_{0}$ converges to $\nicefrac{\mu(\cE)}{|\cE|}$ as
	\begin{equation*}
        \Big|\EE[\hmu_{0}\left( (X_i)_{i=1}^{t}\right)] - \nicefrac{\mu(\cE)}{|\cE|}\Big|
		 \leq 
		 B \, \frac{C}{t \delta(\bPhi)} \,,
	\end{equation*}
	where $\delta(\bPhi)$ is the spectral gap of $\bPhi$ and $C\triangleq \sqrt{\frac{1-\pi_{\bPhi}(X_1)}{\pi_{\bPhi}(X_1)}}$ such that $f\bdot \leq B$.
\end{proposition}
\begin{proof}[Asymptotic unbiasedness]
    Because $\cG$ is undirected, finite and connected, $\bPhi$ is a finite state space, irreducible, time-homogeneous Markov chain and is therefore positive recurrent~\cite[3-Thm.3.3]{bremaud2001markov}.
    The reversibility and stationary distribution holds from the detailed balance test~\cite[2-Cor.6.1]{bremaud2001markov} because
    \begin{equation*}
        \pi_{\bPhi}(u) \, p_{\bPhi}(u,v) = \pi_{\bPhi}(v) \, p_{\bPhi}(v,u) = \frac{\ind{(u,v)\in\cE}}{2|\cE|}\,.
    \end{equation*}
    The ergodic theorem~\cite[3-Cor.4.1]{bremaud2001markov} then applies because $f$ is bounded and we have
    \begin{equation*}
        \lim_{t\to\infty} \frac{1}{t-1}\sum_{i=1}^{t-1} f(X_i,X_{i+1})
         = \sum_{(u,v) \in \cV\times\cV}\pi_{\bPhi}(u) \, p_{\bPhi}(u,v) f(u,v) = \frac{\mu(\cE)}{|\cE|}\,.
    \end{equation*}
\end{proof}
\begin{proof}[Bias]
    Let the $i$-step transition probability of $\bPhi$ be given by $p^{i}_{\bPhi}(u,v)$.
    The bias at the $i$-th step is given by 
    \begin{equation*}
        \begin{split}
            \bias_i &= 
        \Big| \EE\left[ f(X_i,X_{i+1})\right] - \sum_{(u,v) \in \cV\times\cV}\pi_{\bPhi}(u) \, p_{\bPhi}(u,v) f(u,v) \Big| \\
        &= \Big| \sum_{(u,v) \in \cV\times\cV} p^{i}_{\bPhi}(X_1,u) \, p_{\bPhi}(u,v) f(u,v)  - \sum_{(u,v) \in \cV\times\cV}\pi_{\bPhi}(u) \, p_{\bPhi}(u,v) f(u,v) \Big|\\
        &\leq B \Big| \sum_{u \in \cV} p^{i}_{\bPhi}(X_1,u) \sum_{v\in\cV} p_{\bPhi}(u,v)  - \sum_{u \in \cV}\pi_{\bPhi}(u) \sum_{v\in\cV} p_{\bPhi}(u,v) \Big|\\
        &\leq B \Big| \sum_{u \in \cV} p^{i}_{\bPhi}(X_1,u)- \sum_{u \in \cV}\pi_{\bPhi}(u)  \Big|
        \leq B  \sum_{u \in \cV} \Big| p^{i}_{\bPhi}(X_1,u)- \pi_{\bPhi}(u)  \Big| \,,
        \end{split}
    \end{equation*}
    where $f\bdot \leq B$, and the final inequality is due to Jensen's inequality.
    From \cite[Prop-3]{diaconis1991geometric},
    \begin{equation*}
        \bias_{i} \leq B \sqrt{\frac{1-\pi_{\bPhi}(X_1)}{\pi_{\bPhi}(X_1)}} \beta_{*}^{i} \,,
    \end{equation*}
    where $\beta_{*} = 1-\delta(\bPhi)$ is the SLEM of $\bPhi$.
    Because of Jensen's inequality and by summing a GP,
    \begin{equation*}
        \Big|\EE[\hmu_{0}\left( (X_i)_{i=1}^{t}\right)] - \frac{\mu(\cE)}{|\cE|}\Big|
         \leq \frac{1}{t-1} \sum_{i=1}^{t-1} \bias_i 
\leq \frac{B}{t-1} \sqrt{\frac{1-\pi_{\bPhi}(X_1)}{\pi_{\bPhi}(X_1)}} \frac{1-\beta_{*}^{t}}{1-\beta_{*}} \,.
    \end{equation*}
    Assuming that  $\beta_{*}^{t} \approx 0$ and $t-1 \approx t$ when $t$ is sufficiently large completes the proof.
\end{proof}

\begin{lemma}[\citet{Avrachenkov2016}]
	\label{lemma.tour.length.sq}	
	Let $\bPhi$ be a finite state space, irreducible, time-homogeneous Markov chain, and let $\xi$ denote the return time of \RWT started from some $\xo\in\cS$ as defined in \Cref{def:rwt}.
	If $\bPhi$ is reversible, then
	\begin{equation}
		\EE\left[ \xi^{2} \right] \leq \frac{3}{\pi_{\bPhi}(\xo)^{2}\delta(\bPhi)} \,,
	\end{equation} 
    where $\pi_{\bPhi}(\xo)$ is the stationary distribution of $\xo$, and $\delta(\bPhi)$ is the spectral gap of $\bPhi$.
    When $\bPhi$ is not reversible, the second moment of return times is given by \Cref{non.rev.sec.mom}.
\end{lemma}
\begin{proof}
	Using \cite[Eq~2.21]{aldous-fill-2014}, we have
    \begin{equation}
        \label{non.rev.sec.mom}
		\EE\left[ \xi^{2} \right] = \frac{1+2\EE_{\pi_{\bPhi}}(T_{\xo})}{\pi_{\bPhi}(\xo)} \,,
	\end{equation}
	where $\EE_{\pi_{\bPhi}}(T_{\xo})$ is the expected hitting time of $\xo$ from the steady state.
	Combining \cite[Lemma~2.11 \& Eq~3.41]{aldous-fill-2014} and accounting for continuization yields 
    \begin{equation*}
    \begin{split}
		\EE_{\pi_{\bPhi}}(T_{\xo}) \leq \frac{1}{\pi_{\bPhi}(\xo)\delta(\bPhi)} \,.\text{ and therefore, }
		\EE\left[ \xi^{2} \right] \leq \frac{1+\frac{2}{\pi_{\bPhi}(\xo)\delta(\bPhi)}}{\pi_{\bPhi}(\xo)}
		< \frac{3}{\pi_{\bPhi}(\xo)^{2}\delta(\bPhi)} \,,
    \end{split}
	\end{equation*}
	because $\pi_{\bPhi}(\xo)$ and $\delta(\bPhi)$ lie in the interval $(0,1)$.
\end{proof}

\begin{proposition}
    \label{prop.tour.stitch}
    Given a positive recurrent Markov chain $\bPhi$ over state space $\cS$ and a set of $m$ {\RWT}s $\cT$ and
    assuming an arbitrary ordering over $\cT$, where $\bX^{(i)}$ is the $i$th \RWT in $\cT$,
    $\bX^{(i)}$ and $|\bX^{(i)}|$ are i.i.d.\ processes such that $\EE[|\bX^{(i)}|] < \infty$, and when the tours are stitched together as defined next, the sample path is governed by $\bPhi$.
    For $t\geq 1$, define $\Phi_{t} = X^{N_{t}}_{t - R_{N_{t}}}$, where $R_i = \sum_{i'=1}^{i-1} |\bX^{i}|$ when $i>1$ and $R_{1} = 0$ and $N_t = \max\{i \colon R_i < t\}$.
\end{proposition}
\begin{proof}
    $R_i$ is a sequence of stopping times. Therefore, the strong Markov property~\citep[2-Thm.7.1]{bremaud2001markov} states that sample paths before and after $R_i$ are independent and are governed by $\bPhi$.
    Because $\bPhi$ is positive recurrent and $\xo$ is visited i.o.,\ the regenerative cycle theorem~\citep[2-Thm.7.4]{bremaud2001markov} states that these trajectories are identically distributed and are equivalent to the tours $\cT$ sampled according to \Cref{def:rwt}.
    $\EE[|\bX^(i)|] < \infty$ due to positive recurrence.
\end{proof}
\subsection{Proof of \Cref{lem.rwte}}
\label{proof.lem.rwte}
\begin{proof}[Unbiasedness and Consistency]
    Because $\cG$ is connected, $\bPhi$ is positive recurrent with steady state $\pi_{\bPhi}(u) \propto \dg(u)$ due to \Cref{def.ht.mcmc}.
    Consider the reward process $F^{(i)} = \sum_{j=1}^{|\bX^{(i)}|} f(X_j^{(i)},X_{j+1}^{(i)})$, $i\geq1$.
    From \Cref{prop.tour.stitch}, $F^{(i)}$ and $|\bX^{(i)}|$ are i.i.d. sequences with finite first moments, because $F^{(i)} \leq B |\bX^{(i)}|$.
    Let $N_t$ and $R_i$ be as defined in \Cref{prop.tour.stitch}.
    
    Therefore, from the renewal reward theorem~\citep[3-Thm.4.2]{bremaud2001markov}, we have 
	\begin{equation*}
    \frac{\EE[F^{(i)}]}{\EE[|\bX^{(i)}|]} 
    = \lim_{t\to\infty}\frac{\sum_{i=1}^{N_t} F^{(i)}}{t} 
    = \lim_{t\to\infty}\frac{\sum_{i=1}^{N_t} F^{(i)}}{R_{N_{t}}} \cdot \frac{R_{N_{t}}}{t} 
	= \frac{\sum_{i=1}^{N_t} F^{(i)}}{R_{N_{t}}} \,,
    \end{equation*}
    where the final equality holds because $\lim_{t\to\infty}\frac{R_{N_{t}}}{t}  = 1-\lim_{t\to\infty}\frac{t-R_{N_{t}}}{t}$, and $\lim_{t\to\infty}\frac{t-R_{N_{t}}}{t}$ converges to $0$ as $t\to\infty$ because $|\bX^{\bdot}|<\infty$ w.p.\ 1 because $\bPhi$ is positive recurrent.

    From \Cref{prop.tour.stitch} and the definition of $F^{(i)}$, $\sum_{i=1}^{N_t} F^{(i)} = \sum_{j=1}^{R_{N_{t}}} f(\Phi_j,\Phi_{j+1})$, and because $f$ and $\pi_{\bPhi}$ are bounded, we have from the ergodic theorem~\citep[3-Cor.4.1]{bremaud2001markov},
	\begin{equation*}
        \frac{\EE[F^{(i)}]}{\EE[|\bX^{(i)}|]} 
    = \lim_{t\to\infty}\frac{\sum_{j=1}^{R_{N_{t}}} f(\Phi_j,\Phi_{j+1})}{R_{N_{t}}}
    \asc
	\sum_{(u,v)\in \cV\times\cV}\pi_{\bPhi}(u)p_{\bPhi}(u,v) g(u,v)
	= \frac{2\mu(\cE)}{2|\cE|} \,.
    \end{equation*}
    From Kac's formula~\citep[Cor.2.24]{aldous-fill-2014}, $\nicefrac{1}{\EE[|\bX^{(i)}|]} = {\pi_{\bPhi}(\xo)} = \frac{\dg(\xo)}{2|\cE|}$, and
    \begin{equation*}
        \EE\left[\frac{\dg(\xo)}{2} F^{(i)}\right] \asc \mu(\cE) \,.
    \end{equation*}
    $\hmus(\cT; f, \cG)$ is unbiased by linearity of expectations on the summation over $\cT$, and consistency is a consequence of Kolmogorov's SLLN~\citep[1-Thm.8.3]{bremaud2001markov}.
\end{proof}
\begin{proof}[Running Time] 
    From Kac's formula~\citep[Cor.2.24]{aldous-fill-2014}, $\EE[|\bX^{(i)}|] = \frac{2|\cE|}{\dg(\xo)}$.
    From \Cref{prop.tour.stitch}, tours can be sampled independently and thus parallelly.
    All cores will sample an equal number of tours in expectation, yielding the running time bound.
\end{proof}
\begin{proof}[Variance] 
    Because $f\bdot <B$, and tours are i.i.d.,\ the variance is given by 
    \begin{equation*}
        \Var\left(\hmus(\cT)\right) 
        = \Var\left(
            \frac{\dg(\xo)}{2m} \sum_{\bX \in \cT} \sum_{j=1}^{|\bX|} f(X_j,X_{j+1})
            \right)
        \leq \frac{\dg(\xo)^2 B^2}{4m}\Var\left( |\bX| \right) \,.
    \end{equation*}
    From \Cref{lemma.tour.length.sq} and Kac's formula~\citep[Cor.2.24]{aldous-fill-2014}, $\Var\left( |\bX| \right)$ is given by
    \begin{equation*}
        \Var\left( |\bX| \right) \leq 
        \frac{3}{\pi_{\bPhi}(\xo)^{2}\delta(\bPhi)}
        - \frac{1}{\pi_{\bPhi}(\xo)^{2}}
        \leq \frac{3}{\pi_{\bPhi}(\xo)^{2}\delta(\bPhi)}
        = \frac{12 {|\cE|}^{2}}{\dg(\xo)^2\delta(\bPhi)}
        \,.
    \end{equation*}   
\end{proof}

\section{Proofs for \Cref{sec.estimator}}
\label{proofs.sec.estimator}
\begin{assumption}
    \label{as.simplifying}
    For each $\cG_{r}$, $1<r\leq R$ from \Cref{def.gr}, assume $\dg(\vo_{r})$ is known and that $p_{\bPhi_{r}}(\vo_{r}, \cdot)$ can be sampled from.
\end{assumption}
\begin{proposition}[{\RWT}s in $\bPhi_{r}$]
    \label{prop.rwt.phi.r}
    Under \Cref{as.simplifying}, given access only to the original chain $\bPhi$ and stratifying function $\rho$, let $\bPhi_{r}$ be the random walk in the graph stratum $\cG_{r}$ from \Cref{def.gr}.
    To sample an \RWT $(X_{i})_{i=1}^{\xi}$ over $\bPhi_{r}$ from the supernode $\vo_{r}$, we set $X_{1} = \vo_{r}$, sample $X_2 \sim \pvo{r}$, and then, until $\rho(X_{\xi+1}) < r$, we sample 
    \begin{equation*}
        X_{i+1} \sim \unif\left(\N_{G_r}\left(X_i\right)\right) \equiv 
        \begin{cases}
            \unif(\N_{\cG}(X_{i}))  & \text{if } \rho(X_{i}) = r \\
            \unif(\N_{\cG}(X_{i})\cap \cI_{r}) & \text{if } \rho(X_{i}) > r
        \end{cases} \,.
    \end{equation*}
\end{proposition}
\begin{proof}
    The proof is a direct consequence of \Cref{def.gr} and \Cref{def.simplerw}.
\end{proof}
\begin{proposition}[Perfectly Stratified Estimate]
    \label{prop.perfect.sampling}
    Under \Cref{as.simplifying}, given the \EPART (\Cref{def:epart}) stratum $\cG_{r}$ (\Cref{def.gr}), bounded $f \colon \cE\to\RR$ and a set of $m$ {\RWT}s $\cT_{r}$ over $\bPhi_{r}$ from $\vo_r$ from \Cref{prop.rwt.phi.r}, the per stratum estimate is given by
	\begin{equation}
		\label{eq.rwte.r}
		\hmu(\cT_{r}; f, \cG_{r}) = \frac{\dg(\vo_r)}{2m} \sum_{\bX \in \cT_{r}}\sum_{j=2}^{|\bX|-1} f(X_j,X_{j+1})\,,
    \end{equation}
    where $X_j$ is the $j$th state visited in the \RWT $\bX \in \cT_{r}$.
    For all $r>1$, $\hmu(\cT_{r}; f, \cG_{r})$ is an unbiased and consistent estimator of $\mu(\cJ_{r}) = \sum_{(u,v)\in\cJ_{r}} f(u,v)$, where $\cJ_{r}$ is the $r$-th edge stratum defined in \Cref{def.dependent.strat}.
\end{proposition}
\begin{proof}
    Define $f' \colon \cE_{r} \to \RR$ as $f'(u,v) \triangleq \ind{u,v \neq \vo_{r}} f(u,v)$.
    By \Cref{def:rwt}, in each \RWT $\bX \in \cT_{r}$, $f'(X_1,X_2) = f'(X_{|\bX|},X_{|\bX|+1}) = 0$, and therefore, $\hmu(\cT_{r}; f, \cG_{r}) = \hmus(\cT; f', \cG_{r})$, where $\hmus$ is the \RWTE from \Cref{lem.rwte}. Moreover, because $\cG_{r}$ is connected,
    \begin{equation*}
        \EE\left[ \hmu(\cT_{r}; f, \cG_{r}) \right] = \EE\left[ \hmus(\cT; f', \cG_{r}) \right] = \sum_{(u,v) \in \cE_{r}} f'(u,v) = \sum_{(u,v) \in \cJ_{r}} f(u,v) \,,
    \end{equation*}
    where the final equality holds because $\cE_{r}$ is the union of $\cJ_{r}$ and edges incident on the supernode.
    Consistency is also due to \Cref{lem.rwte}.
\end{proof}

\subsection{Proof of \Cref{prop.epart.suff}}
\label{proof.prop.epart.suff}
\begin{proof}
	\ref{component.cond} is necessary because when \ref{component.cond} does not hold, there exists a component such that the minimum value of $\rho$ in that component is $\ddot{r}>0$ such that in $\cG_{\ddot{r}}$ (\Cref{def.gr}), and the supernode $\vo_{\ddot{r}}$ will be disconnected from all vertices.
	If \ref{interconnect.cond} is violated, a vertex $\ddot{u}$ exists that is disconnected in $\cG_{\rho(\ddot{u})}$, and if \ref{preconnect.cond} is violated, the supernode is disconnected.
	Finally, it is easily seen that these conditions sufficiently guarantee that  each stratum is connected, and the stratification is an \EPART.
\end{proof}

\subsection{Proof of \Cref{thm.estimator}}
\label{proof.thm.estimator}
We begin by defining the multi-set containing the end points of edges between vertex strata.
\begin{definition}
    \label{def.border.sets}
    Given $\cG$ stratified into $R$ strata, $\forall 1 \leq q < t \leq R$ define border multi-sets as 
    \begin{equation*}
        \cB{q}{t} \triangleq \{v \, \forall (u,v) \in \cE \colon u \in \cI_{q} \tand v \in \cI_{t}\} \,.
    \end{equation*}
    The degree of the supernode in $\cG_{r}$ (\Cref{def.gr}) is then given by $\dg(\vo_{r}) = \sum_{q=1}^{r-1}|\cB{q}{r}|$, and transitions out of $\vo_{r}$ can be sampled by sampling $q \in \{1,\ldots,r-1\}$ w.p.\ $\propto |\cB{q}{r}|$ and then by uniformly sampling from $\cB{q}{r}$.
\end{definition}

\begin{proposition}
    \label{prop.inductive.claim}
Given the setting in \Cref{def.recursive.step,def.estimator}, for all $1 \leq r < t \leq R$,
\begin{align}
    \label{eq.inductive.claim.beta}
        \lim_{|\dT_{2}| \to\infty} \ldots \lim_{|\dT_{r}| \to\infty} &
        \hbeta{r}{t}
          \asc
         |\cB{r}{t}|\,,\\
\label{eq.inductive.claim.U}
         \lim_{|\dT_{2}| \to\infty} \ldots \lim_{|\dT_{r}| \to\infty} &
         \hbU{r}{t}
           \sim
          \unif(\cB{r}{t})\,,\\
\label{eq.inductive.claim.phi}
         \lim_{|\dT_{2}| \to\infty} \ldots \lim_{|\dT_{r}| \to\infty} &
         p_{\hPhi_{r}}(\bX)
           =
           p_{\bPhi_{r}}(\bX)\,,\quad \forall \bX \in \dT_{r}\,.
\end{align}
i.e., each tour in $\dT_{r}$ is perfectly sampled from $\bPhi_{r}$.
\end{proposition}

\begin{proof}[By Strong Induction]
    The base case for $r=1$ holds by the base case in \Cref{def.recursive.step}.
    Now assume that \Cref{prop.inductive.claim} holds for all strata up to and including $r-1$.
    Because of the inductive claim and by \Cref{def.border.sets},
    \begin{align*}
    \lim_{|\dT_{2}| \to\infty} \ldots \lim_{|\dT_{r-1}| \to\infty}&
        \hdgvo{r} = \sum_{q=1}^{r-1} \hbeta{q}{r} \asc \sum_{q=1}^{r-1} |\cB{q}{r}| = \dgvo{r} \,,\\
\text{and similarly, }
\lim_{|\dT_{2}| \to\infty} \ldots \lim_{|\dT_{r-1}| \to\infty}&
        \hpvo{r} \equiv \pvo{r}
    \end{align*}
    because the inductive claim makes the procedure of sampling transitions out of $\vo_{r}$ in \Cref{def.recursive.step} equivalent to \Cref{def.border.sets}.
    \Cref{eq.inductive.claim.phi} holds because transition probabilities at all states other than $\vo_{r}$ are equivalent in $\bPhi_{r}$ and $\hPhi_{r}$ according to \Cref{def.sn.estimates}.
    Now recall that
    \begin{equation*}
        \hbeta{r}{t} = 
		\frac{\hdgvo{r}}{|\dT_{r}|} \sum_{\bX \in \dT_{r}}\sum_{j=2}^{|\bX|} \ind{\rho(X_{j}) = t} 
		\,.
    \end{equation*}
    Because $\hdgvo{r} = \dgvo{r}$ and the tours are sampled perfectly,
    \begin{equation*}
    \lim_{|\dT_{2}| \to\infty} \ldots \lim_{|\dT_{r-1}| \to\infty}
        \hbeta{r}{t} = \hmus\left(\dT_{r}; f' \right) \,,
    \end{equation*}
    where $f'(u,v) = \ind{\rho(v) = t}$ and $\hmus$ is from \Cref{lem.rwte}, from which we also use the consistency guarantee to show that under an \EPART, \Cref{eq.inductive.claim.beta} holds as
\begin{equation*}
        \lim_{|\dT_{2}| \to\infty} \ldots \lim_{|\dT_{r}| \to\infty}
        \hbeta{r}{t} \asc \sum_{(u,v) \in \cE_{r}} f'(u,v)  = |\cB{r}{t}|\,.
    \end{equation*}
    Because of \Cref{prop.tour.stitch}, concatenating tours $\bX \in \dT_{q}$ yields a sample path from $\bPhi_{r}$, and these samples are distributed according to $\pi_{\bPhi_{r}}$ as $|\dT_{r'}| \to\infty$, $r'\leq r$.
    Therefore,
    \begin{equation*}
        \lim_{|\dT_{2}| \to\infty} \ldots \lim_{|\dT_{r}| \to\infty}
        \uplus_{\bX \in \dT_{q}}\uplus_{j=2}^{|\bX|}
        \left\{X_{j}  \colon \rho(X_{j}) = t \right\}
        \sim \pi'_{\bPhi_{r}} \,,
    \end{equation*}
    where $\pi'_{\bPhi_{r}}(u) \propto  \ind{\rho(u) = t} \dg_{\cG_{r}}(u)$, which is equivalent to $\unif(\cB{r}{t})$ by \Cref{def.gr,def.border.sets}, thus proving \Cref{eq.inductive.claim.U}.
\end{proof}
\begin{proof}[Main Theorem]  
Combining \Cref{prop.inductive.claim} and \Cref{prop.perfect.sampling} proves \Cref{thm.estimator}.
\end{proof}

\subsection{Proof of \Cref{thm.bias.propagation}}
\label{proof.thm.bias.propagation}
\begin{definition}[$L^{2}$ Distance between $\widehat{\pi}$ and $\pi$ \citep{aldous-fill-2014} ]
	\label{def.l2.dist}
	The $L^{2}$ distance between discrete probability distribution $\widehat{\pi}$ and reference distribution $\pi$ with sample space $\Omega$ is given by 
	$\|\widehat{\pi} - \pi\|_2 = \sum_{i \in \Omega} \frac{(\widehat{\pi}(i) - \pi(i))^{2}}{\pi(i)}$.
\end{definition}

\begin{definition}[Distorted chain]
    \label{def.distorted}
    Given a Markov chain $\bPhi$ over finite state space $\cS$ and an arbitrary $\xo \in \cS$, let $\hPhi$ be the distorted chain such that $\forall\,u\neq\xo$, $p_{\hPhi}(u,\cdot) = p_{\bPhi}(u,\cdot)$, and $p_{\hPhi}(\xo,\cdot)$ is an arbitrary distribution with support $\supp(p_{\hPhi}(\xo,\cdot))\subseteq \supp(p_{\bPhi}(\xo,\cdot))$.
    The distortion is given by $\|p_{\hPhi}(\xo,\cdot) - p_{\bPhi}(\xo,\cdot)\|$ as defined in \Cref{def.l2.dist}.
\end{definition}

\begin{lemma}
    \label{lemma.p.ratio}	
    Given a finite state, positive recurrent Markov chain $\bPhi$ over state space $\cS$,
    let $\hPhi$ be the chain distorted at some $\xo \in \cS$ from \Cref{def.distorted}.
    Let 
    \begin{equation*}
        \cX = \left\{(X_{1}, \ldots, X_{\xi}) \colon X_{1} = \xo\,,\, \xi = \min\{t>0 \colon X_{t+1} = \xo\} \,,\, p_{\bPhi}(X_{1}, \ldots, X_{\xi})>0\right\} \,,
    \end{equation*}
    denote the set of all possible arbitrary lengths {\RWT}s that begin and end at $\xo$ from \Cref{def:rwt}.
    Given a tour $\bY \in \cX$ sampled from $\bPhi$ and a bounded function $F\colon\cX \to \RR$,
	\begin{equation}
		\EE_{\bPhi}\left[ \frac{p_{\hPhi}(Y_1,Y_2)}{p_{\bPhi}(Y_1,Y_2)}F(\bY) \right] = \EE_{\hPhi}\left[ F(\bY) \right] \,,
	\end{equation}
	where $\EE_{\bPhi}$ and $\EE_{\hPhi}$ are expectations under the distribution of tours sampled from $\bPhi$ and $\hPhi$.\end{lemma}
\begin{proof}
    All tours in $\cX$ are of finite length because of the positive recurrence of $\bPhi$.
    The ratio of the probability of sampling the tour $\bY = (Y_1, \ldots, Y_{\xi'})$ from the chain $\hPhi$ to $\bPhi$ is given by 
    \begin{equation}
        \label{eq.ratio.tour.wt}
		\frac{p_{\hPhi}(\bY)}{p_{\bPhi}(\bY)} = \frac{\prod_{j=1}^{\xi'} p_{\hPhi}(Y_{j}, Y_{j+1})}{\prod_{j=1}^{\xi'} p_{\bPhi}(Y_{j}, Y_{j+1})}
		= \frac{p_{\hPhi}(Y_1, Y_2)}{p_{\bPhi}(Y_1, Y_2)}  \,,
	\end{equation}
    because $p_{\bPhi}(Y_{j}, \cdot) = p_{\hPhi}(Y_{j}, \cdot)$, $\forall 1<j\leq\xi'$ because $Y_j \neq \xo$ by the definitions of $\cX$ and $\hPhi$.
    Because $\supp(p_{\hPhi}(\xo,\cdot)) \subseteq \supp(p_{\bPhi}(\xo,\cdot))$, $\supp(p_{\hPhi}(\bY)) \subseteq \supp(p_{\bPhi}(\bY))$. 
	The theorem statement therefore directly draws from the definition of importance sampling~\citep[Def~3.9]{robert2013monte} with the importance weights derived in \Cref{eq.ratio.tour.wt}.
\end{proof}

\begin{lemma}
    \label{lem.general.bias.variance}
    Given a simple random walk $\bPhi$ on the connected non-bipartite graph $\cG$ from \Cref{def.simplerw}, let $\hPhi$ be the chain distorted at some $\xo \in \cS$ from with distortion $\nu$ \Cref{def.distorted}.
    Let $\lambda = \nicefrac{\hdg(\xo)}{\dg(\xo)}$.
    Let $f\colon \cE \to \RR$ bounded by $B$, and $F(\bX) = \sum_{j=1}^{|\bX|}f(X_{j}, X_{j+1})$, where $\bX$ is an \RWT as defined in \Cref{proof.lem.rwte}.
    The bias of an \RWTE (\Cref{eq.rwte}) computed using tours sampled over $\hPhi$ and using $\hdg(\xo)$ as the degree is given by 
    \begin{equation*}
        \bias = 
        \left|\EE_{\hPhi}\left[\frac{\hdg(\xo)}{2}F(\bX)\right] - \mu(\cE)\right|
        \leq 
        \left( \lambda\nu + |1-\lambda| \right)
        \frac{\sqrt{3} B|\cE|}{\sqrt{\delta}}\,,
    \end{equation*}
    where $\delta$ is the spectral gap of $\bPhi$, and $B$ is the upper bound of $f$.
\end{lemma}
\begin{proof}
    From \Cref{lemma.p.ratio} and \Cref{lem.rwte} we have, respectively,
    \begin{align*}
            \EE_{\hPhi}\left[\frac{\hdg(\xo)}{2}F(\bX)\right] 
        &= \EE_{\bPhi}\left[ \frac{\hdg(\xo)}{2} \frac{p_{\hPhi}(X_1,X_2)}{p_{\bPhi}(X_1,X_2)}F(\bX) \right] \,,\\
        \mu(\cE) &= 
    \EE_{\bPhi}\left[ \frac{\dg(\xo)}{2} F(\bX) \right] \,.
    \end{align*}	
    Subtracting the two, squaring both sides and using the Cauchy-Schwarz inequality decomposes the squared bias into
    \begin{align*}
        \bias
    &= \left|\EE_{\bPhi}\left[ 
    \left( 
        \frac{\hdg(\xo)}{\dg(\xo)} \frac{p_{\hPhi}(X_1,X_2)}{p_{\bPhi}(X_1,X_2)} -1
     \right)    
    \frac{\dg(\xo)}{2} F(\bX)
    \right]\right| \,.\\
            \bias^{2}&\leq 
    \underbrace{\EE \left[  \left(\frac{\hdg(\xo)}{\dg(\xo)}\frac{p_{\hPhi}(\xo,X_2)}{p_{\hPhi}(\xo,X_2)}-1\right)^{2} \right] }_{\bias_{\text{dist}}}
    \underbrace{\EE \left[  \left(\frac{\dg(\xo)}{2} F(\bX)\right)^{2} \right]}_{\bias_{\text{spectral}}}
    \,,
    \end{align*}
    where the expectation is under $\bPhi$. 
    Using definitions from the theorem statement,
    \begin{equation*}
        \begin{split}
            \bias_{\text{dist}}
            =&
            \frac{\hdg(\xo)^{2}}{\dg(\xo)^{2}}
            \EE \left[\left( \frac{p_{\hPhi}(\xo,X_2)}{p_{\bPhi}(\xo,X_2)} \right)^{2}\right]
            +
            1-2
            \frac{\tdg(\xo)}{\dg(\xo)}
            \EE\left[\frac{p_{\hPhi}(\xo,X_2)}{p_{\bPhi}(\xo,X_2)}\right] \\
            =&
            \lambda^{2}(1+\nu^{2}) +1-2\lambda
            =
            \lambda^{2}+\lambda^{2}\nu^{2} +1-2\lambda\\
            =& \lambda^{2}\nu^{2} + (1-\lambda)^{2}
            \leq
            \left( \lambda\nu + |1-\lambda| \right)^{2}
            \,.
        \end{split}
    \end{equation*}
    Because $F(\bX) \leq B\xi$, the tour length, from \Cref{lemma.tour.length.sq}, we see that 
    \begin{equation*}
        \bias_{\text{spectral}}
        \leq   \frac{\dg(\xo)^{2} B^{2}}{4} \frac{3  }{\pi_{\bPhi}(\xo)^2 \delta} 
        = \frac{3 B^{2} |\cE|^{2}}{\delta}\,,
    \end{equation*}
    and combining both biases completes the proof for $\bias$.
\end{proof}

\begin{proof}[Main Theorem]  
    Note that by linearity of expectations
    \begin{equation*}
        \begin{split}
            \EEX{\hmu\left(\dT_{2:r}; f \right) \big| \dT_{2:r-1} }
        =& \EEX{
            \frac{\hdgvo{r}}{2 |\dT_{r}|}\sum_{\bX \in \dT_{r}}\sum_{j=2}^{|\bX|-1} f(X_{j},X_{j+1})    
} \,,\\
        =& \EE_{\bX \sim \hPhi_r}\left[\frac{\hdgvo{r}}{2} \sum_{j=1}^{|\bX|} f'(X_{j},X_{j+1}) \right] \,,\\
        \end{split}
    \end{equation*}
    where $\bX$ is an \RWT on $\hPhi_r$ that depends on $\dT_{2:r-1}$
    and $f'(u,v) \triangleq \ind{u,v \neq \vo_{r}} f(u,v)$.
    Applying \Cref{lem.general.bias.variance} completes the proof because $\hPhi_{r}$ is a distorted chain by \Cref{def.distorted}.
\end{proof}

\section{Proofs for \Cref{sec.subgraph.counting}}
\label{proofs.sec.subgraph.counting}

\subsection{Proof of \Cref{prop.subgraph.epart}}
\label{proof.prop.subgraph.epart}
\begin{proof}
	From \citet[Thm-3.1]{Wang:2014}, we know that each disconnected component of $G$ leads to a disconnected component in $\cG[k-1]$, and if $\cI_{1}$ contains a subgraph in each connected component, \ref{component.cond} is satisfied.
	We now prove that $\forall \, s\in\cV[k-1]$, if $\rho(s) = r > 1$, $\exists \, s' \in \N(s) \colon \rho(s')<r$ which simultaneously satisfies \ref{interconnect.cond} and \ref{preconnect.cond}.
	
	W.l.o.g. let the vertex with the smallest distance from the seed vertices be denoted by $\hat{u} = \argmin_{u \in V(s)} \dist(u)$.
	When $\dist(\hat{u}) > 0$, there exists $v \in \N_{G}(\hat{u})$ such that $\dist(v) < \dist(\hat{u})$ by the definition of $\dist$. 
	More concretely, $v$ would be the penultimate vertex in the shortest path from the seed vertices to $\hat{u}$.
	Let $v' \neq \hat{u}$ be a nonarticulating vertex of $s$, which is possible because any connected graph has at least 2 nonarticulating vertices.
	Let $s_1 = \sg{G}{V(s)\backslash \{v'\} \cup \{v\}} \in \cV[k-1]$.
	Now, $\rho(s_1) < \rho(s)$ because $v'$ has been replaced with a vertex at necessarily a smaller distance and because the indicator in the definition of $rho$ will always be $0$ in this case.
	Moreover, $\comp{s_1}{s} = \sg{G}{V(s) \cup \{v\}}\in \cV[k]$, and hence an edge exists between the two.

	When $\dist(\hat{u}) = 0$, there exists $v \in \N_{G}(\hat{u})$ such that $\dist(v)=0$. 
	There exists a nonarticulating $v' \in V(s)\backslash \Vstar$ because otherwise $\Vstar$ would have been disconnected.
	Observing that $\dist(v') + \ind{v' \in V(\cI_{1}) \backslash \Vstar)} >0$ completes the proof of ergodicity.
\end{proof}

\subsection{Proof of \Cref{prop:r-sampling}}
\label{proof.prop:r-sampling}
\begin{proof}[Sampling Probability]
	Consider the lines \Cref{alg:line:remove,alg:line:anchor,alg:line:add}.
	The probability of sampling the pair $(u,v)$ from $ V(s) \times \N_{G}\left(V(s)\right)$ is given by
	\begin{equation*}
		\begin{split}
			\prob(u,v) &= \sum_{a \in V(s)\backslash \{u\}}\prob(v|a,u)\prob(a|u)\prob(u)  \\
			&= \sum_{a \in V(s)\backslash \{u\}} \frac{\ind{v \in \N(a)}}{\dg(a)} \frac{\dg(a)}{\deg_s - \dg(u)} \frac{\deg_s - \dg(u)}{(\kdesc-1)\deg_s} \\
			&\propto \sum_{a \in V(s)\backslash \{u\}} \ind{v \in \N(a)} = |N(v) \cap V(s) \backslash \{u\}| = \bias \,,
		\end{split}
	\end{equation*}
	where $\bias$ is defined in \Cref{alg:line:bias2} and corrected for in \Cref{alg:line:rej1}.
	After the rejection, therefore, $(u,v) \sim \unif( V(s) \times \N_{G}\left(V(s)\right))$.
	
    \Cref{alg:line:rej2} constitutes an importance sampling with unit weight for pairs $(u,v)$, where removing $u$ from and adding $v$ to $V(s)$ produces a \CIS[\kdesc] and zero otherwise.
    In \Cref{alg:line:rej2}, because removing a nonarticulating vertex and adding another vertex to $s$ cannot lead to a disconnected subgraph, we can avoid a DFS when $u \notin \cA_{s}$.
	This completes the proof.
\end{proof}
\begin{proof}[Time Complexity]
	Assuming access to a precomputed vector of degrees, the part up to \Cref{alg:line.ap} is $\Order(\kdesc^{2})$.
	In each proposal, \Cref{alg:line:remove,alg:line:anchor} are $\Order(\kdesc)$, and \Cref{alg:line:add} is $\Order(\Delta_s)$.
	\Cref{alg:line:bias2} is $\Order(\kdesc)$, and the expected complexity of \Cref{alg:line:rej2} is $\Order(\kdesc^{2} \, \nicefrac{|\cA_{s}|}{\kdesc})$ because in expectation only $\nicefrac{|\cA_{s}|}{\kdesc}$ graph traversals will be required.
	The acceptance probability is $\geq \nicefrac{1}{\kdesc}$ is \Cref{alg:line:rej1} and $\geq \frac{\kdesc-|\cA_{s}|}{\kdesc}$.
	The expected number of proposals is therefore $\leq \frac{\kdesc^{2}}{\kdesc-|\cA_{s}|}$.
	As such, the expected time complexity is $\Order(\kdesc^2 (1 +\frac{\Delta_s + \kdesc |\cA_{s}|}{\kdesc-|\cA_{s}|} ))$.
\end{proof}

\section{Additional Implementation Details}
\label{sec.name.summary.alg}
\subsection{Parallel Sampling with a Reservoir Matrix.}
\label{sec.res.sampl}
Given a reasonably large $\rvrSize$ and the number of strata $R$, we initialize an upper triangular matrix of empty reservoirs $[\hbU{r}{t}]_{2\leq r < t \leq R}$ and a matrix of atomic counters $[\hrvrSize{q}{r}]_{2\leq r < t\leq R}$ initialized to $0$.
In each stratum $r$, while being sampled in parallel whenever a tour enters the $t$-th stratum, $\hrvrSize{r}{t}$ is incremented, and with a probability $\min(1, \nicefrac{\rvrSize}{\hrvrSize{r}{t}})$, the state is inserted into a random position in the reservoir $\hbU{r}{t}$ and rejected otherwise.
The only contention between threads in this scheme is at the atomic counter and in the rare case where two threads choose the same location to overwrite, wherein ties are broken based on the value of the atomic counter at the insertion time, guaranteeing thread safety. 
The space complexity of a reservoir matrix is therefore $O(R^2 \rvrSize)$.

A toy example of this matrix is presented in \Cref{fig:rsv}, where $R=5$, and the {\RWT}s are being sampled on the graph stratum $\cG_{2}$.
Whenever (non-gray) states in $\cI_{3:5}$ are visited, they are inserted into the corresponding reservoirs–$\hbU{2}{5}$ is depicted in detail.

\begin{figure}\centering
	\begin{minipage}[t]{0.45\linewidth}
		\centering
		\includegraphics[width=\linewidth]{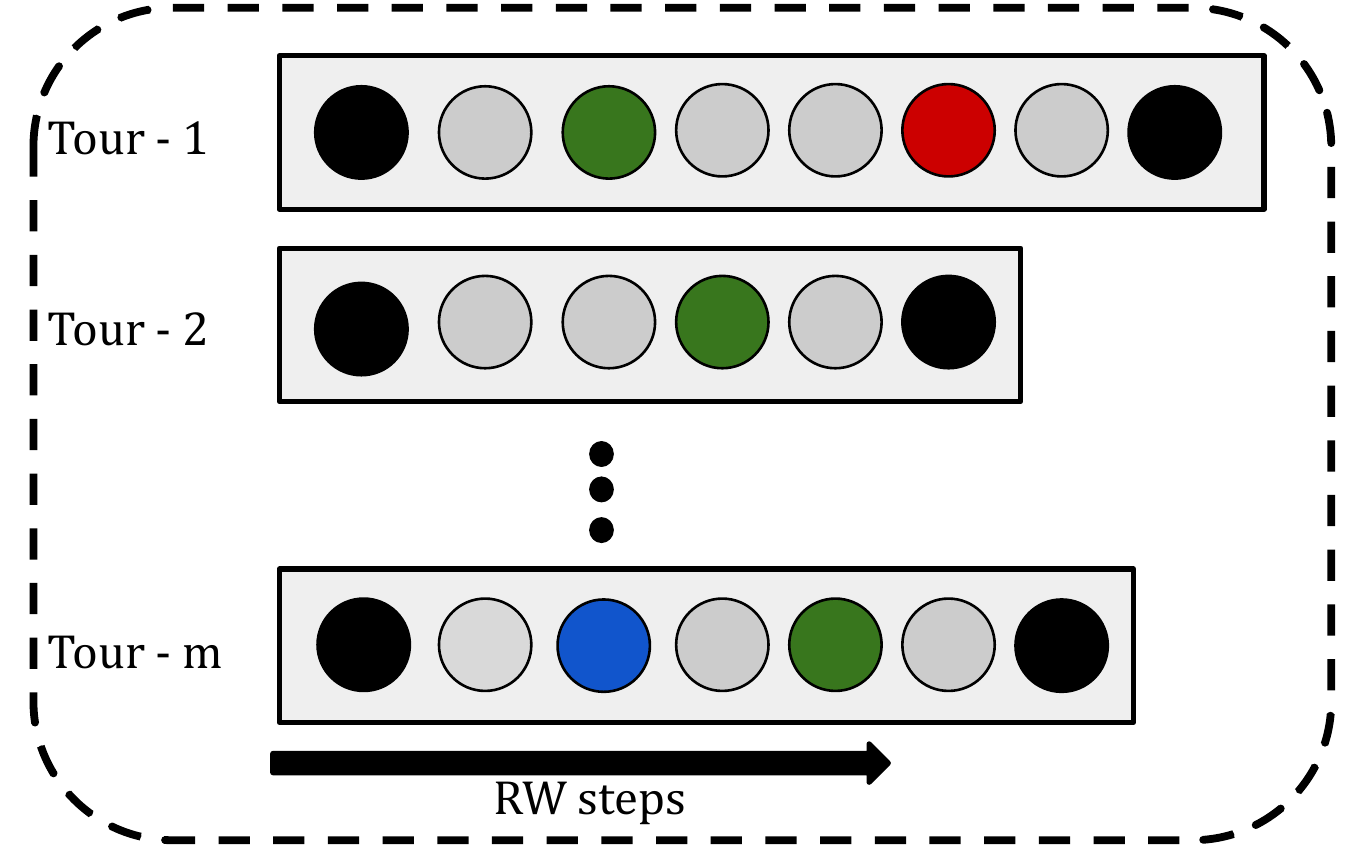}	
		\subcaption{$\cT_{2}$, the set of {\RWT}s sampled in $\cG_{2}$.} 
		\label{fig:rsv.tours}
	\end{minipage}
	\begin{minipage}[t]{0.45\linewidth}
		\centering
		\includegraphics[width=\linewidth]{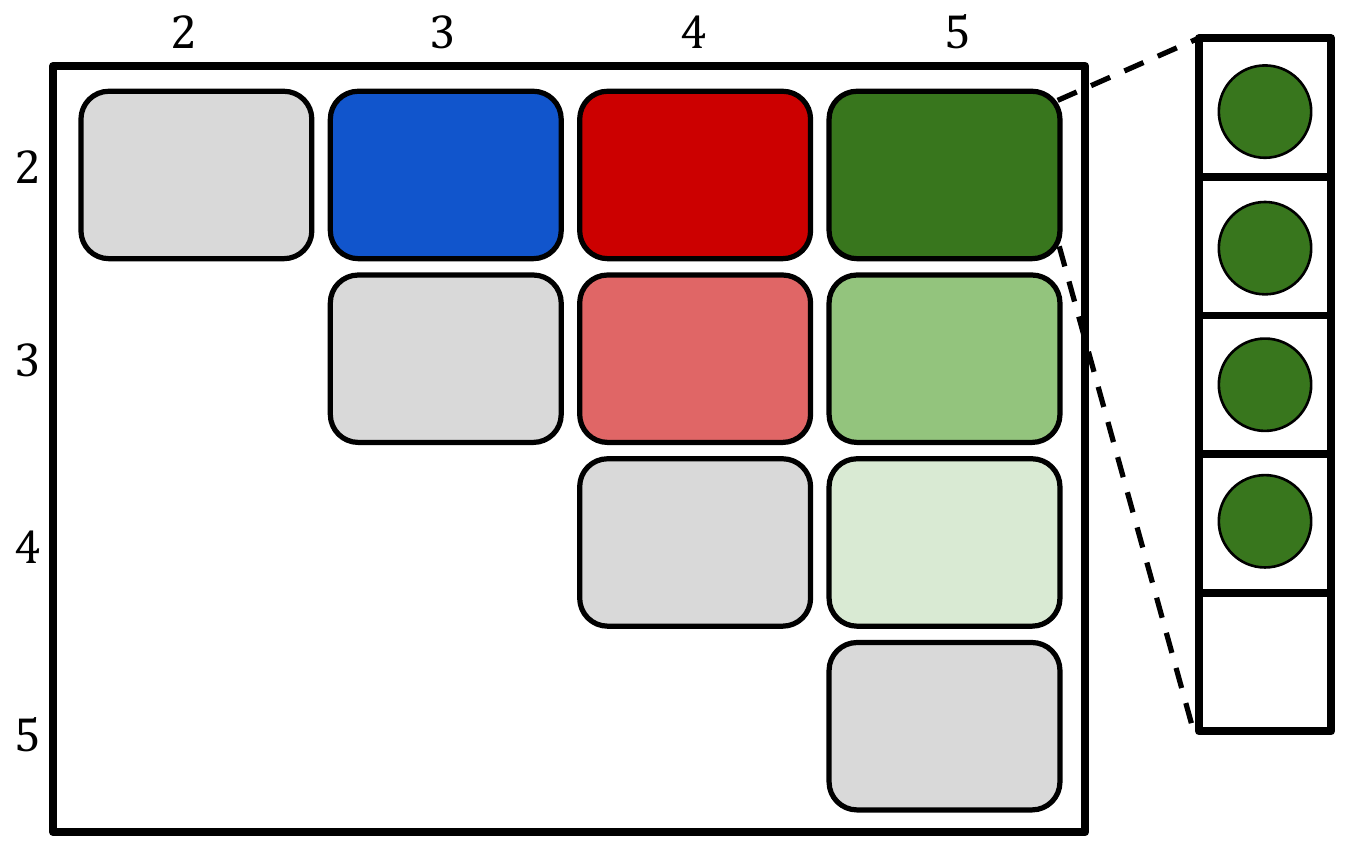}	
		\subcaption{The reservoir matrix, $\hbU{r}{t}$ for $2 \leq r < t \leq R$.}
		\label{fig:rsv.matrix}
	\end{minipage}
	\caption{
		Parallel {\RWT}s and Reservoirs: 
		\Cref{fig:rsv.tours} shows the set of $m$ {\RWT}s sampled on $\cG_{2}$ in parallel, where the supernode $\vo_{2}$ is colored black.
		The gray, blue, red and green colors represent states in stratum $2$--$5$, respectively.
		\Cref{fig:rsv.matrix} shows the upper triangular reservoir matrix in which the cell in the $r$-th row and $t$-th column contains samples from $\hbU{r}{t}$.
	}
	\label{fig:rsv}
\end{figure}

\subsection{\PSRW Neighborhood}
\label{sec.nei.alt.def}
The neighborhood of a \CIS[k] $s$ in $\cG[k]$ is the set of all vertices $u,v \in V$ such that replacing $u$ with $v$ in $s$ yields a \CIS[k].
Formally,
\begin{equation}
	\label{eq.nei.alt.def}
	\N_{\cG[k]}(s) \equiv \left\{ (u,v) \in V(s) \times \N_{G}\left(V(s)\right)  \colon \sg{G}{V(s)\cup \{v\} \backslash \{u\}} \in \cV[k]\right\} \,,
\end{equation}
where $\N_{G}(V(s)) = \cup_{x \in V(s)}\N_{G}(x)$ is the union of the neighborhood of each vertex in $s$.
The size of the neighborhood is then $\Order(k\,\N_{G}(V(s)))\in\Order(k^2\Delta_{G})$ because $\N_{G}(V(s)) \in \Order(k \Delta_{G})$, where $\Delta_{G}$ is the maximum degree in $G$.
Each potential neighbor further requires a connectivity check in the form of a BFS or DFS, which implies that the naive neighborhood sampling algorithm requires $\Order(k^{4} \Delta_{G})$ time.

\subsubsection{Articulation Points}
\label{sec.art.points}
Apart from the rejection sampling algorithm from \Cref{alg:r-sampling}, we use articulation points to efficiently compute the subgraph bias $\gamma$ from \Cref{eq.mu.edge.sum}. 
Specifically, given the \CIS[\kdesc], $s$, $\gamma(s) = \binom{\kappa-\cA_{s}}{2}$, $\cA_{s}$ is the set of articulation points of $s$.
This draws directly from \cite[Sec-3.3]{Wang:2014} and the definition of articulation points.
\citet{hopcroft1973algorithm} showed that for any simple graph $s$ the set of articulation points can be computed in $O(|V(s)| + |E(s)|)$ time.

\begin{algorithm}[ht!] \caption{\name for Subgraph Counting} \label{alg:summary}
	\KwIn{Input graph $G$, Order $k$, Set of subgraph patterns $\mathcal{H}$ of interest}
\KwIn{Initial vertex stratum $\cI_{1}$, Reservoir Size $\rvrSize$ and Error Bound $\epsilon$}
	\KwOut{$\hmu$, an asymptotically unbiased estimate of $\cC[k]$}
	\SetKw{KwOr}{or}
	\SetKwFor{ParallelWhile}{parallel while}{do}{}

	\BlankLine
	\tcc{Initialization}
	$\hmu = 0$, $\hbeta{q}{t} = 0$, $\hbU{q}{t} = \emptyset$, $\forall 1\leq q<t\leq R$\;
    Run BFS for stratification $\rho \colon \cV[k-1] \to \{1,\ldots,R\}$, with $\cI_{1}$ (\Cref{prop.subgraph.epart})\\ 
    
    \BlankLine
	\tcc{Exact computation in the first stratum}
	\ForEach{$u \in \cI_{1}, v \in \N_{\cG[k-1]}(u)$} { \label{alg:first_strata1}
		Update $\hbeta{1}{\rho(v)} \increment 1$ , $\hbU{1}{\rho(v)} \append v$\\
		Update $\hmu \increment \left(\frac{\ig{\comp{u}{v}}}{\gamma(\comp{u}{v})}\right)_{H \in \cH}$ \tcp*{\Cref{eq.mu.edge.sum}} \label{alg:first_strata2}
	}
    
    \BlankLine
    \tcc{Estimate remaining strata}
    \For{$r \in 2,\dots,R$}{ \label{alg:size.strata}
		Initialize $\hmu_{r} = 0$, $m_{r} = 0$ \\
		
		\ParallelWhile{\Cref{eq.auto.tours} is not satisfied}{\label{alg:numtours}
Sample $q$ from $\{1, \ldots, r-1\}$ w.p.\ $\hbeta{q}{r}$ \label{alg:start.tour.1} \\
			Sample $u$ from $\hbU{q}{r}$ \tcp*{\Cref{eq.prwt.stitch.tour}} \label{alg:start.tour.2}
			Sample $v \sim \unif(\N_{\cG[k-1]}(u))$ \tcp*{\Cref{alg:r-sampling}} \label{alg:start.tour.3}
			\While{$\rho(v) \geq r$}{ \label{alg:tourlength}
				Update $\hmu_{r} \increment \left(\frac{\ig{\comp{u}{v}}}{\gamma(\comp{u}{v})}\right)_{H \in \cH}$ \tcp*{\Cref{eq:prwt.partial}}\label{alg.mu.update}

				\If{$\rho(v) > r$}{
					Update $\hbeta{r}{\rho(v)} \increment 1$ \tcp*{\Cref{eq:prwt.deg}}
					Update $\hbU{r}{\rho(v)} \append v$ \tcp*{\Cref{eq.prwt.stitch.tour}}
				}

				$u \coloneqq v$ \\
				\tcc{\Cref{prop.rwt.phi.r} and \Cref{alg:r-sampling}}
				\If{$\rho(u) = r$}{
					Sample $v \sim \unif(\N_{\cG[k-1]}(u))$
				}\Else{
					\While{$\rho(v) \neq r$}{\label{alg:rejection1}
						Sample $v \sim \unif(\N_{\cG[k-1]}(u))$ \label{alg:rejection2} 
					}
				}
			}
			$m_r+=1$\\
		}
		Compute $\hdeg{r} = \sum_{q=1}^{r-1}\hbeta{q}{r}$ \tcp*{\Cref{eq:deg.full}}
		$\hmu \increment \frac{\hdeg{r}}{2\, m_r}\hmu_{r}$ \tcp*{\Cref{eq:prwt.partial,eq:prwt.full}}
		Update $\hbeta{r}{t} \prodincr \frac{\hdeg{r} }{m_r}$, $\forall t>r$  \tcp*{\Cref{eq:prwt.deg}}
	}
	\Return $\hmu$	
\end{algorithm}

\subsection{Proof of \Cref{prop.ripple.complexity.short}}
\label{proof.prop.ripple.complexity.short}
\begin{proposition}[Extended Version of \Cref{prop.ripple.complexity.short}]
	\label{prop.ripple.complexity} 
	We assume a constant number of tours $m$ in each stratum and ignore graph loading.
	The \name estimator of \CIS[k] counts described in \Cref{alg:summary} has space complexity in 
	\begin{equation*}
	    \Order( k^{3} \diam_{G}^{2} \rvrSize + |\cH|) \equiv \widehat{\Order}(k^3 + |\cH|) \,,
	\end{equation*}
	where $\widehat{\Order}$ ignores all factors other than $k$ and $|\cH|$, $\rvrSize$ is the size of the reservoir from \Cref{sec.implementation}, $\diam_{G}$ is the diameter of $G$, and $|\cH|$ is the number of patterns of interest.
	
	The total number of random walk steps is given by $\Order(k^3  m \diam_{G}  \Delta_{G} \rejec)$, where $\rejec$ is the number of rejections in \Cref{alg:rejection1} of \Cref{alg:summary}, $\Delta_{G}$ is the largest degree in $G$, and the total time complexity is $\widehat{\Order}(k^7 + |\cH|)$.
\end{proposition} 
\begin{remark}
    In practice, we adapt the proposals in \Cref{alg:r-sampling} to minimize $\rejec$ using heuristics over the values of $\dist\bdot$ from \Cref{prop.subgraph.epart}.
\end{remark}

\begin{lemma}
    \label{lem.stratum.ret.time} 
    Given a graph stratum $\cG_{r}$ from \Cref{def.gr}, for some $r>1$, define $\alpha_r = \nicefrac{|\{u \in \cI_{r} \colon \N(u) \cap \cI_{1:r-1} \neq \emptyset\}|}{|\cI_{r}|}$ as the fraction of vertices in the $r$-th vertex stratum that share an edge with a previous stratum.
    The return time $\xi_r$ of the chain $\bPhi_{r}$ to the supernode $\vo_r \in \cV_{r}$ follows 
$\EE_{\bPhi_{r}}[\xi_r] \leq \frac{2 \bdg_{r}}{\alpha_{r}}$, where $\bdg_{r}$ is the average degree in $\cG$ of all vertices in $\cI_{r}$.
\end{lemma}
\begin{proof}
    Because $\alpha_r \cI_{r}$ vertices have at least one edge incident on $\vo_r$, $\dg_{\cG_{r}}(\vo_r) \geq \alpha_r \cI_{r}$.
    From \Cref{def.gr}, because all edges not incident on $\cI_{r}$ are removed from $\cG_{r}$, $\Vol(\cG_{r}) \leq 2 \sum_{u \in \cI_{r}} \dg_{\cG}(u)$.
    Therefore, from \Cref{lem.rwte},
    \begin{equation*}
        \EE_{\bPhi_{r}}[\xi_r]  = 
        \frac{\Vol(\cG_{r})}{\dg(\vo_{r})}
        \leq \frac{2 \sum_{u \in \cI_{r}} \dg_{\cG}(u)}{\alpha_r \cI_{r}} 
        = \frac{2 \bdg_{r}}{\alpha_r} \,.
    \end{equation*}
\end{proof}

\begin{proposition}
	\label{prop.subgraph.alpha}
	The \epart from \Cref{prop.subgraph.epart} is such that $\alpha_{r}=1$ for all $r>1$ as defined in \Cref{lem.stratum.ret.time}, and consequently, the diameter of each graph stratum is $\leq 4$.
	The total number of strata $R \in \Order(k\,\diam_{G})$, where $\diam_{G}$ is the diameter of $G$.
\end{proposition}
\begin{proof}
    We show in \Cref{proof.prop.subgraph.epart} that for each vertex $s\in\cV[k-1]$, if $\rho(s) = r > 1$, there exists $s' \in \N(s)$ such that $\rho(s')<r$.
    This implies that $\alpha_{r} = 1$.
	In $\cG_{r}$, therefore, from $\vo_{r}$, all vertices in $\cI_{r}$ are at unit distance from $\vo_{r}$, and vertices in $\N(\cI_{r}) \backslash \cI_{r}$ are at a distance of $2$ from $\vo_{r}$.
    Because no other vertices are present in $\cG_{r}$, this completes the proof of the first part.
    Trivially, $R \leq (k-1) \cdot \max_{u\in V} \dist(u) \in \Order(k \cdot \diam_{G})$.
\end{proof}

\begin{proof}[Memory Complexity]
    From \Cref{alg:summary}, we compute a single count estimate per stratum and maintain reservoirs and inter-partition edge count estimates for each $2\leq q<t \leq R$.
    Because a reservoir $\hbU{q}{t}$ needs $\Order(k \rvrSize)$ space (\Cref{sec.res.sampl}), the total memory requirement is $\Order( R^{2} \, k \rvrSize)$, where $R$ is the number of strata.
    From \Cref{prop.subgraph.alpha}, plugging $R \in \Order(k \diam_{G})$, and because storing the output $\hmu$ requires $\Order(|\cH|)$ memory the proof is completed.
\end{proof}

\begin{proof}[Time Complexity]
    The stratification requires a single BFS $\in \Order(|V|+|E|)$ from \Cref{eps.subgraph}. 
    In \Cref{alg:first_strata1}, the estimation phase starts by iterating over the entire higher-order neighborhood of each subgraphs in $\cI_1$.
    Based on \Cref{sec.art.points}, \Cref{alg:first_strata2} is in $\Order(k^2)$.
    Because the size of the higher-order neighborhood of each subgraph is $\Order(k^2 \Delta_{G})$ from \Cref{sec.nei.alt.def}, the initial estimation phase will require $\Order(|\cI_1| \, k^4 \Delta_{G})$ time.
    
    In all other strata $r=2, \ldots, R$, we assume that $m$ tours are sampled in \Cref{alg:numtours}.  
    Starting each tour (\Cref{alg:start.tour.1,alg:start.tour.2,alg:start.tour.3}) requires order of magnitude $R$ time, leading to a total time of $\Order(m\,R^2) \in \Order(m k^2 \diam_{G}^2)$ because $R \in \Order(k \diam_{G})$ from \Cref{prop.subgraph.alpha}.
    The total time for these ancilliary procedures is $\Order(m k^2 \diam_{G}^2 + |\cI_1| \, k^4 \Delta_{G})$
    
	Therefore, the time complexity of bookkeeping and setup is $\Order(m k^2 \diam_{G}^2 + |\cI_1| \, k^4 \Delta_{G} + |V| + |E|) \in \widehat{\Order}(k^4)$.
	The time complexity at each random walk step is $\Order(\kdesc^{2}\Delta_{G} + \kdesc^4)\in \widehat{\Order}(k^4)$ from \Cref{proof.prop:r-sampling} and \Cref{sec.art.points}. 
	We assume that the expected number of rejections in \Cref{alg:rejection1} is given by $\rejec$.
    The total number of random walk steps is given by $\Order(R\,m\,\rejec)$ times the expected tour length.
    By \Cref{lem.stratum.ret.time,prop.subgraph.alpha}, the expected tour length is $\Order(\Delta_{\cG[k-1]})\equiv \Order(k^2 \Delta_{G})$.
    Therefore, the total number of random walk steps is $\Order(k^3  m \diam_{G}  \Delta_{G} \rejec)$.
    
    $\Order(|\cH|)$ time is to print the output $\hmu$.
    We assume that updating $\hmu$ is amortized in constant order if we use a hashmap to store elements of the vector, and because updating a single key in said hashmap is by \Cref{eq.mu.edge.sum} increments, the proof is completed.
\end{proof}

\section{Additional Results}
\label{app:additional-results}
We now present the results of additional experiments performed on \name.
\Cref{tab:disp-ripple} shows the dispersion, $\frac{\max-\min}{\mean}$, of the estimates that were used to measure the running time and space utilization of \name in \Cref{sec.scalability}.
\Cref{fig:acc-linf-5} shows the L-$\infty$ norm from the ground truth for $k=5$ with $\rvrSize=10^7$ while $\epsilon$ and $\cI_1$ vary.
\begin{table}\centering
	\scalebox{0.9}{
\begin{tabular}{l|cccc}
 \bottomrule
 \multirow{2}{*}{\textbf{Graph}}  & \multicolumn{4}{c}{\textbf{Rel. dispersion of estimates}}\\
   & $6$ & $8$ & $10$ & $12$\\
   \bottomrule   
          Amazon 
          & $0.203$
          & $0.241$
          & $0.285$
          & $0.268$\\
        DBLP
        & $0.023$
        & $0.023$
        & $0.041$
        & $0.054$\\
        Patents
        & $0.037$
        & $0.083$
        & $0.093$
        & $0.123$\\
        Pokec
        & $0.065$
        & $0.044$
        & $0.037$
        & $0.046$\\
        LiveJ.
        & $0.050$
        & $0.060$
        & $0.033$
        & $0.066$\\
        Orkut
        & $0.021$
        & $0.761$
        & $0.053$
        & $0.031$\\
\bottomrule
 \end{tabular}
}   \caption{Dispersion, $\frac{\max-\min}{\mean}$, of \name's estimates of $|\cV[k]|$ computed using $\epsilon=0.003$, $|\cI_1|=10^4$ and $\rvrSize=10^7$ for $k =6,8,10,12$ used in the analysis in \Cref{sec.scalability}. 
The selected hyper-parameters provide reasonably similar estimates over $10$ independent runs. 
  Orkut for $k=8$ exhibits the largest dispersion because of the presence of a single outlier.}
	\label{tab:disp-ripple}
	\vspace{-10pt}
\end{table} 
\begin{figure}\centering
\begin{minipage}[t]{0.28\linewidth}
	\centering
	\includegraphics[width=\linewidth]{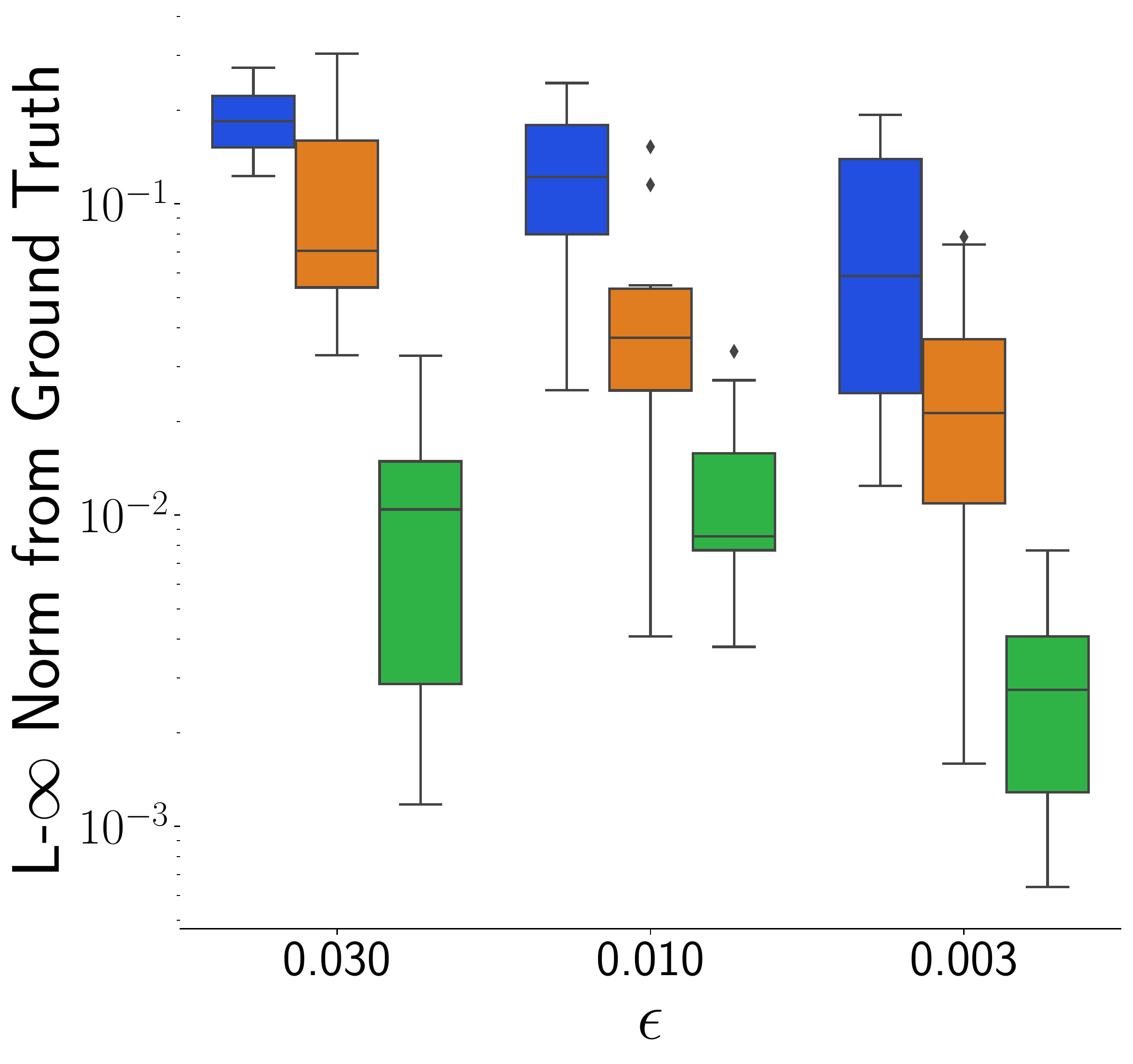}	
	\subcaption{Amazon}
\end{minipage}
\qquad
\begin{minipage}[t]{0.28\linewidth}
	\centering
	\includegraphics[width=\linewidth]{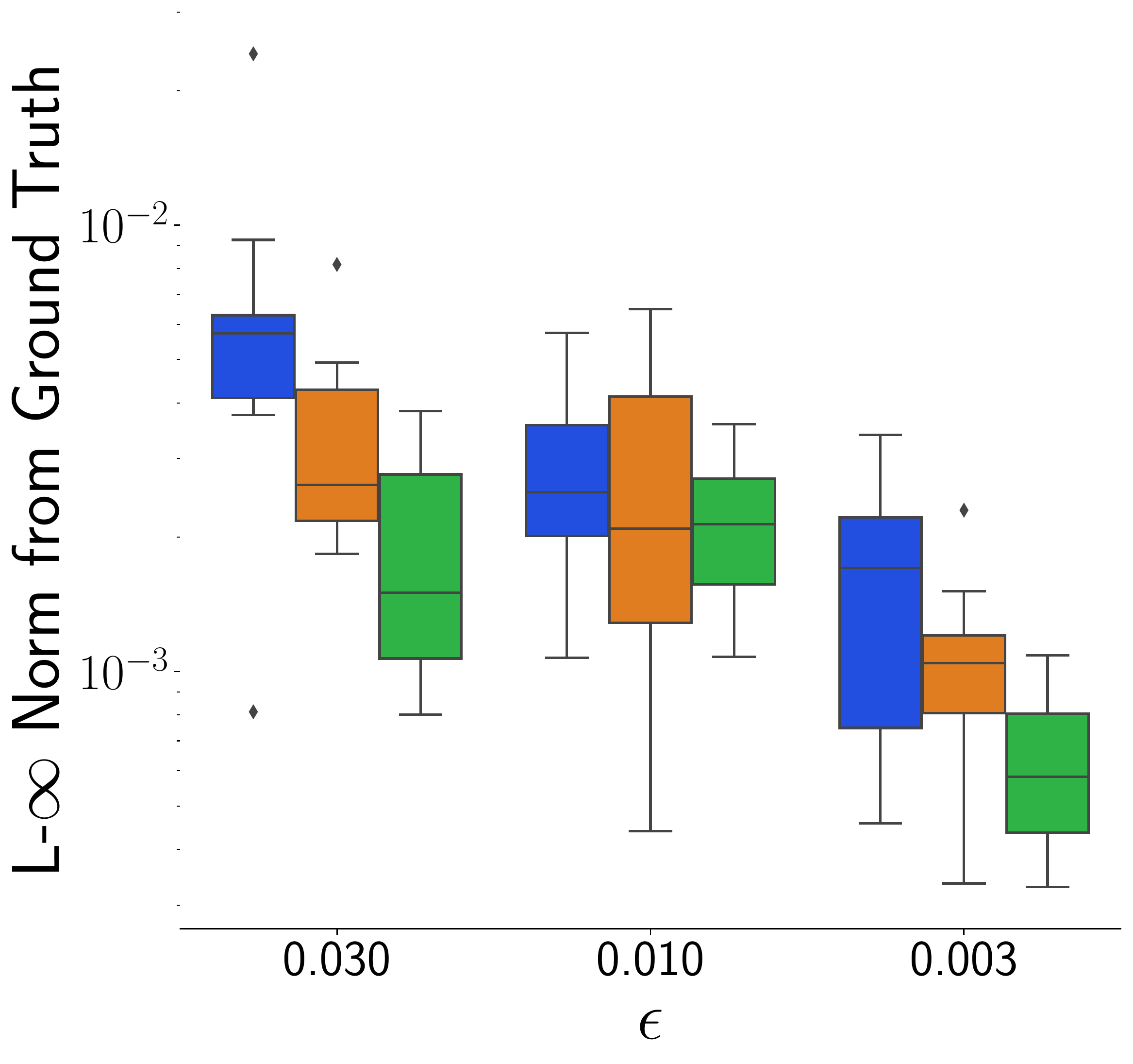}
	\subcaption{DBLP}
\end{minipage}
\qquad
\begin{minipage}[t]{0.28\linewidth}
	\centering
	\includegraphics[width=\linewidth]{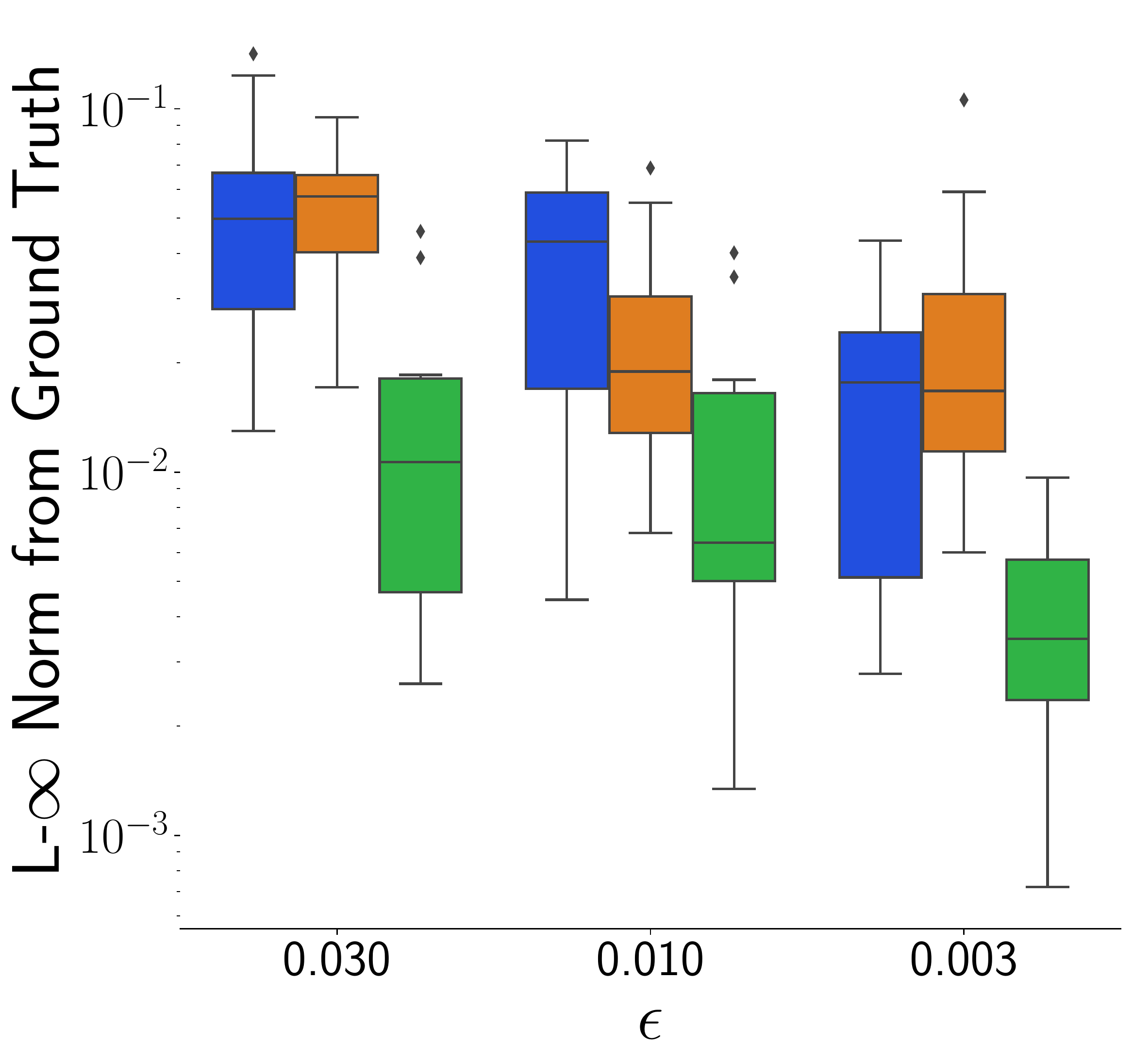}	
	\subcaption{Patents}
\end{minipage}
\qquad
\begin{minipage}[t]{0.28\linewidth}
	\centering
	\includegraphics[width=\linewidth]{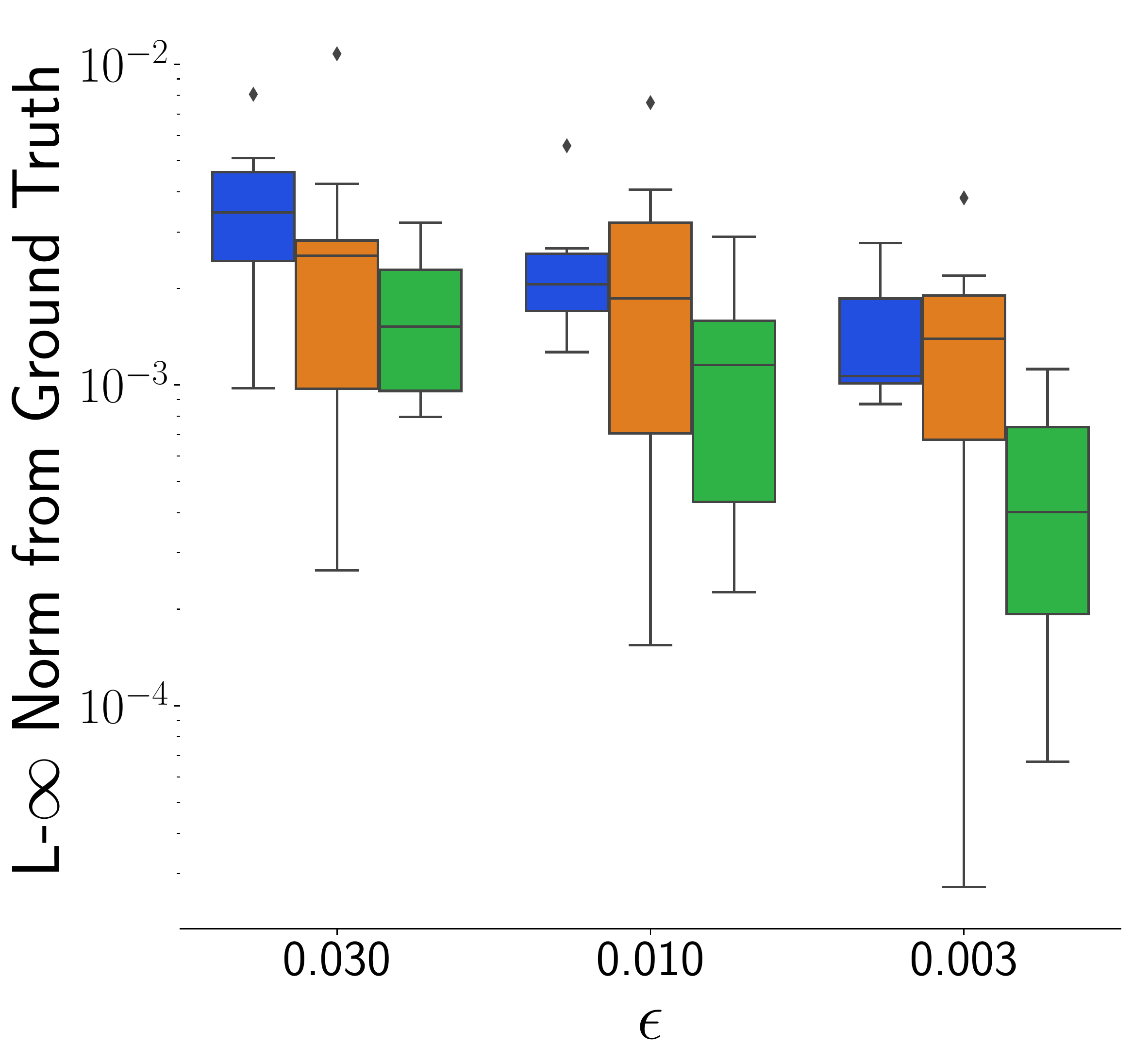}	
    \subcaption{Pokec}
\end{minipage}
\qquad
\begin{minipage}[t]{0.28\linewidth}
	\centering
	\includegraphics[width=\linewidth]{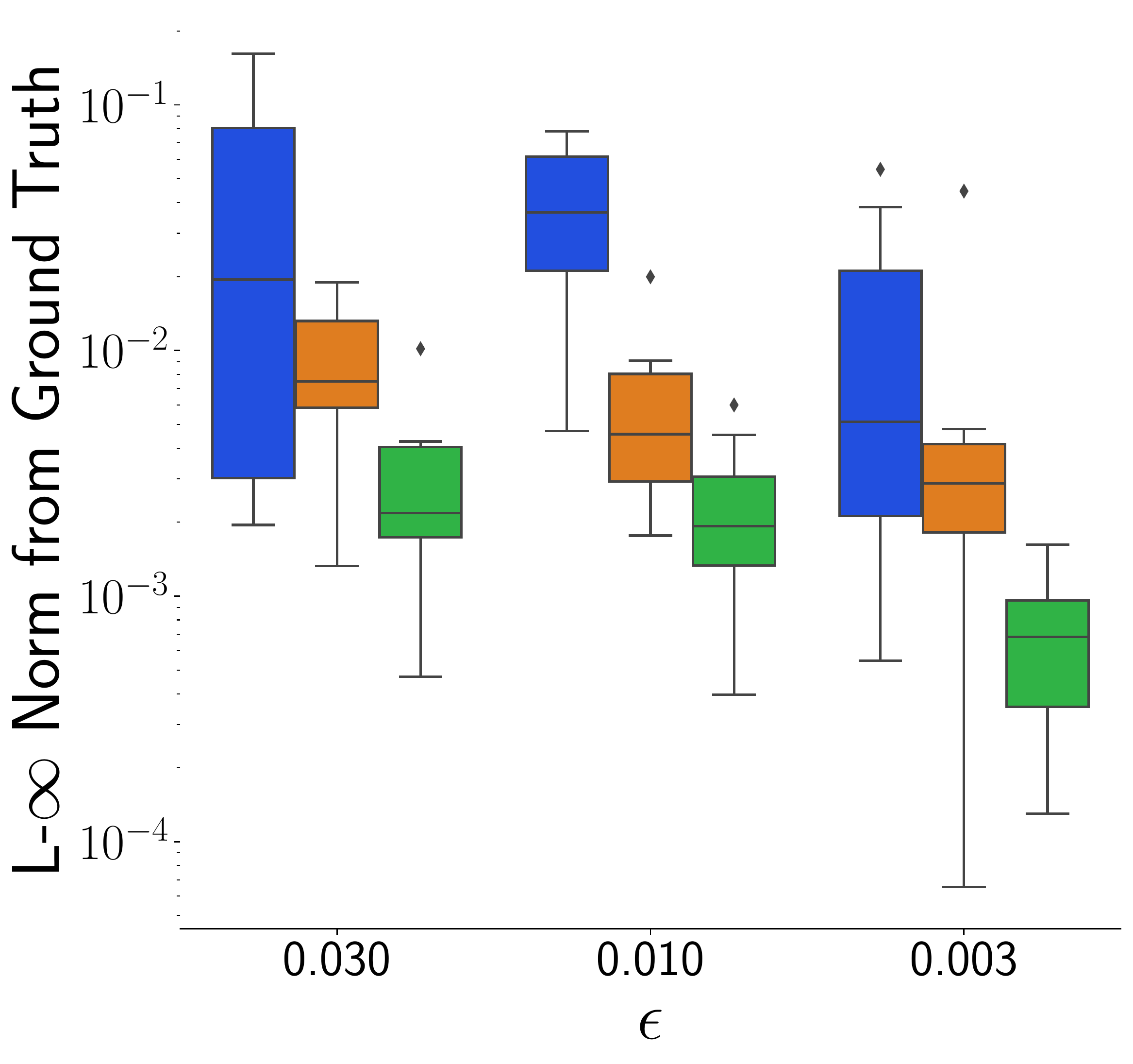}
	\subcaption{Live Journal}
\end{minipage}
\qquad
\begin{minipage}[t]{0.28\linewidth}
	\centering
	\includegraphics[width=\linewidth]{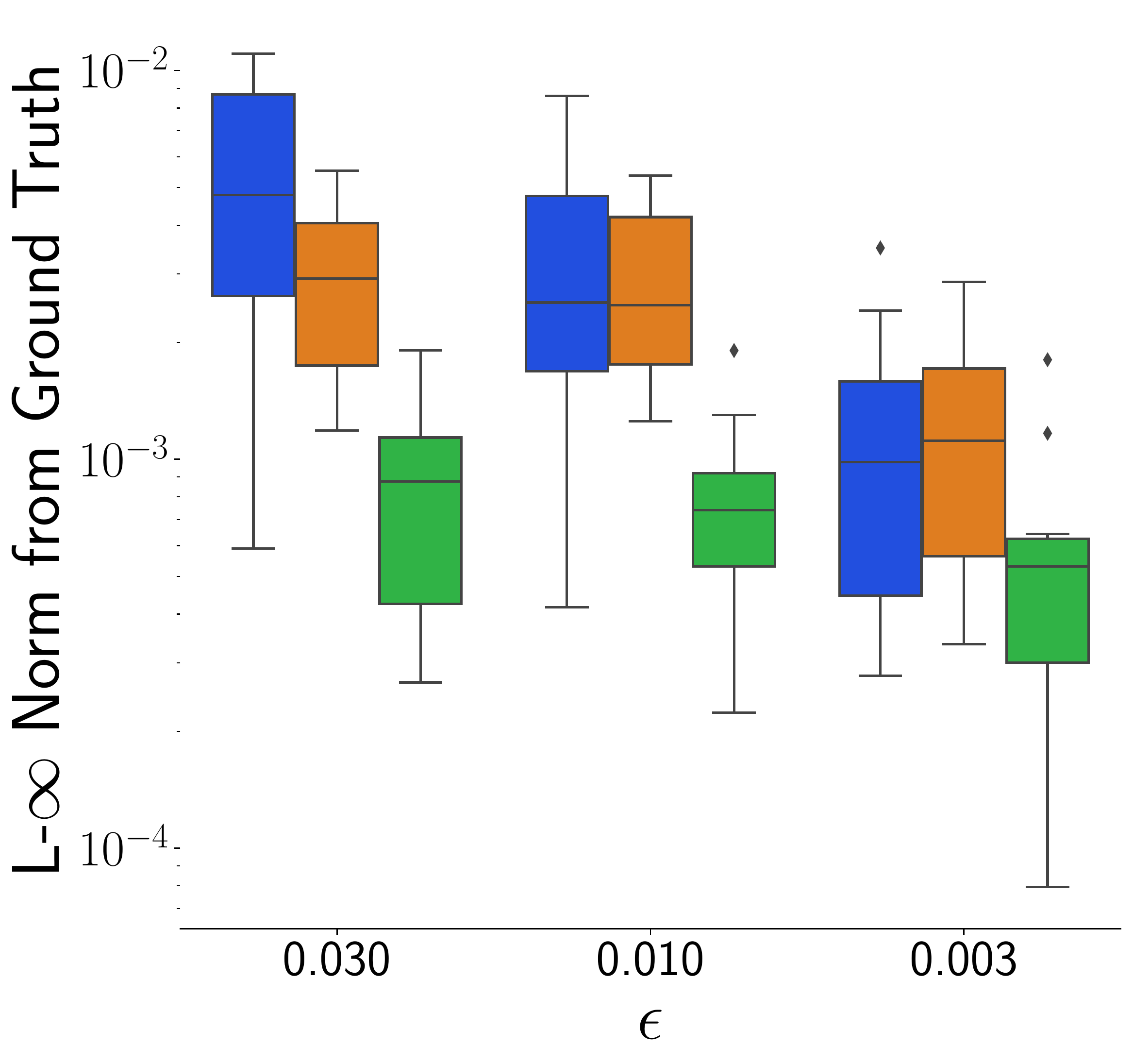}
	\subcaption{Orkut}
\end{minipage}
\caption{Accuracy and convergence analysis for \CIS[5]s using L-$\infty$ norm. 
The $y$-axes show the L-$\infty$ norm between the \name estimate and the ground truth vector $\cC[5]$, containing counts of all possible non-isomorphic subgraph patterns for various settings of $\epsilon$ and $\cI_1$. 
As expected, the accuracy improves as $\epsilon$ decreases and $|\cI_1|$ increases.
Each box and whisker represents $10$ runs.
}
\label{fig:acc-linf-5}
\end{figure}

\begin{figure}[!ht]
\centering
\begin{minipage}[t]{0.29\linewidth}
	\centering
	\includegraphics[width=\linewidth]{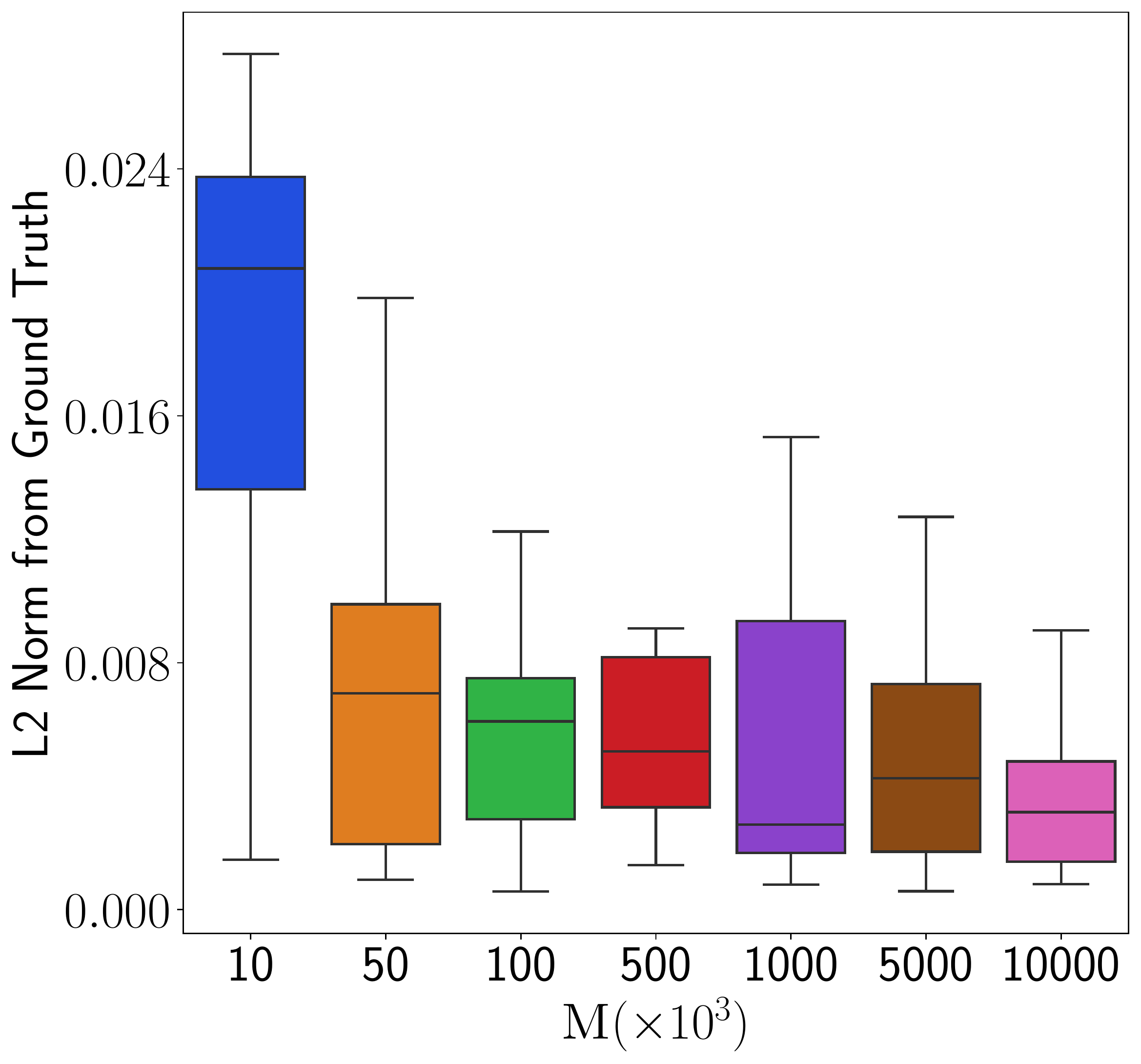}	\subcaption{Amazon}
\end{minipage}
\qquad
\begin{minipage}[t]{0.29\linewidth}
	\centering
	\includegraphics[width=\linewidth]{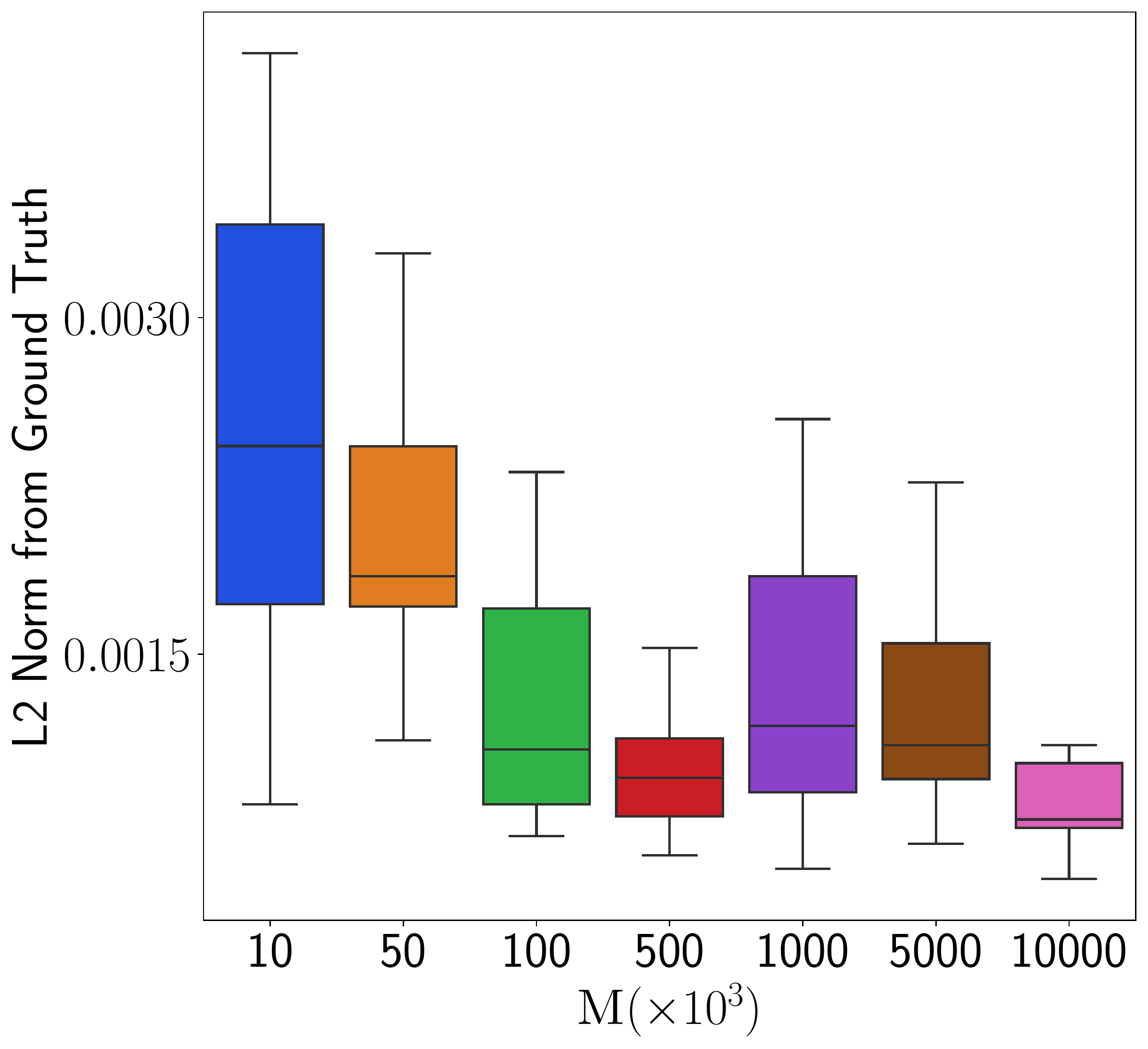}
	\subcaption{DBLP}
\end{minipage}
\qquad
\begin{minipage}[t]{0.29\linewidth}
	\centering
	\includegraphics[width=\linewidth]{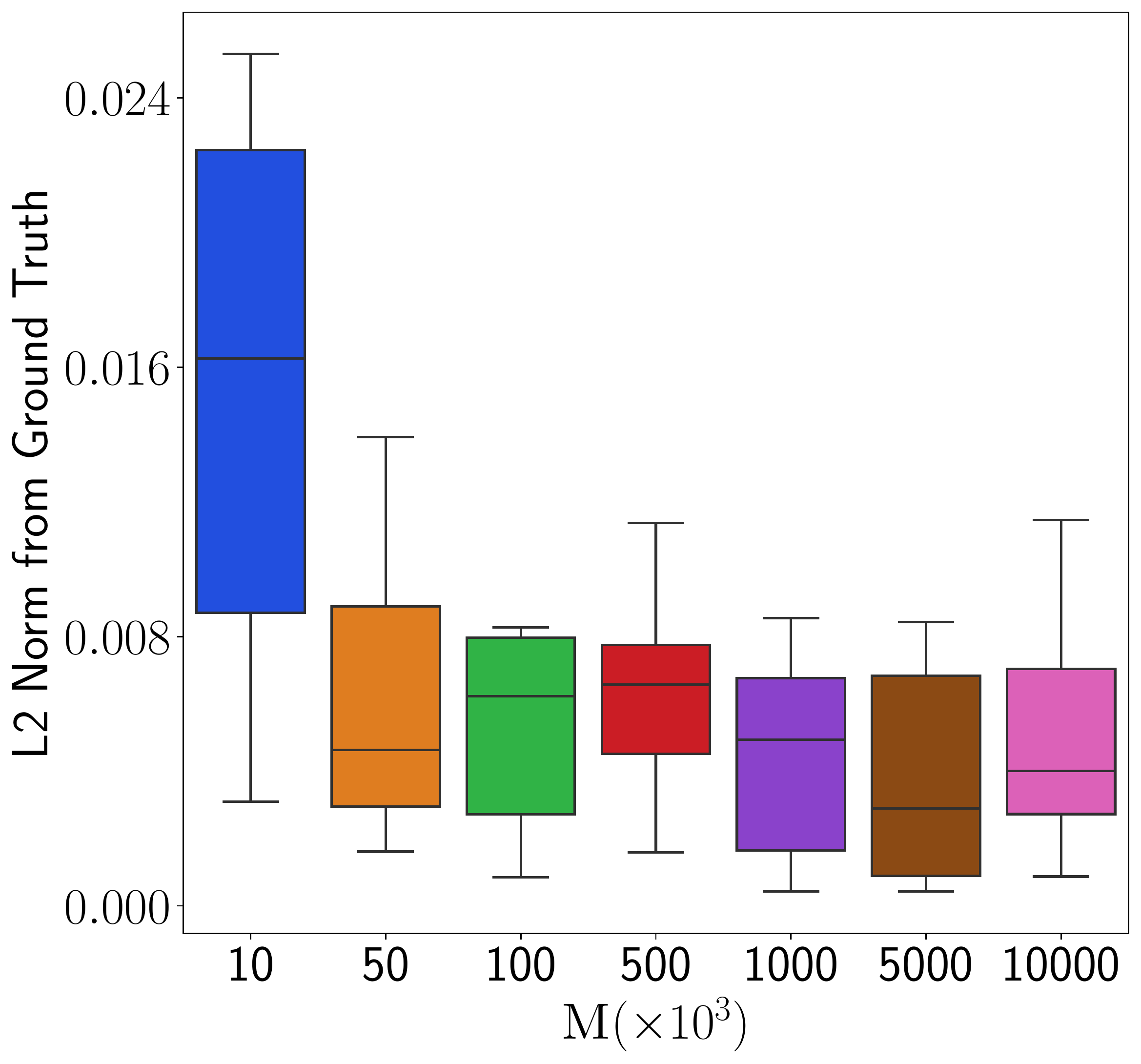}	
	\subcaption{Patents}
\end{minipage}
\qquad
\begin{minipage}[t]{0.29\linewidth}
	\centering
	\includegraphics[width=\linewidth]{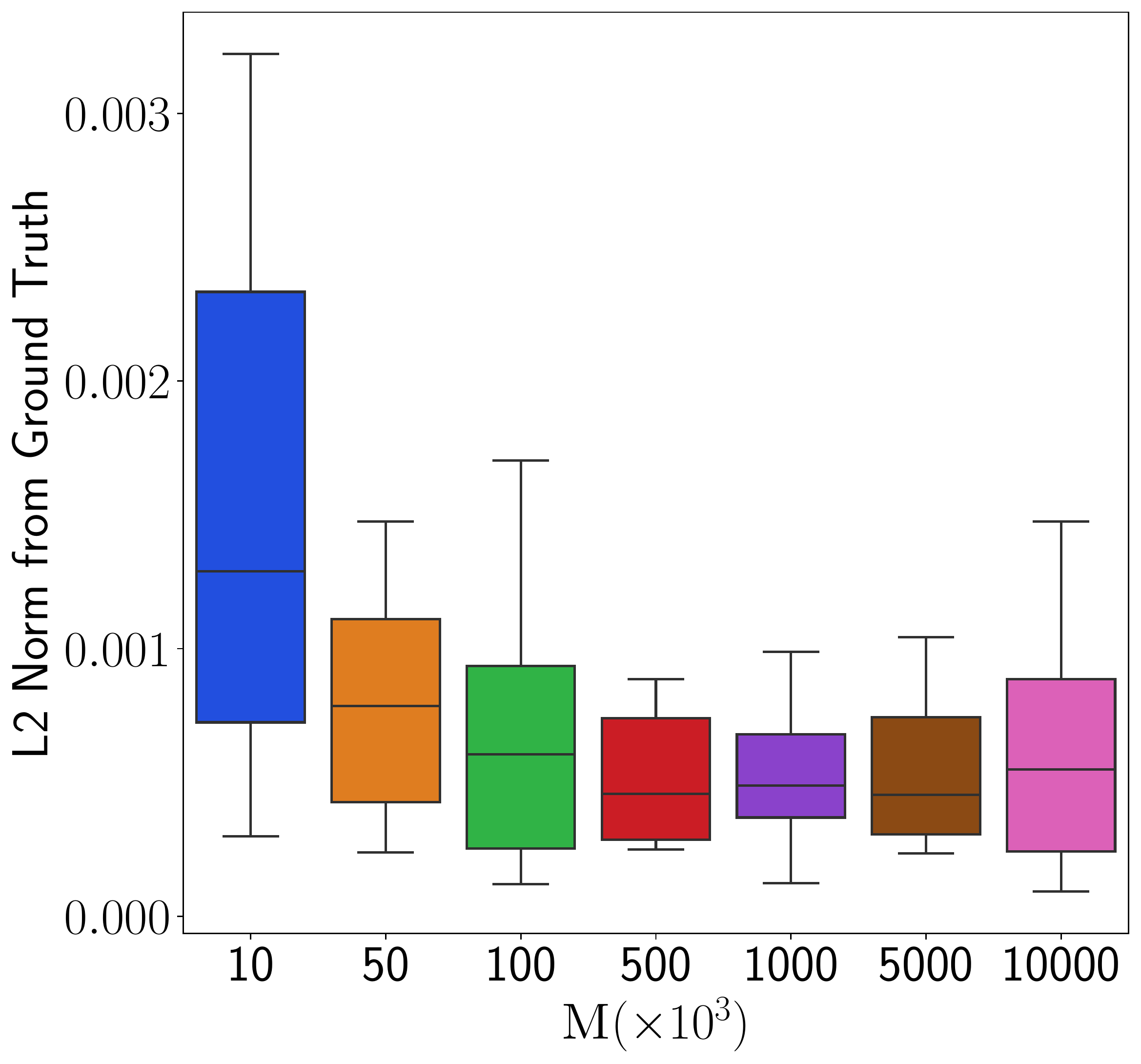}	
    \subcaption{Pokec}
\end{minipage}
\qquad
\begin{minipage}[t]{0.29\linewidth}
	\centering
	\includegraphics[width=\linewidth]{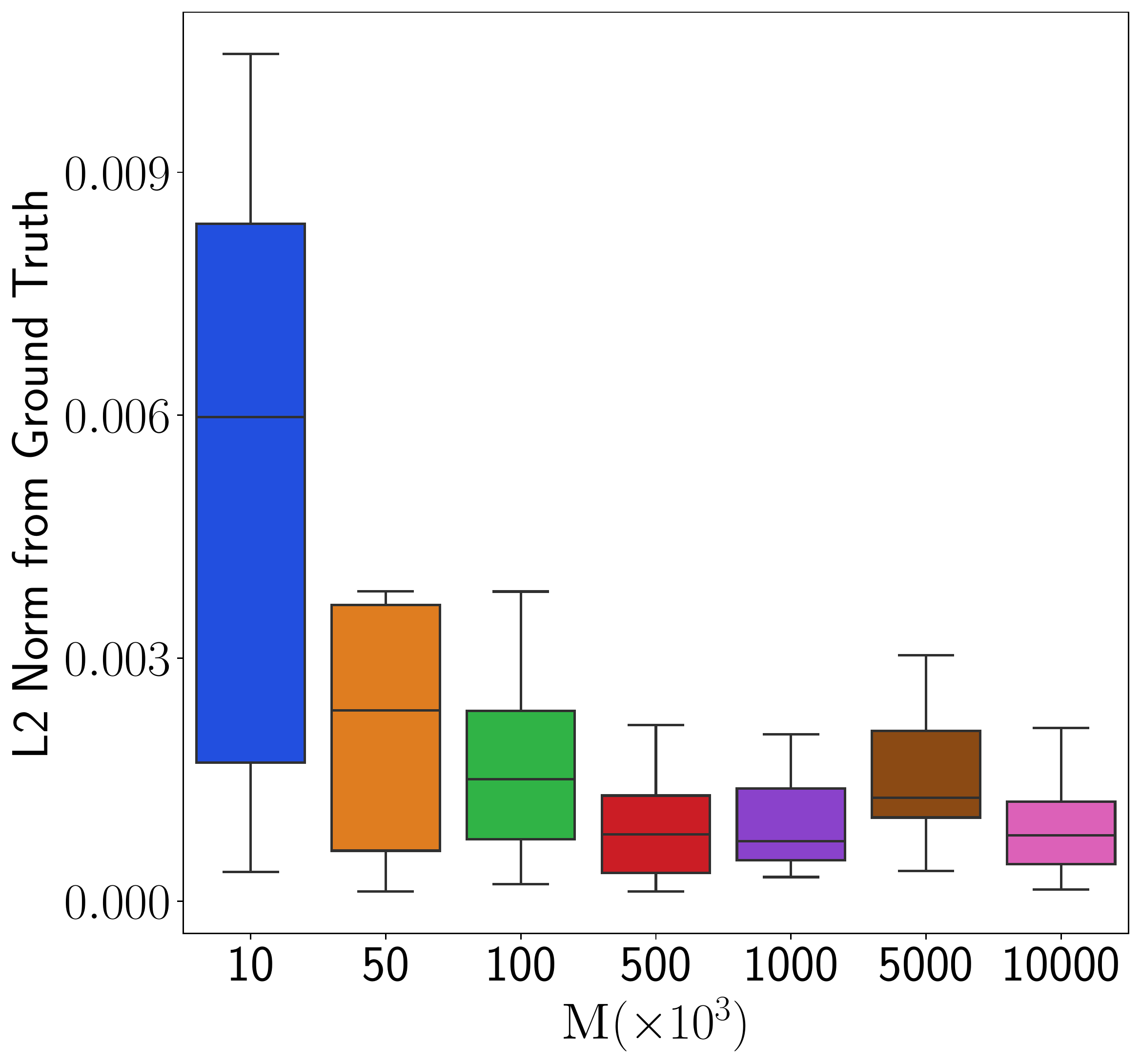}
	\subcaption{Live Journal}
\end{minipage}
\qquad
\begin{minipage}[t]{0.29\linewidth}
	\centering
	\includegraphics[width=\linewidth]{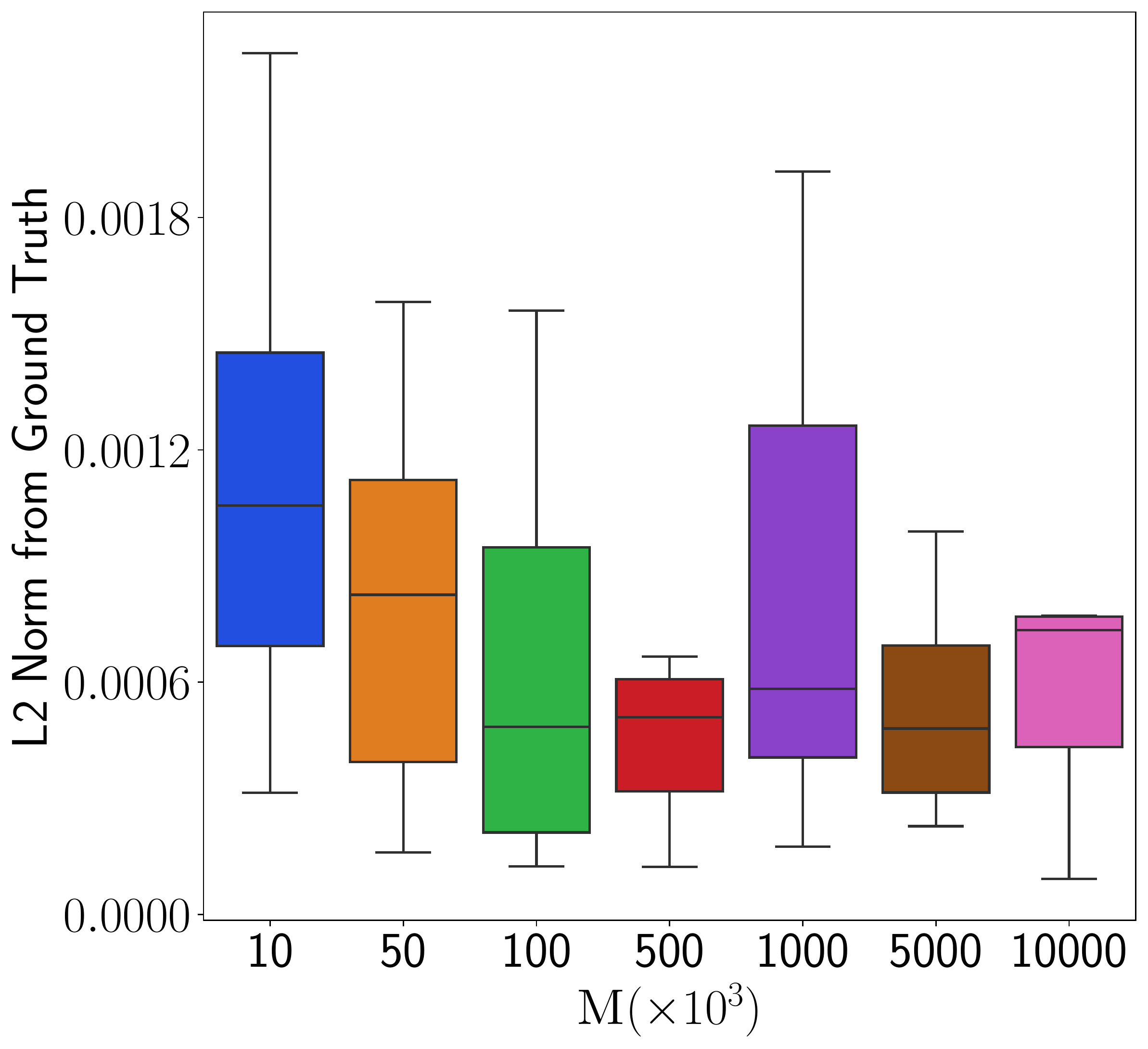}
	\subcaption{Orkut}
\end{minipage}
\caption{Sensitivity of \name to the reservoir capacity $\rvrSize$ for $k=5$. 
We verify that a larger reservoir improves the accuracy of \name estimates in all graphs because it reduces oversampling bias. Each box and whisker plot represents $10$ runs.}
\label{fig:acc-rvr-5}
\end{figure} 
\paragraph{Trade-off between Convergence and Reservoir size.}
Next, we measure the effect of the reservoir capacity $\rvrSize$ on accuracy, as discussed in \Cref{sec.implementation}. 
We vary $\rvrSize$ from $50000$ to $10^7$ while keeping the other parameters fixed as $\epsilon=0.003$ and $|\cI_1|=10^4$ and measure the L2-norm between the \name estimate and the exact value of the count vector $\cC[5]$, such as in \Cref{sec:conv-analysis}.
We see that larger reservoirs reduce oversampling bias and improve the convergence and accuracy in all datasets.

 \begin{figure}[!ht]
\centering
\begin{minipage}[t]{0.29\linewidth}
	\centering
	\includegraphics[width=\linewidth]{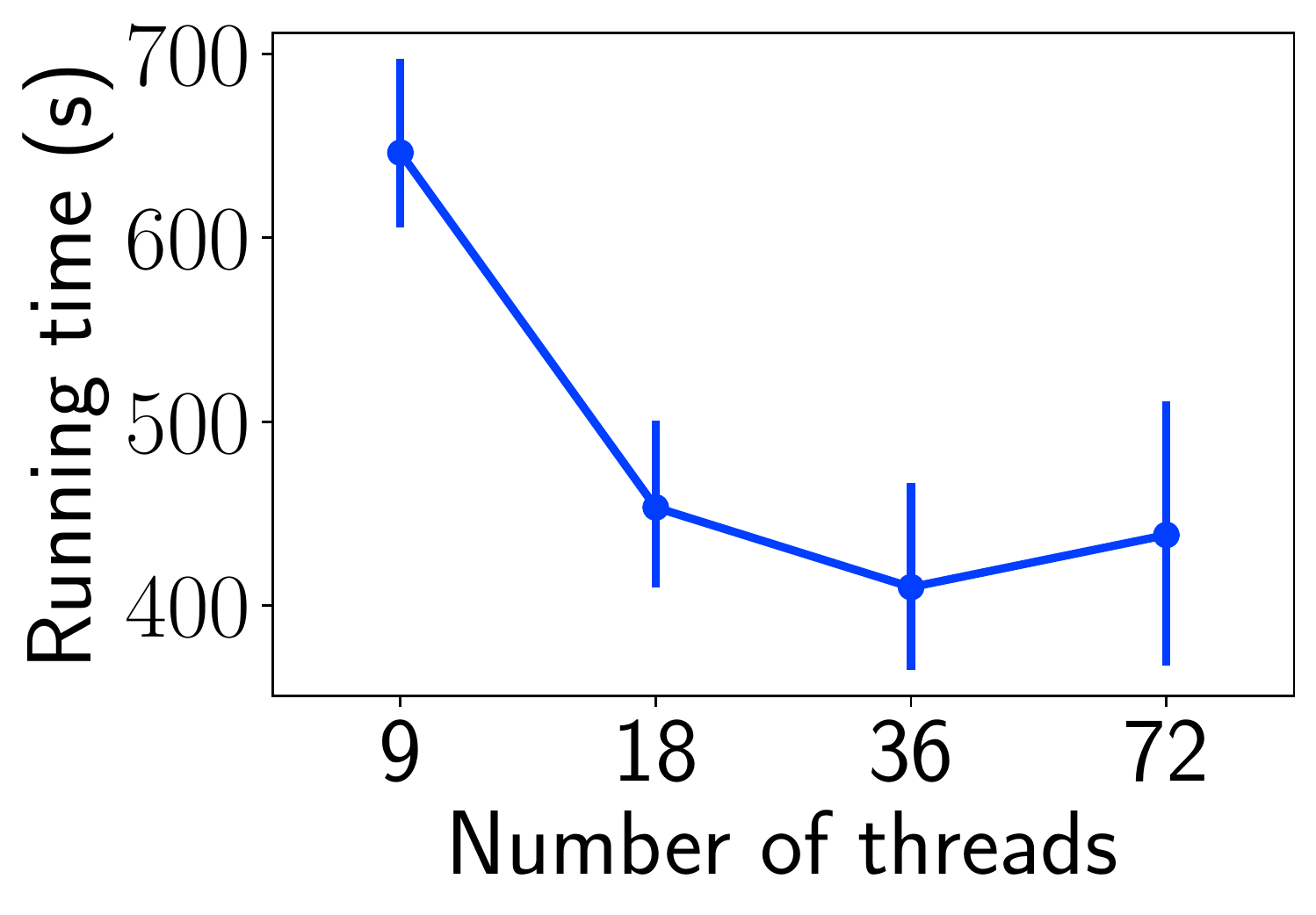}	
	\subcaption{Amazon}
\end{minipage}
\qquad
\begin{minipage}[t]{0.29\linewidth}
	\centering
	\includegraphics[width=\linewidth]{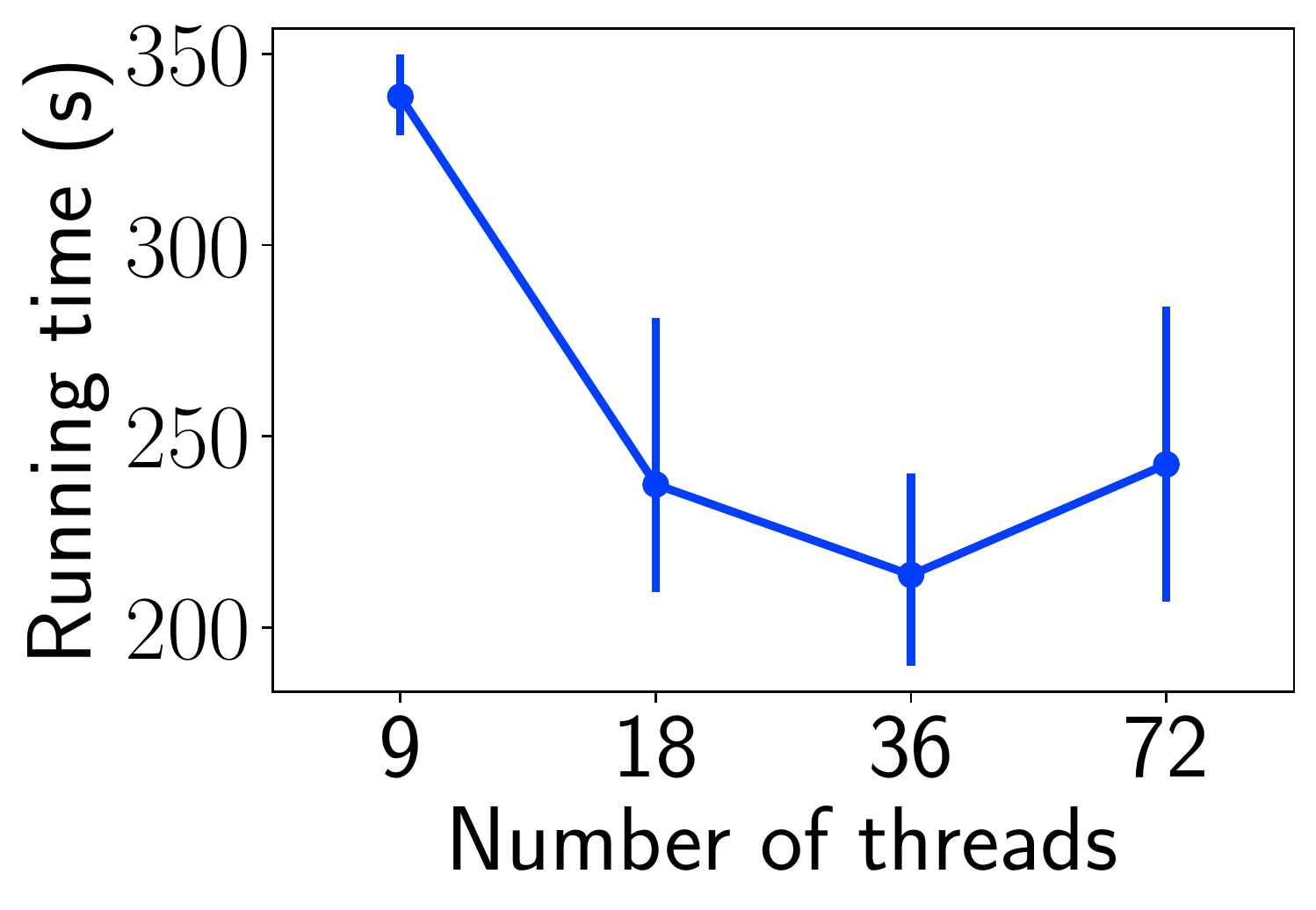}
	\subcaption{DBLP}
\end{minipage}
\qquad
\begin{minipage}[t]{0.29\linewidth}
	\centering
	\includegraphics[width=\linewidth]{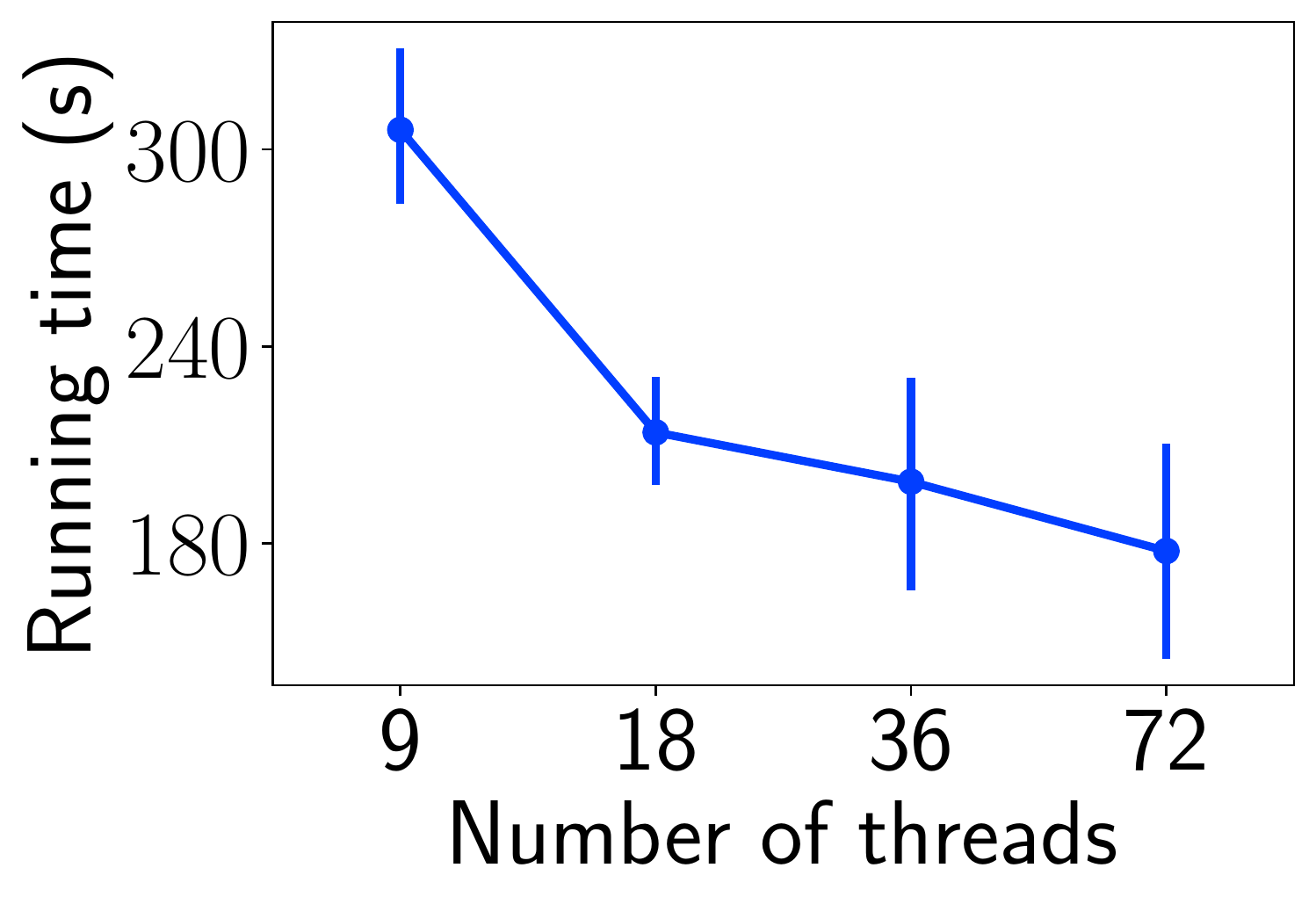}	
	\subcaption{Patents}
\end{minipage}
\qquad
\begin{minipage}[t]{0.29\linewidth}
	\centering
	\includegraphics[width=\linewidth]{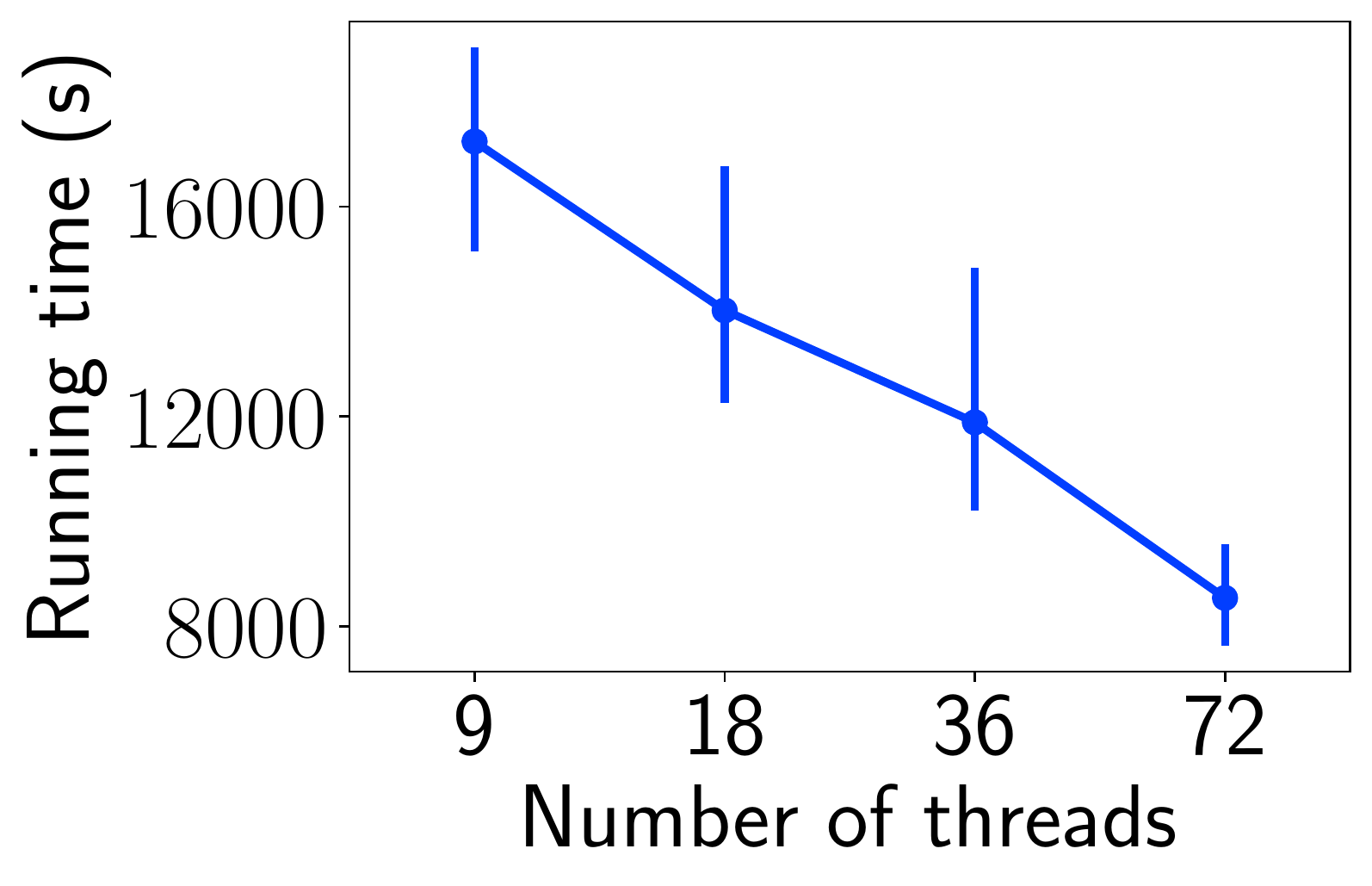}	
    \subcaption{Pokec}
\end{minipage}
\qquad
\begin{minipage}[t]{0.29\linewidth}
	\centering
	\includegraphics[width=\linewidth]{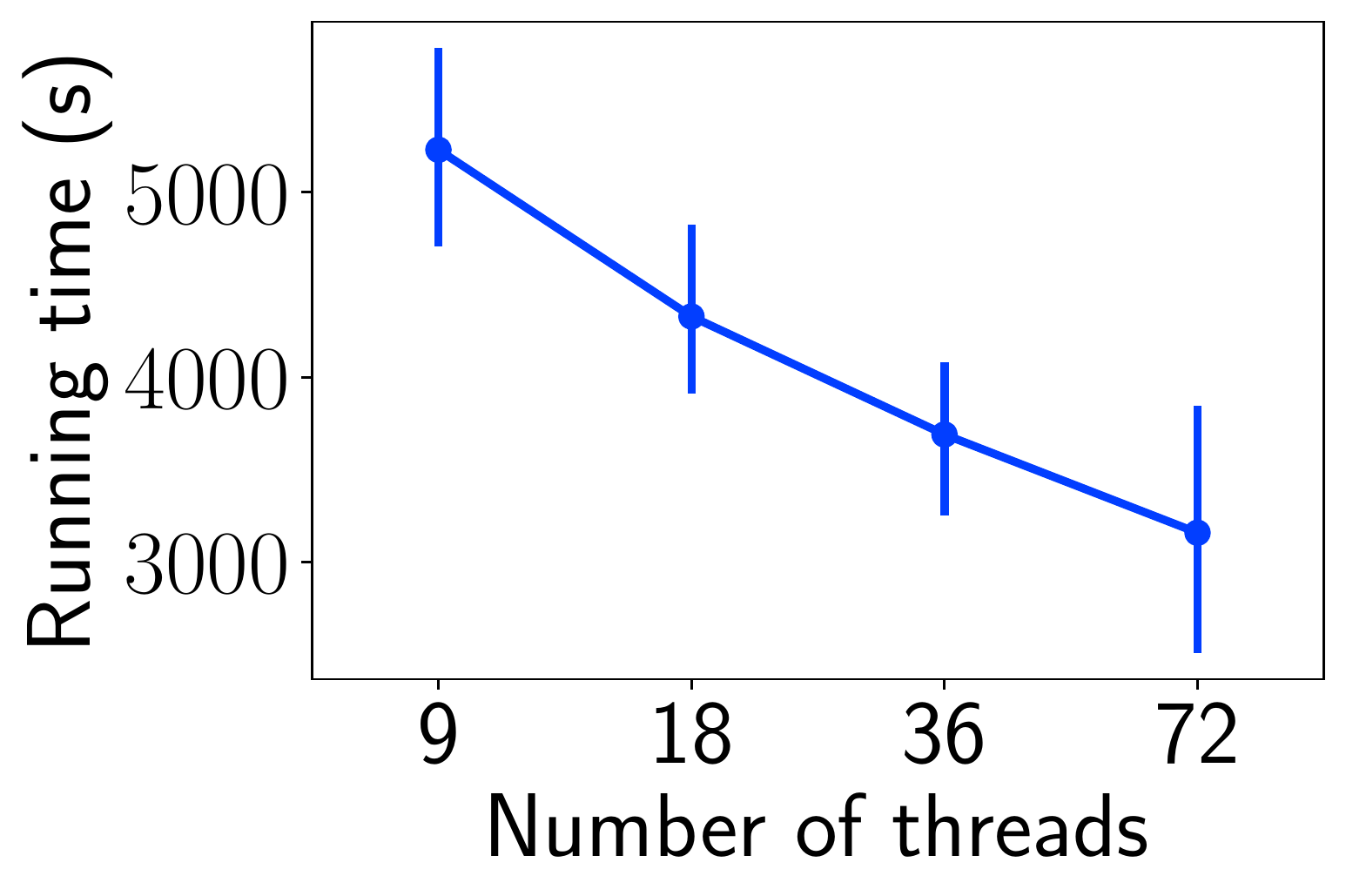}
	\subcaption{Live Journal}
\end{minipage}
\qquad
\begin{minipage}[t]{0.29\linewidth}
	\centering
	\includegraphics[width=\linewidth]{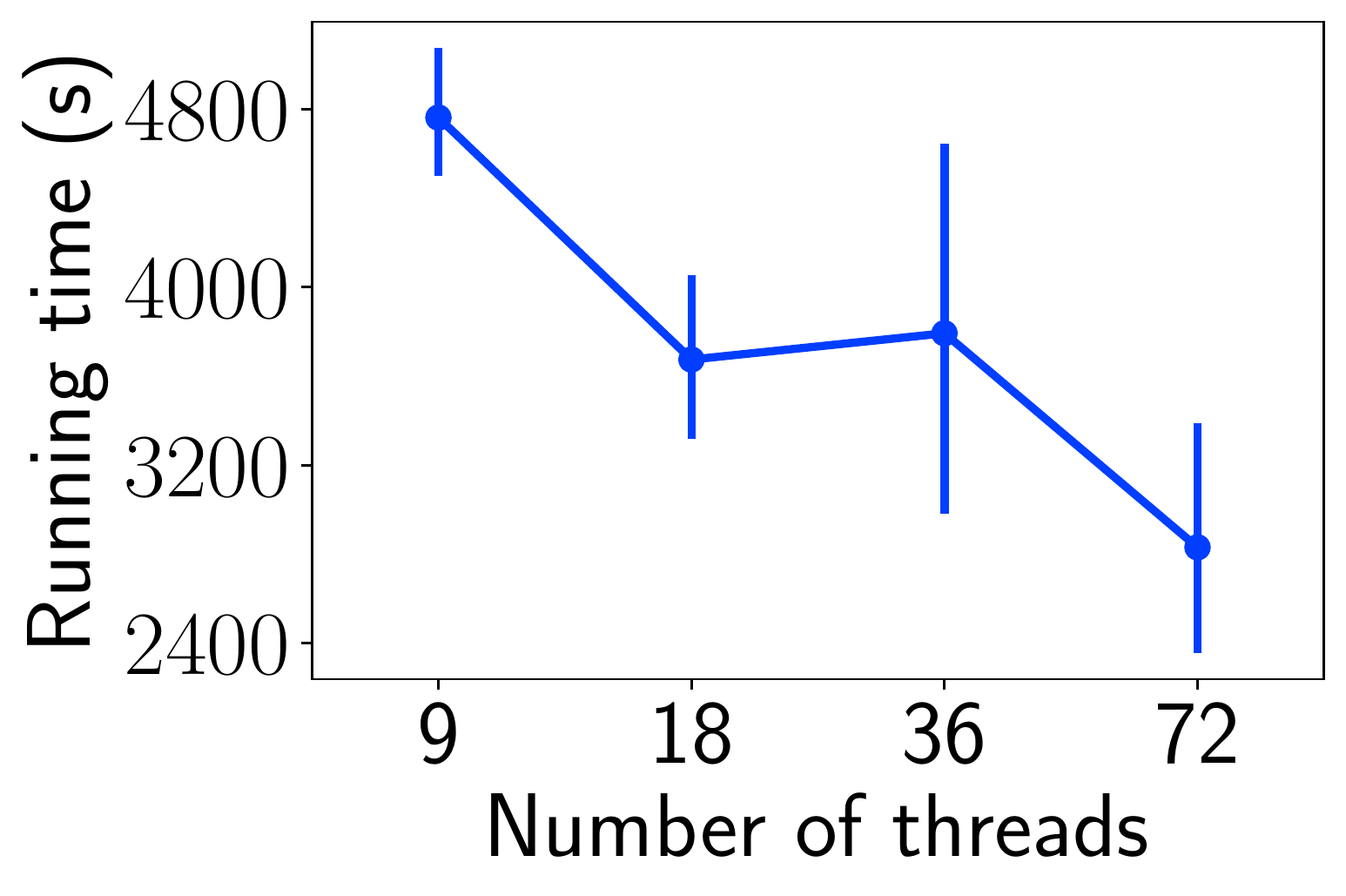}
	\subcaption{Orkut}
\end{minipage}
\caption{Scalability w.r.t.\ the Number of Threads for \CIS[5]. 
Despite a noticeable reduction in running time, the scalability is not linear. 
In fact, as we double the number of cores, the running time decreases by approximately a quarter instead of half, which may be due to the memory bandwidth limit coupled with the lack of memory locality, which is a ubiquitous problem in graph mining algorithms.}
\label{fig:scl-5}
\end{figure} 

\paragraph{Scalability on Number of Threads (\Cref{fig:scl-5}).}
In this experiment, because of \Cref{eq.auto.tours}, we set $\epsilon = 0.001$ to force a larger number of tours, thereby increasing the load per core and ensuring a sufficient workload. 
Further, we fix $k=5$, set $|\cI_1|=10^4$, $\rvrSize=10^7$ and compute running times over 10 executions while excluding the graph read time, which is not parallel.  
We observe that our implementation does not scale linearly: as we double the number of cores, the running time decreases by $\approx \nicefrac{1}{4}$ rather than $\nicefrac{1}{2}$. 
Local profiling using hardware performance counters (Linux's \textit{perf})
suggests that this overhead is an outcome of increased 
random-access patterns of in-memory graph data, which limits the overall use 
of the underlying processing pipeline.
Indeed, sub-optimal access patterns of graph data are a known issue that is currently
handled by dedicated accelerator hardware deploying optimized and specific 
caching mechanisms and memory access policies for workloads dominated by subgraph 
enumeration \citep{Accelerator}. 

\end{document}